\documentclass[12pt]{article}
\usepackage[dvipsnames,svgnames,x11names]{xcolor}
\usepackage{amsmath}
\usepackage{amsthm}
\usepackage{xspace}
\usepackage{float}
\usepackage[T1]{fontenc}
\usepackage{mdframed}
\usepackage{mathtools}
\usepackage{enumitem}
\usepackage[linesnumbered,ruled]{algorithm2e}
\usepackage{algpseudocode}
\usepackage{authblk}

\usepackage{fullpage}
\usepackage{thmtools}
\usepackage[colorlinks]{hyperref}
\hypersetup{colorlinks={true},linkcolor={blue},citecolor=magenta}

\usepackage{amssymb,amsfonts,amsmath,amsthm,amscd,dsfont,mathrsfs}
\usepackage{complexity}
\usepackage{tikz}
\usetikzlibrary{decorations.pathreplacing,backgrounds}
\usepackage{multirow}
\usepackage{mathpazo}
\usepackage{braket}
\usepackage{comment}
\usepackage{quantikz}

\usepackage[
backend=biber,
style=alphabetic,
sorting=ynt,
maxnames=99
]{biblatex}
\bibliography{ref}
\usepackage{url}
\usepackage{bbm}
\usepackage[capitalize]{cleveref}
\usepackage{hyperref} 
\usepackage[margin=1in]{geometry}
\hypersetup{breaklinks=true}

  \theoremstyle{plain}
  \newtheorem{theorem}{Theorem}[section]
  \newtheorem{lemma}[theorem]{Lemma}
  \newtheorem{corollary}[theorem]{Corollary}
  \newtheorem{definition}[theorem]{Definition}
  \newtheorem{remark}[theorem]{Remark}
  \newtheorem{claim}[theorem]{Claim}

  \newtheorem{world}{World}
 \newtheorem*{theorem*}{Theorem}
\newtheorem*{definition*}{Definition}

\newmdtheoremenv[backgroundcolor=gray!10,
                 linewidth=0pt,
                 innerleftmargin=4pt,
                 innerrightmargin=4pt,
                 innertopmargin=1pt,
                 innerbottommargin=4pt,
            splitbottomskip=4pt]{problem}[prob]{Problem}

\newmdtheoremenv[backgroundcolor=gray!10,
                 linewidth=0pt,
                 innerleftmargin=4pt,
                 innerrightmargin=4pt,
                 innertopmargin=1pt,
                 innerbottommargin=5.5pt,
            splitbottomskip=4pt]{conjecture}[conj]{Conjecture}

\newcommand{\abs}[1]{\left| #1 \right|}

\newcommand{\N}{\mathbb{N}}

\newcommand{\Z}{\mathbb{Z}}

\newcommand{\cO}{\mbox{\sffamily cO}}
\newcommand{\inD}{\Pi^{\in\mathsf{db}}}
\newcommand{\ninD}{\Pi^{\not\in\mathsf{db}}}
\newcommand{\Q}{\mathsf{Q}}
\newcommand{\CC}{\mathsf{C}}

\newcommand\numberthis{\addtocounter{equation}{1}\tag{\theequation}}

\renewcommand{\proj}[1]{\ensuremath{|#1\rangle \langle #1|}}

\newcommand{\norm}[1]{\left\Vert {#1} \right\Vert}

\renewcommand{\C}{\mathbb{C}}
\renewcommand{\E}{\mathbb{E}}

\newcommand{\bit}{\{0,1\}}

\newcommand{\algo}{\mathcal}

\newcommand{\negl}{\ensuremath{\operatorname{negl}}\xspace}

\newcommand{\Sp}{\mathsf{Sp}}
\newcommand{\MSp}{\mathsf{MSp}}

\newcommand{\outerprod}[2]{|#1\rangle\langle #2|}

\newcommand{\ig}{\mathsf{ig}}
\newcommand{\rg}{\mathsf{rg}}

\def\myvdots{\ \vdots\ }

\title{The Sponge is Quantum Indifferentiable}
\author[1,2]{Gorjan Alagic}
\author[1]{Joseph Carolan}
\author[3]{Christian Majenz}
\author[3]{Saliha Tokat}
\affil[1]{University of Maryland}
\affil[2]{National Institute of Standards and Technology}
\affil[3]{Technical University of Denmark}

\date{}

\begin{document}

\maketitle

\abstract{The sponge is a cryptographic construction that turns a public permutation into a hash function. When instantiated with the Keccak permutation, the sponge forms the NIST SHA-3 standard. SHA-3 is a core component of most post-quantum public-key cryptography schemes slated for worldwide adoption.

While one can consider many security properties for the sponge, the ultimate one is \emph{indifferentiability from a random oracle}, or simply \emph{indifferentiability}. The sponge was proved indifferentiable against classical adversaries by Bertoni et al. in 2008. Despite significant efforts in the years since, little is known about sponge security against quantum adversaries, even for simple properties like preimage or collision resistance beyond a single round. This is primarily due to the lack of a satisfactory quantum analog of the lazy sampling technique for permutations.

In this work, we develop a specialized technique that overcomes this barrier in the case of the sponge. We prove that the sponge is in fact indifferentiable from a random oracle against quantum adversaries. Our result establishes that the domain extension technique behind SHA-3 is secure in the post-quantum setting. Our indifferentiability bound for the sponge is a loose $O(\mathsf{\poly}(q) 2^{-\min(r, c)/4})$, but we also give bounds on preimage and collision resistance that are tighter.}

\newpage
\tableofcontents
\newpage

\section{Introduction}

\subsection{The sponge construction}

The sponge construction \cite{BDPvA07,KeccakSponge3} is a domain extension scheme that uses a public cryptographic permutation to construct a hash function or extendable-output function (XOF). The sponge is most prominently used in the SHA-3 standard~\cite{fips202}. The SHA-3 hash functions and XOFs, in turn, are core ingredients in all post-quantum cryptographic schemes recently standardized by NIST~\cite{fips203,fips204,fips205}. Ascon, a suite of lightweight symmetric-key schemes selected by NIST for standardization is also based on the sponge~\cite{sp800-232}.

To instantiate the sponge, one chooses a permutation $\varphi: \bit^n \rightarrow \bit^n$ as well as positive ``rate'' $r$ and ``capacity'' $c$ such that $r+c = n$. Given this data, the sponge $\Sp^{\varphi}$ is a XOF defined as follows. Parse the input $x = x_1 \| x_2 \| \cdots x_\ell$ into length-$r$ blocks. Initialize the computation with the state $0^r \| 0^c$. XOR $x_1$ into the rate, and then apply $\varphi$ to the entire state; repeat this for each block of the input, i.e., $\ell$-many times. This completes the \emph{absorbing phase}. Next, output the $r$ bits contained in the rate, and then apply $\varphi$ to the entire state; repeat this process until the desired number of output bits is produced. This completes the \emph{squeezing phase}. Turning this XOF into a hash function amounts to selecting a fixed output length. For example, an $r$-bit output on a $2r$-bit input is simply the first $r$ bits of $\varphi((x_2\|0^c) \oplus \varphi(x_1 \|0^c))$, as depicted in \Cref{fig:sponge-example}.

\begin{figure}[h]
    \centering
    \includegraphics[width=0.5\linewidth]{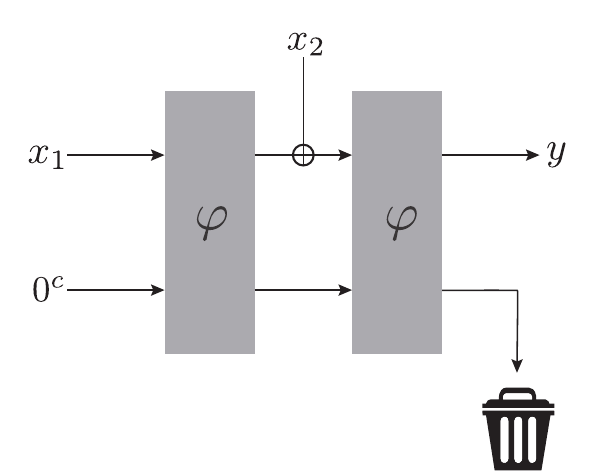}
    \caption{Sponge construction on input $x_1 \Vert x_2$, with output $y$.}
    \label{fig:sponge-example}
\end{figure}

\paragraph{Classical security.}

The security of the sponge construction is usually analyzed in the ideal permutation model, where the cryptographic permutation $\varphi$ is modeled as a uniformly random permutation given as an oracle. As with all hash functions, the minimum security requirements include one-wayness and collision resistance. The ``gold standard'' for domain extension scheme security is indifferentiability from a random oracle~\cite{MRH04}. Roughly speaking, indifferentiability in this context means that the pair (construction $\Sp^{\varphi}$, primitive $\varphi$) is  query-indistinguishable from the pair (truly random function $f$, simulated primitive $\algo S^f$). The sponge construction where $\varphi$ is a random permutation is indifferentiable from a random oracle, up to advantage $O(q^2/2^c)$ for a classical $q$-query adversary~\cite{BDPvA08}. This result is tight, and it implies tight collision resistance, as well as tight preimage resistance for some parameters. For the XOF setting, tight preimage resistance was proven separately \cite{LM22}.

The proofs for these classical results follow a general blueprint. In order to evaluate $\varphi$, the adversary has to query it. The simulator can examine the queries of the adversary and use its access to $f$ to try to respond in a consistent manner. When an adversary attempts to break a certain property (e.g., find a collision, or distinguish the real and ideal worlds of the indifferentiability experiment), the list of input-output pairs of the permutation that the adversary has learned is analyzed. It is then shown that a query list allowing the adversary to achieve its goal is unlikely.

\paragraph{Quantum security.}

Rapid technological advances have made it increasingly necessary to ensure that cryptographic mechanisms remain secure against adversaries equipped with a quantum computer. In the case of SHA-3, such adversaries are free to evaluate both the Keccak permutation and SHA-3 itself in coherent superposition. It is thus appropriate to analyze the sponge construction in the quantum-accessible ideal permutation model (QIPM). In this model, the adversary can make quantum queries to both the random permutation $\varphi$ and the sponge $\Sp^{\varphi}$. Proving non-trivial security results in the QIPM has turned out to be challenging. One obvious obstacle is that there is no quantum analog of the generic proof blueprint outlined above. In particular, quantum queries to a random permutation cannot be recorded into a list due to the fact that quantum information cannot be copied (\emph{no-cloning principle}). This obstruction can be circumvented, but only in the related setting of quantum-queryable \emph{random functions}, where Zhandry's compressed oracle technique \cite{Zhandry2018} has enabled generalizing a multitude of proof techniques to the post-quantum setting. Indifferentiability appears particularly challenging in the quantum-query context. For domain extension schemes, it is known to require stateful simulators \cite{ristenpart07careful}, and the only plausible quantum technique for that is the compressed oracle.

As a consequence of these challenges, prior results on the quantum security of the sponge either instantiate it with a random function instead of a random permutation (as in~\cite{CHS19,CGBHS18}), or limit the sponge to one round of absorption only
(as in~\cite{Zhandry21,carolan24oneway,cpz24precomp,MMW24}). There have also been attempts to generalize the compressed oracle technique to the random permutation setting. The fundamental idea of \emph{purifying} the random primitive is still sound \cite{ABKM22,ABKMS24,MMW24}. The statistical dependence of the outputs of the random permutation, however, complicate the task of \emph{oracle forensics}, the task of analyzing an adversary's behavior by inspecting the internal state of the oracle simulation. As a result, attempts at constructing a compressed permutation oracle that offers this feature have been unsuccessful \cite{Czajkowski2019,Unruh2021}\footnote{The difficulty of this problem has inspired very nice works that present approaches and characterizing results \cite{rosmanis2022tight,Unruh2023}} or are currently limited to the analysis of simple problems \cite{MMW24}.

The state of affairs is thus still very unsatisfactory: even basic security properties like quantum preimage or collision resistance, let alone indifferentiability, are not known to hold for the sponge, and seem out of range for existing techniques. This is a concerning gap, as SHA-3 is in widespread use and cryptographically-relevant quantum computers might appear in the not-too-distant future.

\subsection{Our approach}
\label{subsec:approach}

In this work, we establish the quantum security properties of the sponge discussed above. Rather than  constructing a fully general permutation analogue of the compressed oracle, we develop a tailored approach for our important special case. 

We begin by observing that there are subgroups of the full permutation group where random elements have nice compressed oracles. For example, one round of the Feistel construction~\cite{LR88} is a (highly structured) permutation, and the set of these forms a subgroup (with composition corresponding to XORing the underlying round functions). A random element from this set is constructed from a random function, and is thus amenable to the compressed oracle technique\footnote{Feistel itself is unfortunately still resistant to quantum security analysis, partly because there's no compressed oracle for the overall permutation.}.  For us it is natural to consider an unbalanced Feistel setting, split between the $c$ bits of the capacity and the $r$ bits of the rate. We express the sponge permutation $\varphi$ as  
\begin{equation}\label{eq:intro-permutation}
    \varphi = \omega_h \circ \tau_{k'} \circ \pi \circ \sigma_k\,,
\end{equation}
where $\sigma_k(x\|y) = x\|y \oplus k(x), \tau_{k'}(x\| y) = x \| y \oplus k'(x)$ and $\omega_h(x\|y) = x \oplus h(y)\|y$, where $k, k', h$ are random functions, and $\pi$ is a random permutation. In other words, first apply a one-round Feistel cipher, then apply a random permutation $\pi$ on $\bit^{r+c}$, and end with a two-round Feistel cipher. We are able to show that the three Feistel rounds combined with some global statistical properties of $\pi$ are sufficient to account for the intractability of $\varphi$ necessary for sponge security. Our focus here is on applying this idea to the sponge, but this technique may be useful for proving quantum security of other constructions based on ciphers and permutations, such as the Luby-Rackoff or Feistel construction itself. 

The reason this decomposition is useful is that we have ``offloaded'' the hardness of $\varphi$ onto the $h, k, k'$ functions, which we have quantum tools to analyze. The sponge becomes insecure when an adversary can find two inputs which lead to the same state after absorbing, and it turns out that finding two such inputs will require many queries to the random functions $h, k, k'$. This holds even in the case where an adversary gets to see the whole truth table of $\pi$.

When directly proving query bounds, such as for finding a sponge collision, we can now relax the problem: we provide oracle access to $h, k$ and $k'$, and hand the query algorithm the entire truth-table description of $\pi$. We can then formulate the problem as a search problem relative to a quantum-accessible random oracle which is entirely amenable to the compressed oracle technique. We apply that technique using the query bound framework of~\cite{chung2021compressed}.

For indifferentiability, we begin by showing that constructing $\varphi$ as in \Cref{eq:intro-permutation} is perfectly indifferentiable from sampling it uniformly at random. We then make a simplifying (but also perfectly indifferentiable) modification to the absorbing phase of the sponge: in each round, instead of XORing in the next message block, we simply replace the contents of the rate by that block. With these preparation steps out of the way, our high-level approach can then be seen as analogous to that of the quantum indifferentiability proof of the Merkle-Damg\r ard construction~\cite{Zhandry2018}. 

However, analyzing our construction proves more complicated. One reason is that, in Merkle-Damg\r ard, one can assume the round function has no state collisions. For the sponge, it is straightforward to find such collisions using inverse permutation calls; however, it is hard to find such collisions where both preimages are ``rooted'' at the initial state. In this way, one has to consider global properties that depend on many rounds, whereas single-round properties were sufficient for Merkle-Damg\r ard. Along the way, we uncover a gap in the Merkle-Damg\r ard proof. Our proof for the sponge avoids this problem, and we expect similar ideas to apply to Merkle-Damg\r ard, though they likely result in a looser bound.

\subsection{Summary of results}
\label{subsec:summary}

Recall that, given a permutation $\varphi: \bit^{r+c} \rightarrow \bit^{r+c}$, the sponge construction using $\varphi$ is denoted by $\Sp^{\varphi}$. 

\paragraph{Preimage and collision resistance.} For the case of a single round of squeezing (but arbitrary rounds of absorbing), we prove the following (non-tight) bounds for preimage and collision resistance. See \cref{sec:bounds-summary} for a technical overview.

\begin{theorem}[Informal summary of \cref{thm:quantum-collision,cor:quantum-preimage}]\label{thm:intro-CR}
The probability that a quantum algorithm $\mathcal{A}$ making no more than $q$ quantum queries to a random permutation $\varphi \in S_{\{0,1\}^{r+c}}$ and its inverse $\varphi^{-1}$ outputs a sponge collision is upper bounded as
\begin{align*}
    \underset{\varphi \leftarrow S_{\{0,1\}^n};\,m,m'\leftarrow \mathcal{A}}{\Pr}[m\neq m'\wedge \Sp^\varphi(m)=\Sp^\varphi(m')]\leq O(q^5n2^{-\min(r,c)}).
\end{align*}
The same bound also holds for preimage-resistance. For constant success probability, $\tilde\Omega\bigl(2^{\min(r,c)/5}\bigr)$ quantum queries are thus necessary to find a collision or a preimage.
\end{theorem} 

In the preimage resistance notion considered above, the adversary is given a uniformly random $y \in \bit^r$. This should be distinguished from one-wayness, where $y$ is instead be the image of a uniformly random input.

\paragraph{Indifferentiability.}

The following formalizes our main result, i.e., that the sponge is indifferentiable against adversaries making quantum queries to $\varphi,  \varphi^{-1}$ and $\Sp^{\varphi}$. Here $\Sp^{\varphi}$ denotes the full sponge, i.e., both the input and the output are of arbitrary length. For a given quantum query to $\Sp^{\varphi}$, the length of that query is simply the length of the longest message in the superposition. 

\begin{theorem}[Informal summary of \Cref{thm:main}]
    \label{thm:intro-indiff}
    There exists an efficient simulator $\algo S$ such that for all adversaries $\algo A$ making $q$ quantum queries each of length at most $r \cdot \ell$, 
    \begin{align*}
        \abs{\Pr[\algo A^{\varphi, \varphi^{-1}, \Sp^{\varphi}}() = 1] - \Pr[\algo A^{\algo S^f, f}() = 1]} = O\left(\ell^3 \sqrt[4]{q^9 2^{-\min(r, c)}}\right).
    \end{align*}
    In the above, $\varphi$ is a uniformly random permutation while $f$ is a uniformly random function with the same domain and range as $\Sp^{\varphi}$.
\end{theorem}

\subsection{Additional related work}
The only previous quantum indifferentiability result with a stateful simulator is for the Merkle-Damg\r ard construction \cite{Zhandry2018}. Prior to this, some barrier results for quantum indifferentiability were shown~\cite{CETU18}. Our approach to the direct query bounds for collision and preimage finding uses the query bound framework from \cite{chung2021compressed} and draws inspiration from the techniques built on top of it \cite{DFMS22,AMHJRY23,HJMN24,FAEST}. Some query bounds have been proven using the compressed oracle without additional generic tools \cite{liu19multicols,Hamoudi_2023}, and an alternative generic framework via oracle games presented in \cite{CMS19} is similarly versatile (see, e.g., \cite{LR24}).

\subsection{Acknowledgements}

Joseph Carolan thanks Alexander Poremba and Mark Zhandry for many helpful discussions. Gorjan Alagic was supported in part by NSF award CNS-21547. Joseph Carolan acknowledges support from the U.S. Department of Energy grant DE-SC0020264.

\newpage
\section{Technical summary}

\subsection{A View of permutation oracles}\label{sec:oracle-summary}

    We avoid the aforementioned barriers to quantum lazy-sampling permutations by developing an alternative view of permutation oracles that enables us to apply existing techniques. We now describe this approach in somewhat more detail. The full construction with proofs is given in \cref{sec:oracle-proofs}.
    
\paragraph{Stateful oracles inspired by Feistel.}
    Recall that we may write a uniform random permutation $\varphi$ as the product $\varphi = \pi \circ \sigma$ of a uniform random permutation $\pi$ with a potentially structured permutation $\sigma$. We observe that, for certain distributions on structured permutations $\sigma$, we can use existing compressed oracle techniques. For example, suppose we define our distribution by sampling a random function $k$, and defining 
    \begin{align}
        \sigma(x \Vert y) = x \Vert y \oplus k(x).
    \end{align}
    Clearly, answering a query to $\sigma$ can be done with one query to $k$.
    We also have $\sigma = \sigma^{-1}$, evading the problem of inverse queries. To build intuition for why this is useful, consider a toy problem introduced by Unruh~\cite{Unruh2021,Unruh2023}.

    \begin{problem}[Double-Sided Zero Search]
        Given query access to a permutation $\varphi$ on $\bit^{2r}$ and its inverse, find a ``zero pair'' $(x, y) \in \bit^r$ s.t. $\varphi(x \Vert 0^r) = y \Vert 0^r$.
    \end{problem}

    Lower bounds for this problem are known~\cite{carolan24oneway,MMW24}. Still, thinking about how to prove such bounds provides a simple jumping-off point for our technique. We will select independent random functions $k, k' : \bit^r \rightarrow \bit^r$, and define permutations $\tau$ and $\sigma$ as 
    \begin{align}
        \sigma(x \Vert y) =& x \Vert y \oplus k(x), &&
        \tau(x \Vert y) = x \Vert y \oplus k'(x).
    \end{align}
    Consider randomly selecting a $\pi$, and defining $\varphi = \tau \circ \pi \circ \sigma$. It follows that $\varphi$ is random, so it suffices to prove our lower bound against an adversary querying $\varphi$ and $\varphi^{-1}$. We can also strengthen the access model, so the adversary can query $k$ and $k'$ (which suffices to implement $\sigma$ and $\tau$), and also receives the \emph{entire truth table} for $\pi$. It turns out that, for almost all $\pi$, such an adversary must still make $\tilde \Omega(2^{r/2})$ quantum queries to $k$ and $k'$ to find a zero pair in $\varphi$. 
    
    To see why such a bound is likely to hold, consider how lazy-sampling classical queries to $k,k'$ might fail. When querying $\sigma$ on input $x \Vert 0^r$, a randomly selected suffix $w \in \bit^r$ is selected. This gives a candidate zero pair $y \Vert z = \pi(x \Vert w)$. In the case where $y$ is not yet queried to $k'$, this will be completed to a zero pair only when $k'(y)=z$, i.e., with probability $2^{-r}$. In the case where $y$ is already queried to $k'$ with image $z'=k'(y)$, observe that for most permutations $\pi$ the value of $z$ will have nearly $r$ bits of entropy from the entropy in $w$, making the event $z=z'$ have probability $\approx 2^{-r}$ (one can make this precise with a tail bound for $\pi$). To handle the case of quantum queries, one can use compressed oracles for $k, k'$ in place of lazy sampling and obtain the appropriate quadratically weaker bound.

    This simple example conveys the key observation underlying our approach. By writing a random permutation as a product involving a random permutation and well-chosen structured permutations, we can reduce (certain) problems about two-way-accessible random permutations to problems about random functions. 
    
\paragraph{The case of the sponge.}
    The example of Double-Sided Zero Search is somewhat analogous to the single round sponge with $c=r$. To analyze the many-round sponge, it turns out that we will need one more factor in our decomposition of the sponge permutation $\varphi$. Recall that $\varphi : \bit^n \rightarrow \bit^n$ where $n=r+c$, and that this determines the sponge construction $\Sp^{\varphi}$. We will construct $\varphi$ out of a random permutation $\pi : \bit^n \rightarrow \bit^n$ and permutations corresponding to functions $k,k' : \bit^r \rightarrow \bit^c$ and $h : \bit^c \rightarrow \bit^r$. We define the structured permutations by \begin{equation}
        \sigma_k(x \Vert y) = x \Vert y \oplus k(x), 
        \qquad \tau_{k'}(x \Vert y) = x \Vert y \oplus k'(x), 
        \qquad \omega_h(x \Vert y) = x\oplus h(y) \Vert y.
    \end{equation}
    Note that the application of $h$ in $\omega_h$ is ``upside down'' w.r.t. the applications of $k, k'$ in $\sigma_k, \tau_{k'}$, as in the Feistel or Luby-Rackoff construction. We then write 
    \begin{equation}
        \varphi \coloneqq \omega_h \circ \tau_{k'} \circ \pi \circ \sigma_k.
    \end{equation} 
    We require a particular property from $\pi$: for any fixed $x_1, x_2 \in \bit^r$, the number of suffixes $z \in \bit^c$ such that $\pi(x_1 \Vert z)$ begins with $x_2$ is not much larger than its average value of $2^{c-r}$ (or not larger than $O(n)$ if $r \geq c$). A standard argument shows that almost all permutations satisfy this property, described in~\Cref{sec:perm-tail-bounds}.

    With this in place, we will henceforth discuss the sponge construction in terms of a fixed permutation $\pi$ together with random functions $h, k, k'$ that are implemented using compressed oracles (or lazy sampling in the classical case). We can then talk about query transcripts or databases for these functions, which we will refer to as $D_h, D_k, D_{k'}$. Roughly speaking, these databases contain the points that the adversary knows, and do not contain unqueried points. 
    Central to our argument is the notion of a \emph{tail} for a given capacity value $z$; this is essentially an input which reaches state $z$.

    \begin{definition*}[Tail]\label{def:tail-intro}
        Let $z \in \bit^c$, and fix $D_k, D_{k'}$ and $\pi$. We say that $z$ has a ``tail'' under the following recursive conditions. 
        \begin{enumerate}[label=(\arabic*)]
            \item $z=0^c$ has a tail.
            \item $z$ has a tail if it can be reached as an internal state by a single round of absorption from a state $z_p$ which has a tail. That is, there exist $x_p \in D_k$, $x_i \in D_{k'}$, a $z_p \in \bit^c$ with a tail, and a $z_i \in \bit^c$ such that 
            \begin{align*}
                x_i \Vert z_i =& \pi(x_p \Vert z_p \oplus k(x_p)) \\
                z =& z_i \oplus k'(x_i),
            \end{align*}
        \end{enumerate}
        A tail of $z$ is a string of inputs required to reach $z$ according to the above conditions. Specifically, in case (1) the empty string is the unique tail of $z = 0^c$, and in case (2), any tail of $z_p$ concatenated with $x_p$ is a tail of $z$. We denote by $\mathsf{tail}(z)$ the set of tails of $z$.
    \end{definition*}

    A key property is that the tail does \emph{not} depend on $D_h$. This is because $h$ affects the top wire of the sponge immediately before another input is absorbed, and the input may be arbitrarily chosen. The adversary has full control over this wire during the absorption phase, and can simply undo the application of $h$ if it wishes. In this sense, the function $h$ does not prevent the adversary from reaching any given capacity state of the sponge---but, as it turns out, $h$ will be critical when we discuss obtaining final outputs from the sponge. On the other hand, the functions $k$ and $k'$ serve to ensure that the adversary can reach only a small number of states with a small number of queries, and furthermore that this set of reachable states is sufficiently random. This is necessary because an adversary that knows two tails that reach the same state can easily construct many collisions in the sponge, breaking security.

    Towards making this precise, we first define the notion of an intermediate pair. This will help us to treat queries to $k$ and $k'$ somewhat symmetrically.

    \begin{definition*}[Intermediate pair]
    \label{def:IP-intro}
        Fix a database pair $(D_k, D_{k'})$. We say a pair $(x, z)$ is an intermediate pair if there exists $x_p \in D_k$ and a $z_p \in \bit^c$ with a tail such that 
        \begin{align*}
            x \Vert z = \pi(x_p \Vert z_p \oplus k(x_p))\,.
        \end{align*}
        We let $\mathsf{IP}(D_k, D_{k'})$ denote the set of all intermediate pairs.
    \end{definition*}

    Note that because $\pi$ is a permutation, there is no multiplicity for intermediate pairs.  We are now ready to introduce the concept of a good database.
    \begin{definition*}[Good]
    \label{def:good-intro}
        A database pair $(D_k, D_{k'})$ is \textbf{good} if every $z$ has at most one tail, and $\mathsf{IP}(D_k, D_{k'})$ has unique $x$-values. In other words, we require both of the following conditions:
        \begin{align}
            \forall z \in \bit^c&, \,|\mathsf{tail}(z)| \leq 1 \\
            \forall (x_1, z_1), (x_2, z_2) \in \mathsf{IP}(D_k, D_{k'})&, \,(x_1, z_1) \neq (x_2, z_2) \Leftrightarrow x_1 \neq x_2
        \end{align}
    \end{definition*}

    We show in \Cref{sec:oracle-proofs} that any randomly selected output of $k$ or $k'$ is exponentially unlikely to create either an intermediate pair prefix collisions or a state collision. The intuition is that these are both collision type events, and so the number of possible ``bad outputs'' is bounded by $\mathsf{poly}(t)$ after $t$ queries. Using the framework of compressed oracles \cite{Zhandry2018} and quantum transition capacities \cite{chung2021compressed}, we can use this to show that the quantum databases for $D_k$ and $D_{k'}$ remain almost entirely in the good subspace.

\subsection{Query lower bounds}\label{sec:bounds-summary}

We now briefly describe our approach to proving quantum-query bounds for finding collisions or preimages. The details are given in \cref{sec:bounds-proofs}. The first step to establishing these bounds is to devise a simple reduction from a well-chosen search task. For the preimage finding case, for example, this task is roughly as follows: given a $y \in \bit^r$, output a ``path'' through the sponge that results in output $y$. Our reduction works for both classical and quantum oracle access, and is based on the indifferentiability of a random permutation oracle from the stateful oracle described in \cref{sec:oracle-summary}. 
We also derive the probability of an adversary with classical oracle access having a colliding input-output pair in its database.
Taken together, these results are sufficient for deriving classical-query lower bounds. 

As it turns out, the tools discussed above can also be used to establish quantum-query lower bounds. This can be done using the framework developed by~\cite{chung2021compressed}. In this framework, the square root of the probability that a database satisfies a predicate $\P$ after $q$ queries, denoted $[\![ \emptyset \xrightarrow{q} \P]\!]$, can be bounded using 
\begin{align*}
    [\![ \emptyset \xrightarrow{q} \P]\!]
    &\leq \sum_{t=1}^q  [\![ \neg \P_{t-1} \xrightarrow{} \P_t]\!]\,.
\end{align*} 
Here $[\![ \neg \P_{t-1} \xrightarrow{} \P_t]\!]$ is the maximum probability of transitioning, during the $t$-th query, from not satisfying predicate $\P$ to satisfying it. Using the framework of~\cite{chung2021compressed}, a bound on the probability of this transition in the classical setting also yields a bound in the quantum case. At a high level, this allows us to translate the classical tools described above into quantum-query lower bounds. The actual proof is more involved; for instance, it requires avoiding certain ``bad'' databases during the simulation. We do this by separating good and bad database at the beginning, and bounding the bad database case separately by treating them as a predicate.

\begin{theorem}[quantum collision resistance]
The probability that a quantum algorithm $\mathcal{A}$ with quantum query access to a random permutation $\varphi \in S_{\{0,1\}^{r+c}}$ and its inverse, making at most a total of $q$ queries, returns $m,m'\in\left(\{0,1\}^r\right)^{\le l}$ for $l\le q$ such that $m\neq m'$ and  $\Sp^\varphi(m)=\Sp^\varphi(m')$, can be upper bounded as
\begin{align*}
\Pr_{\substack{\varphi \sim S_{\{0,1\}^{r+c}} \\ m,m'\leftarrow \mathcal{A}^{\varphi,\varphi^{-1}}}}
    [\Sp^\varphi(m)=\Sp^\varphi(m')]\leq O\left(q^5n2^{-\min(r,c)}\right).
\end{align*}
\end{theorem}

The derivation for the bound of preimage finding is very similar to the collision finding case. The bound is dominated by terms contributed as a result of the bad database predicate.

\begin{theorem}[quantum preimage resistance]
Given $y\sim\bit^r$, the probability that a quantum algorithm $\mathcal{A}$ with quantum query access to a random permutation $\varphi \in S_{\{0,1\}^{r+c}}$ and its inverse, making at most a total of $q$ queries, returns $m\in\left(\{0,1\}^r\right)^{\leq l}$ for $l\leq q$ 
such that $y=\Sp^\varphi(m)$, can be upper bounded as
\begin{align*}
    \Pr_{\substack{y \sim \bit^r \\ m'\leftarrow \mathcal{A}^{\varphi,\varphi^{-1}}}}[y=\Sp^\varphi(m)]\leq O\left(q^5n2^{-\min(r,c)}\right).
\end{align*}
\end{theorem}

These results are not tight. A sponge collision can be found in $O\bigl(2^{\min(r,c)/3}\bigr)$ quantum queries using, e.g., the BHT algorithm for collision finding in case $r<c$, and the BHT algorithm for claw finding in case $c\le r$ (``meet in the middle''). A preimage can be found using $O\bigl(\min\bigl(2^{c/3},2^{r/ 2}\bigr)\bigr)$ queries using Grover or claw-finding BHT.

\subsection{Sponge indifferentiability}\label{sec:indiff-summary}
    
Showing indifferentiability amounts to constructing a simulator $\algo S$ which answers permutation queries in a way consistent with the ideal (i.e., random-oracle) functionality $f$ corresponding to the sponge. We show the following bound.

\begin{theorem}
    There exists an efficient simulator $\algo S$ for the sponge construction such that all adversaries $\algo A$ making $q$ queries of block length at most $l$ have distinguishing advantage 
    \begin{align*}
        \norm{\Pr[\algo A^{\varphi, \varphi^{-1}, \Sp^{\varphi}}() = 1 - \algo A^{\algo S^f, f}() = 1]} = O\left(l^2 \sqrt{q^9 2^{-\min(r, c)}} + l^3 \sqrt[4]{q^5 2^{-\min(r, c)}}\right).
    \end{align*}
    \label{thm:main-intro}
\end{theorem}

We now briefly discuss our approach for constructing and analyzing $\algo S$. For this discussion only, we restrict to a single round of squeezing. Our simulator is somewhat similar to Zhandry's quantum simulator for Merkle-Damg\r{a}rd \cite{Zhandry2018}, in that it will analyze the stored databases to determine whether the adversary is computing on an input to the sponge (in which case it answers using the ideal functionality), or on an arbitrary ``disconnected'' input (in which case it answers using a compressed oracle). 

The simulator $\algo S^f$ will simulate $\pi$ with a pseudorandom permutation on $\bit^{r+c}$ and use it only as a black box. It must also provide oracles $h, k, k'$ which are, together with $\pi$, consistent with the ideal functionality $f$, so that it appears to the adversary as if $f$ is the sponge built from $h, k, k'$ and $\pi$. For $k$ and $k'$, the simulator simply answers queries using a compressed oracle database. By our bounds on the probability of a bad event, these databases will remain almost entirely on the subspace where each state value $z$ has a unique tail. 

Our procedure for answering queries to $h$ is analyzed with a slightly modified sponge construction, the ``Msponge''. Whenever the sponge XORs an input block into the top wire, the Msponge will instead replace the content of the top wire with that input block. 
In \Cref{sec:oracle-proofs}, we show that indifferentiability of the Msponge implies indifferentiability of the sponge. The intuition is that one can compute the output of an $l$-block input to the standard sponge using $l$ queries to the Msponge by simply querying on each prefix, and then adapting the following block by XORing in the output before calling on the next prefix. This reduction runs both ways, so we simply work with the Msponge.

Conceptually, this makes it clear that only the final $h$ query is relevant, as all others are applied to a wire that is then immediately discarded. We show that the databases for $k$ and $k'$ will contain a complete record of the absorbed input whenever the adversary computes the sponge. Then, on the final $h$ query to $z$, the simulator can reconstruct $\mathsf{tail}(z)$, which is precisely the input used to reach said $z$ value. At this stage, the simulator knows to answer using the ideal functionality $f$ instead. Of course, this analysis only applies on good databases.

While our analysis is conceptually similar to Zhandry's Merkle-Damg\r{a}rd indifferentiability proof \cite{Zhandry2018}, the technical details are somewhat more involved due to the complexity of our construction. Further, we address the gap in the original work mentioned at the end of \Cref{subsec:approach}. This gap stems from the fact that projections onto ``valid'' databases (which always contain an output after decompressing) and ``good'' databases (which project out certain unwanted databases) do not commute. We elaborate in \Cref{subsec:MD-gap}, but it is worth noting that repairing this gap is one of the main reasons our indifferentiability bound is looser than our query bounds\footnote{Essentially, because we cannot project onto both good and valid simultaneously, we pay for the component on each subspace on each query.}. We expect that a similar idea could be applied to Merkle-Damg\r{a}rd as well, and would result in a similarly looser bound.

\newpage

\section{Preliminaries}
\label{sec:prelim}

\subsection{Quantum}

We consider quantum states as unit vectors of a Hilbert space $\mathbb C^D$. We define the Euclidean norm of a vector in such a Hilbert space as $\norm{\ket{\psi}}^2 = \braket{\psi | \psi}$.

\begin{definition}
    The Euclidean distance between quantum states $\ket{\psi},\ket{\phi} \in \mathbb C^D$ is \begin{align*}
        d(\ket{\psi}, \ket{\phi}) := \min_{\theta} \norm{\ket{\psi} - e^{i\theta} \ket{\phi}} = \sqrt{2(1-|\braket{\psi | \phi}|)}.
    \end{align*}
    \label{defn:state-distance}
\end{definition}

We will need the following standard result that close quantum states cannot be distingushed.

\begin{claim}
    Let $\ket{\psi},\ket{\phi} \in \mathbb C^D$ be quantum states that satisfy \begin{align*}
        d(\ket{\psi}, \ket{\phi}) = \epsilon.
    \end{align*}
    Then, no measurement can distinguish (a single copy of) $\ket{\psi}$ from $\ket{\phi}$ with advantage exceeding $\epsilon$.
    \label{claim:states-indis}
\end{claim}

An operator $O : \mathbb C^{D_1} \rightarrow \mathbb C^{D_2}$ is a linear map between Hilbert spaces. We say that $O$ is an isometry if inner products are preserved up to global phase. We say that $O$ is a unitary if it is an isometry, and $D_1 = D_2$. We can define both a norm and a Euclidean distance measure on operators.

\begin{definition}
    The norm of an operator $O : \mathbb C^{D_1} \rightarrow \mathbb C^{D_2}$ is \begin{align*}
        \norm{O} := \max_{\ket{\psi}} \frac{\norm{O\ket{\psi}}}{\norm{\ket{\psi}}}.
    \end{align*}
    \label{defn:operator-norm}
\end{definition}

\begin{definition}
    The distance of two operators $O, O' : \mathbb C^{D_1} \rightarrow \mathbb C^{D_2}$ is \begin{align*}
        d(O, O') := \max_{\ket{\psi}, \norm{\ket{\psi}} = 1} d(O \ket{\psi}, O' \ket{\psi}). 
    \end{align*}
    \label{defn:operator-distance}
\end{definition}

Note that $d(O, O') \leq \norm{O - O'}$, and similarly for states. We will occasionally say that an isometry ``appends a blank register'' or a similar statement, for instance an isometry $A$ which acts as \begin{align*}
    A_X \ket{x}_X = \ket{x}_X \ket{0}_Y.
\end{align*}
We use subscripts to denote the (domain) Hillbert space on which an operator acts. We may define the commutator on operators whose domain and range are equal in the usual way, though by an abuse of notation we also sometimes define a commutator between operators where one has an expanded range. This is defined as follows.

\begin{definition}
    Let $V : \algo H_A \rightarrow \algo H_{AB}$ and $U : \algo H_A \rightarrow \algo H_A$ be operators. We define their commutator as an operator from $\algo H_A \rightarrow \algo H_{AB}$ as \begin{align*}
        [U_A, V_A] = U_A V_A - V_A U_A,
    \end{align*}
    where the operator $U_A V_A$ has $U$ acting on the subspace $\algo H_A$ of $\algo H_{AB}$.
    \label{def:gen-comm}
\end{definition}
We will generally use consistent labeling within sections for subsystems, which allows one to track the range of an operator with distinct domain and range.

Finally, the following lemmas will prove useful in our analysis.

\begin{lemma}
    Let $A, B, B'$ be Hilbert spaces where $B$ is of a smaller or equal dimension to $B'$. Let $\ket{\psi}_{AB} \in \mathcal H_{AB}$ and $\ket{\phi}_{AB'} \in \mathcal H_{AB'}$ be quantum states. Suppose that there exists an isometry $V : \mathcal H_B \rightarrow \mathcal H_{B'}$ such that \begin{align}
        d(\ket{\phi}_{AB'}, I_A \otimes V_B \ket{\psi}_{AB}) \leq \epsilon.
    \end{align}
    Then no measurement of subsystem $A$ can distinguish $\ket{\psi}_{AB}$ from $\ket{\phi}_{AB'}$ with advantage exceeding $\epsilon$.
    \label{lem:approx-uhlmann}
\end{lemma}

\begin{proof}
    The mixed state of subsystem $A$ is invariant under any quantum channel applied to the other subspace, so $\ket{\psi}_{AB}$ cannot be distinguished from $I \otimes V \ket{\psi}_{AB}$ by a measurement on subsystem $A$. However, by \Cref{claim:states-indis}, no measurement can distinguish $I \otimes V \ket{\psi}_{AB}$ from $\ket{\phi}_{AB'}$ with advantage exceeding $\epsilon$.
\end{proof}

\begin{lemma}
    Let $O : \mathcal H_N \rightarrow \mathcal H_N$ be a quantum operator. Consider orthogonal subspaces $S_1, \dots, S_m \in \mathcal H_N$ of dimension $\mathsf{dim}(S_i)$, where the $S_i$ span all of $\mathcal H_N$. Suppose that
    \begin{align}
        \forall \ket{x_i} \in S_i, \ket{x_j} \in S_j, i\neq j, \braket{x_i | O^\dagger O |x_j} = 0,
    \end{align}
    and we further have \begin{align*}
        \forall \ket{x_i} \in S_i, \norm{O \ket{x_i}} \leq \epsilon_i.
    \end{align*}
    Then we have \begin{align*}
        \norm{O} \leq& \max_{i \in [m]} \epsilon_i.
    \end{align*}
    \label{lem:ortho-subspaces-norm}
\end{lemma}

\begin{proof}
    \cite{don2021extract}, Lemma 2.1.
\end{proof}

\begin{lemma}
    Let $\Pi : \algo H \rightarrow \algo H$ be a projector, and $\ket{\psi} \in \algo H$ be a unit norm quantum state. Suppose that \begin{align*}
        \norm{\Pi \ket{\psi}} \geq 1 - \epsilon.
    \end{align*}
    Then we have \begin{align*}
        \norm{\Pi^{\perp} \ket{\psi}} \leq \sqrt{2\epsilon}.
    \end{align*}
    \label{lem:proj-no-close}
\end{lemma}

\begin{proof}
    We have \begin{align*}
        \norm{\ket{\psi}}^2 =& \norm{\Pi \ket{\psi}}^2 + \norm{\Pi^{\perp} \ket{\psi}}^2 \\
        \geq& (1 - \epsilon)^2 + \norm{\Pi^{\perp} \ket{\psi}}^2,
    \end{align*}
    and $\norm{\ket{\psi}}^2=1$. Rearranging these inequalities, we obtain \begin{align*}
        \norm{\Pi^{\perp} \ket{\psi}}^2 \leq& 2\epsilon - \epsilon^2 \\
        \leq& 2\epsilon.
    \end{align*}
\end{proof}

\subsection{Indifferentiability}
Let us consider proving that oracle $O$ and construction $C^O$ is indifferentiable from oracle $Q$. Often, $Q$ will be a random oracle.

\begin{definition}
    Consider a construction $C^O$ with access to an oracle $O$, which syntactically matches ideal primitive $Q$. A simulator $\algo S$ is $(q, \epsilon)$-\emph{indifferentiable} for construction $C$ if, for all $q$-query quantum algorithms $\algo D$,
    \begin{align*}
        \left|\Pr[\algo D^{\algo S ^ Q, Q}() = 1] - \Pr[\algo D^{O, C^O}() = 1]\right| \leq \epsilon.
    \end{align*}
    The construction $C^O$ is indifferentiable from $Q$ if there exists an efficient simulator that is $(q(n), \negl(n))$-indifferentiable for all polynomial $q(n)$.
    \label{defn:sim-indiff}
\end{definition}

The following notions and helper lemma will be useful.

\begin{definition}
    A simulator $\algo S$ is $(q, \alpha)$-\emph{indistinguishable} if, for all $q$-query quantum algorithms $\algo D$,
    \begin{align*}
        \left|\Pr[\algo D^{\algo S ^ Q}() = 1] - \Pr[\algo D^O() = 1]\right| \leq \alpha.
    \end{align*}
    \label{defn:sim-ind}
\end{definition}

\begin{definition}
    A simulator $\algo S$ is $(q, \beta)$-\emph{consistent} if, for all $q$-query quantum algorithms $\algo D$,
    \begin{align*}
        \left|\Pr[\algo D^{\algo S ^ Q, Q}() = 1] - \Pr[\algo D^{\algo S ^ Q, C^{\algo S ^ Q}}() = 1]\right| \leq \beta.
    \end{align*}
    \label{defn:sim-cons}
\end{definition}

\begin{lemma}[\cite{Zhandry2018}]
    A simulator which is $(q, \alpha)$-indistinguishable and $(q, \beta)$-consistent is $(q, \alpha+\beta)$-indifferentiable.
    \label{lem:cons-and-ind}
\end{lemma}

It is worth observing that indifferentiability is not in general symmetric, because the simulator is allowed to be stateful. However, it is transitive in the following sense, which will be used to justify our modifications of the sponge.

\begin{lemma}
    Suppose that oracle $A$ and construction $C_1^A$ is indifferentiable from oracle $B$. Suppose that oracle $B$ and construction $C_2^B$ is indifferentiable from oracle $Q$. Then oracle $A$ and construction $C_2^{C_1^A}$ is indifferentiable from oracle $Q$.
    \label{lem:indiff-trans}
\end{lemma}

\begin{proof}
    We prove this via a simple sequence of hybrids, where $\algo D$ is a $q$-query distinguisher taking oracles which syntactically match $A$ and $Q$. Let $\algo S_1^B$ be a simulator in the indifferentiability experiment between $A$ and $B$, and $\algo S_2^Q$ be the simulator in the indifferentiability experiment between $B$ and $Q$. Let $\lambda \in \mathbb N$ be our security parameter. \begin{enumerate}[label=\textbf{(Hybrid \arabic*)},align=left]
        \item The provided oracles are $A, C_2^{C_1^A}$. Let $p_0 = \Pr[\algo D^{A, C_2^{C_1^A}}() = 1]$.
        \item The provided oracles are $\algo S_1^B, C_2^B$. Let $p_1 = \Pr[\algo D^{\algo S_1^B, C_2^B}() = 1]$. 
        \item The provided oracles are $\algo S_1^{\algo S_2^Q}, Q$. Let $p_2 = \Pr[\algo D^{\algo S_1^{\algo S_2^Q}, Q}() = 1]$. 
    \end{enumerate}
    The first to the second hybrid replaces the $A$ oracle with $S_1^B$, and the $C_1^A$ oracle with $B$. We know that $\algo D$ and $C_2$ are both efficient and therefore make a polynomial number of queries to these oracles, and so by the indifferentiability of $A, C_1^A$ from $B$ we have $\abs{p_1-p_0} = \negl(\lambda)$. The second to the third hybrid similarly replaces $B$ and $C_2^B$ with $\algo S_2^Q$ and $Q$. Once again, the distinguisher, simulators, and constructions are all efficient, so by indifferentiability of $B, C_2^B$ from $Q$ we have $\abs{p_2-p_1} = \negl(\lambda)$. The union bound now implies the result.
\end{proof}

\subsection{The sponge construction}\label{sec:sponge}

The sponge is a construction that uses a fixed-length permutation to produce a function whose domain and codomain are arbitrarily large strings. Classically, the sponge is known to be indifferentiable from a random oracle~\cite{BDPvA08}. The international hash standard SHA-3 is a sponge construction instantiated with the public Keccak permutation family~\cite{KeccakSub3, KeccakSponge3}.

Let $r$ and $c$ be positive integers, and let $\varphi : \bit^{r+c} \rightarrow \bit^{r+c}$ be a permutation. This data defines a sponge function $\Sp^{\varphi}$, described below. The natural security parameters are the ``capacity'' $c$ and the ``rate'' $r$. The rate and capacity of a state $x\in\bit^n$ are denoted with $x|_0^r$ and $x|_r^{r+c}$, respectively. In general, $x|_a^b=x_ax_{a+1}\dots x_{b-1}$ for $x=x_0x_{1}\dots x_{n-1}$.

\begin{definition}
    A string $x \in \bit^*$ is a valid input to the sponge if it is of the form $x = x_1 \Vert \dots \Vert x_p$ for $x_i \in \bit^r$, where $p \geq 1$ and $x_p \neq 0^r$.
    \label{defn:valid-sponge-ins}
\end{definition}

To obtain a proper function $\bit^* \rightarrow \bit^*$, we can use an injective pad function from $\bit^*$ to the set of valid sponge inputs, and then apply the sponge. A simple example of such a function is \begin{align*}
    \textsf{PAD}(x) = x \Vert 1 0^{r-(|x| +1 \mod r)}.
\end{align*}

The sponge function $\Sp^{\varphi}$ is defined as follows. To process an input, we initialize the state $0^r \Vert 0^c$. We then alternate (i.) XORing in the next block of the input to the top $r$ bits, beginning with $x_1$, with (ii.) applying the permutation $\varphi$ to the state, until the entire input is ``absorbed''. To produce the output, we then alternate (i.) outputting the top $r$ bits with (ii.) applying $\varphi$, until as many bits as needed are output. The process of producing the output is called squeezing, and is what gives the sponge unbounded range.

As an alternative to a function with unbounded range, we can consider a random function which maps inputs of the form $x = x_1 \Vert \dots \Vert x_p$ for $x_i \in \bit^r$ \textit{but} which possibly end in $0^r$, to random $r$ bit strings. While such a construction has only finite output length, using a \textsf{PAD} function satisfying the above it naturally corresponds to functions with unbounded range in the following way. Let $f : (\bit^r)^* \rightarrow \bit^r$ be such a function, and define function $g: \bit^* \rightarrow \bit^*$ as \begin{align*}
    g(x) = f(\textsf{PAD}(x)) \Vert f(\textsf{PAD}(x) \Vert 0^r) \Vert f(\textsf{PAD}(x) \Vert 0^{2r}) \dots.
\end{align*}

In this way, we have the following remark.

\begin{claim}
    Suppose that the sponge construction with a single round of squeezing and no constraint on the inputs, $\Sp^{\varphi} : (\bit^r)^* \rightarrow \bit^r$, is indifferentiable from a random oracle $f : (\bit^r)^* \rightarrow \bit^r$. Then, the full sponge construction with arbitrary squeezing and a valid pad function is indifferentiable from a random oracle $f: \bit^* \rightarrow \bit^*$. 
    \label{claim:sponge-approach}
\end{claim}

We will prove the former, i.e., that $\Sp^{\varphi} : (\bit^r)^* \rightarrow \bit^r$, is indifferentiable from a random oracle. 

\newpage
\section{Compressed oracle framework}
\label{sec:comp-o}

The compressed oracle technique, introduced by Zhandry \cite{Zhandry2018}, can be viewed as a quantum analog of lazy sampling. Since it is the main technical workhorse of our result, we give a complete and self-contained introduction to our formulation of the technique. The exposition is most similar to that of \cite{don2021extract,chung2021compressed}.

By analyzing a suitable purification of a random oracle\footnote{More precisely: a suitable Stinespring dilation (in fact, a generalization thereof) of the random oracle viewed as a quantum channel with memory} in a certain ``compressed'' basis, one obtains a query transcript corresponding to points the adversary has queried. This technique is essentially the only known tool for stateful simulation of a quantum-accessible randomized oracle, as is required for quantum indifferentiability proofs of domain extenders. Further, the technique allows for simple and direct proofs of query lower bounds for many natural random oracle search problems. 

\subsection{Purified oracles}

Let $f:[M] \rightarrow [N]$ be a uniform random function. We are interested in analyzing quantum algorithms that query such functions, and which know nothing about $f$ other than the distribution from which it was picked. For this introductory section we will neglect the algorithms workspace. We will assume here that $M=2^m$ and $N=2^n$ for integers $m$ and $n$. 

Going forward, we will use subscripts to indicate which quantum register(s) a state is supported on: here $X$ will denote a register holding an input, $Y$  a register holding an output. A subscript of an operator or a state is always a label of a register; we use superscripts for additional information. We will sometimes omit labels on states or operators if it is clear from context. We define the quantum oracle for $f$ as $O^f$, which acts on an input and output register as
 \begin{align*}
    O^f \ket{x}_X \ket{y}_Y :=& \ket{x}_X\ket{y \oplus f(x)}_Y.
\end{align*}
We could instead write a purified oracle $\algo P$ that acts on an input, output, and a third function register (subscript $F$) as 
\begin{align*}
    \algo P \ket{x}_X\ket{y}_Y \ket{f}_F := \ket{x}_X \ket{y \oplus f(x)}_Y \ket{f}_F.
\end{align*}
The third register has the set $\mathfrak F$ of all functions $[M]\rightarrow [N]$ as computational basis. Such a representation allows us to model a distribution on functions as the purified superposition over functions. As $f$ is drawn uniformly at random, we can instead prepare an initial function oracle 
\begin{align}
    \ket{\mathfrak F}_F := \frac{1}{\sqrt{|\mathfrak F|}} \sum_{f \in \mathfrak F} \ket{f}_F \label{eqn:init-func},
\end{align}
and replace each application of $O^f$ with the purified oracle $\algo P$.  Measuring the $F$ register to obtain $f$ after an algorithm has finished, and drawing $f$ uniformly at random and then answering queries with $O^f$, yield identical joint distributions of $f$ and the algorithms (possibly quantum) output, showing that the purified oracle is an exactly faithful simulation of a random oracle.

\subsection{The compressed basis}

We will now introduce the ``compressed basis'', which is a convenient basis in which to analyze the function register. First, let us consider enlarging the Hilbert space of the function register to the span of the set $\mathfrak D$ of functions from $[M]$ to $[N] \cup \{\bot\}$. The resulting function register can be viewed  as an $(M+1)$-dimensional register for every output. In particular, for $g \in \mathfrak D$, 
\begin{align*}
    \ket{g}_F = \ket{g(0)}_0 \ket{g(1)}_1 \dots \ket{g(M-1)}_{M-1}\,.
\end{align*}
We will use $\ket{\mu}$ to denote the uniform superposition over (non-$\bot$) outputs, i.e.,
\begin{align*}
    \ket{\mu} = \frac{1}{\sqrt{N}}\sum_{y \in [N]} \ket{y}.
\end{align*}
The initial superposition over all total functions (i.e., functions with no $\bot$ outputs) can then be written as \begin{align*}
    \ket{\mathfrak F}_F = \bigotimes_{x \in [M]} \ket{\mu}_x.
\end{align*}
The transformation from the standard to the compressed basis is the unitary which, for each input $x$, exchanges the uniform superposition $\ket{\mu}_x$ with $\ket{\bot}_x$.

\begin{definition}
    Define the register compression function $C^x \in \mathrm{SU}(N+1)$, acting on register $x \in [M]$ of a function, as \begin{align*}
        C^x := I_x - \proj{\bot}_x - \proj{\mu}_x + \outerprod{\bot}{\mu}_x + \outerprod{\mu}{\bot}_x.
    \end{align*}
    We can use this to define the full compression function $C_F \in \mathrm{SU}(\mathfrak D)$ as \begin{align*}
        C_F := \bigotimes_{x \in [M]} C^x.
    \end{align*}
\end{definition}
It is easy to verify that $C^2 = I$, and the initial purification compresses to the constant $\bot$ function, meaning $C \ket{\mathfrak F}_F = \bigotimes_{x \in [M]} \ket{\bot}_x$. Going forward, it will be helpful to denote basis vectors in the compressed basis by the set of non-$\bot$ input-output pairs (i.e. the input values where the partial function is defined, and the corresponding outputs). We call such a set $D$ as a ``database''. We write $x \in D$ to mean that there exists a pair $(x, y) \in D$ (with $y \neq \bot$), and let $|D|$ denote the total number of such pairs in $D$.

In this notation, the compression of the uniform superposition over all functions compresses to the empty database, $C \ket{\mathfrak F}_F = \ket{\emptyset}_F$. The action of querying can be understood in the compressed basis by undoing the compression operator (which is self-inverse), applying the purified oracle $\algo P$, and then re-compressing. It turns out that one only needs to (de)compress the register being queried. To define this, we use the notation
\begin{align}\label{eq:local-compression}
	L_{XF} \ket{x}_X \ket{f}_F \coloneqq  \ket{x}_XC^x_F\ket{f}_F,
\end{align}
in other words $L$ is the ``local compression'', compressing the $x$-th register of the $F$ register conditioned on the value $x$ in the first register. Note that this is formally distinct from $C^x$, which compresses a fixed register; the operators are related by the identity $L_{XF} \ket{x} = C^x_F$.

\begin{lemma}
    We have \begin{align*}
        C_F \algo P_{XYF} C_F^\dagger = L_{XF} \algo P_{XYF} L_{XF}.
    \end{align*}
    \label{lem:only-one-comp}
\end{lemma}

\begin{proof}
    Observe that $[C^x, C^{x'}] = 0$, and for $x' \neq x$ we have $\algo P C^{x'} \ket{x}_X \ket{y}_Y \ket{f}_F = C^{x'} \algo P \ket{x}_X \ket{y}_Y \ket{f}_F$. Working out the action on an arbitrary basis state, have \begin{align*}
        C_F \algo P_{XYF} C_F^\dagger \ket{x}_X \ket{y}_Y\ket{f}_F &= C_F \algo P_{XYF} C_F \ket{x}_X \ket{y}_Y\ket{f}_F \\
        &= \left(\bigotimes_{x' \in [M]} C^{x'}\right) \algo P_{XYF} \left(\bigotimes_{x'' \in [M]} C^{x''}\right) \ket{x}_X \ket{y}_Y\ket{f}_F\\
        &= \left(\bigotimes_{x' \in [M]} C^{x'}\right) \left(\bigotimes_{x'' \in [M], x'' \neq x} C^{x''}\right) \algo P_{XYF} C_{x} \ket{x}_X \ket{y}_Y\ket{f}_F \\
        &= L_{XF} \algo P_{XYF} L_{XF}\ket{x}_X \ket{y}_Y\ket{f}_F
    \end{align*}
\end{proof}

\begin{definition}
    We will sometimes use the notation $\cO_{XYF} \coloneqq L_{XF} \algo P_{XYF} L_{XF}$ to denote compressed oracle calls.
\end{definition}

Consider an experiment where a $t$ query adversary has workspace $A$, query register $X$ and output register $Y$, interacting with a random function oracle initialized as the empty database in register $F$. Parameterizing the adversary by unitaries $U^0, U^1, \dots, U^t$, we can write the final state \begin{align*}
    \ket{\psi^t}_{AXYF} = U^t_{AXY} \dots U^1_{AXY} \cO_{XYF} U^0_{AXY} \ket{0}_A \ket{0}_X \ket{0}_Y \ket{\emptyset}_F.
\end{align*}

\begin{corollary}[of \Cref{lem:only-one-comp}]
    In the state $\ket{\psi^t}$, the function register is supported fully on databases having at most $t$ non-$\bot$ output values.
    \label{cor:at-most-t}
\end{corollary}

Note that the function register needn't have exactly $t$ non-$\bot$ values; indeed, quantum algorithms can forget already learned outputs by uncomputing them, in which case they will not appear in the database.

Another important Corollary is that there is always an image of $x$ after decompressing the $x$-th register, so long as the databases are queried using $\cO$. We call the span of such states valid. To define them, we use the notation $\Pi^{x \in \mathsf{db}}_F \coloneqq I_x - \proj{\bot}_x$.

\begin{definition}
    We say that a database state $\ket{\psi}_F$ is ``valid'' if for all $x \in [M]$, it satisfies \begin{align*}
        \Pi^{x \in \mathsf{db}} C^{x} \ket{\psi} = C^x \ket{\psi}.
    \end{align*}
    Observe that such states form a subspace, and we use $\Pi^v$ to denote the projector onto this subspace.
\end{definition}

\begin{corollary}[of \Cref{lem:only-one-comp}]
    The operation $\cO$ preserves membership in the valid subspace.
    \label{cor:always-valid}
\end{corollary}

We will sometimes also need a projector onto databases containing a certain value, controlled on an external register. We use the notation $\Pi^{\in \mathsf {db}}_{XF}$ for this, such that \begin{align*}
    \inD_{XF} \ket{x}_X = \Pi^{x \in \mathsf{db}}_F.
\end{align*}
We have seen that the initial state $\ket{\emptyset}_F$ in the compressed basis evolves to a superposition of input-output pairs as queries are made. These states are sometimes called compressed database states. We will now analyze how these database states evolve under the query operator.

\subsection{Transition capacities}\label{sec:transitions}

Here we recall the framework of \cite{chung2021compressed} for proving query bounds using the compressed oracle. Our exposition of the compressed oracle is slightly different, and so in some places we include alternative proofs. Note also that certain aspects of the framework described here are simplified from that of \cite{chung2021compressed}, as we are not concerned with parallel queries.

Recall that $\mathfrak D$ denotes the set of functions from $[M]$ to $[N] \cup \{\bot\}$. We let $\mathfrak D_t$ denote the subset of $\mathfrak D$ to functions with at most $t$ non-$\bot$ values.

Let us define a predicate $\P$ as a subset of databases. A simple example of a predicate is the collision predicate $\mathrm{C}$ over $\mathfrak D$, which is the set of all databases which contain points $x, x' \in D$ with $x \neq x'$ and $D(x) = D(x')$. Recall that in our notation $x \in D \Leftrightarrow D(x) \neq \bot$, so collisions on $\bot$ do not count. 

\begin{definition}
    A predicate $\P$ is a subset of databases $P \subset \mathfrak D$. We sometimes also refer to a database $D \in \mathfrak D$ as ``satisfying'' $\P$ if $D \in \P$, or say $\P(D)=\top$ in this case (and otherwise $\P(D)=\bot$).
    
    A predicate has a negation, and a restriction to databases of size at most $t$, which we denote as follows. \begin{align*}
        \neg \P = \mathfrak D \setminus \P, && \P_t = \mathfrak D_t \cap \P, && \neg \P_t = \mathfrak D_t \setminus \P_t.
    \end{align*}
\end{definition}

The quantum transition capacity between predicate $\neg \P$ and predicate $\Q$ measures the amplitude transfer from databases in $\neg \P$ to ones in $\Q$ under a single query.

\begin{definition} \label{def:transition-capacity}
    The single-query quantum transition capacity from predicate $\neg \P$ to predicate $\Q$ is defined as follows. 
    \begin{align*}
        [\![\neg \P \rightarrow \Q]\!] = \norm{\Pi^{\Q}_F C_F \algo P_{XYF} C_F^\dagger \Pi^{\neg \P}_F}.
    \end{align*} Similarly, the quantum transition capacity for $q$ sequential queries is defined as $$[\![ \neg \P \xrightarrow{q} \Q]\!]=\sup_{U_1,\ldots,U_{q-1}} \| \Pi_F^{\Q} C_F \mathcal{P}_{XYF} \,U_{q-1} \, \mathcal{P}_{X,Y,F} \cdots \mathcal{P}_{XYF} \, U_1 \, \mathcal{P}_{XYF} \, C_F^\dagger \Pi_F^{\neg \P} \|  \, $$
where the supremum is over all positive $d \in \Z$ and all unitaries $U_1,\ldots,U_{q-1}$ acting on $\C[\mathcal{X}] \otimes \C[\mathcal{Y}] \otimes \C^d$. 
\end{definition}

Note that the definition of the quantum transition capacity is the same as in \cite{chung2021compressed}, Definition 5.5. We will state some of its properties that we will use while deriving the query lower bounds in the quantum adversary setting. We list them under the same lemma.

\begin{lemma}[\cite{chung2021compressed}, Lemma 5.6]  \label{quant-trans-cap-splitting-queries}
For any sequence of predicates $\P_0,\P_1,\ldots,\P_q$, 
$$ [\![\neg\P_0 \xrightarrow{q} \P_q]\!]
 \leq \sum_{s=1}^q [\![\neg\P_{s-1} \rightarrow {\P_{s}}]\!]\, .
$$    
\end{lemma}
\begin{lemma}[\cite{chung2021compressed}, Lemma 5.31] \label{quant-trans-cap-union}
 For any predicates $\P,\P'$ and $\Q$, we have
 $$[\![\Q\rightarrow {\P}]\!]\, \leq [\![\Q\rightarrow {\P} \cup \P']\!],$$
 $$[\![\Q\rightarrow {\P} \cup \P']\!] \leq [\![\Q\rightarrow {\P} ]\!]+[\![\Q\rightarrow \P']\!] , \text{ and}$$
 $$[\![{\P} \cap \P' \rightarrow \Q]\!] \leq  [\![{\P}  \rightarrow \Q]\!].$$
\end{lemma}

Let us now introduce the concept of a local property. Informally, a local property is a predicate which is determined only by the value of a database on a single special input point. Further, local properties are required to be monotone w.r.t. the special input. We will here use the notation $D|_x$ to denote the set of databases which agree with $D$ on all points $x' \neq x$. The following definition is adapted from \cite{chung2021compressed}, removing an ambiguity over which domain the local property is to be considered.

\begin{definition}
     Let $\L^{x, D} \subseteq D|_x$ be a property parameterized by an $x \in \mathcal{X}$ and $D \in \mathfrak{D}$. Consider $\L^{x,D}: D|_x \rightarrow \bit $ as a function to map $D'$ to $1$ if $D' \in \L^{x,D}$ and $0$ otherwise. As a database property $\L^{x,D}$ is local if
     \begin{enumerate}
         \item $\exists \hat{\L}^x:\mathcal{\bar Y} \rightarrow \bit$ with $\hat{\L}^x(D'(x))=\L^{x,D}(D') \, \text{ for all } D' \in D|_x$,
         \item $D' \in \L^{x, D} \wedge D'(x)= \bot$, then $D'[x \mapsto r] \in \L^{x,D} \text{ }$ for all $ r \in \mathcal{Y}$.
     \end{enumerate}
\end{definition}

We will usually use the symbol $\L$ to refer to a collection $\L \coloneqq \{\L^{x, D} \, : \, x \in \mathcal X, D \in \mathfrak D\}$ of local properties. By an abuse of notation, we will also sometimes use $\L$ to denote the union of these sets $\L \coloneqq \bigcup_{x, D} \{\L^{x, D} \, : \, x \in \mathcal X, D \in \mathfrak D\}$, depending on context. We say that $\L$ interpolates the transition $\neg \P \rightarrow \Q$ if it satisfies $\Q \subset \L \subset \P$ (here using the latter interpretation). We can define the distance $\Delta$ of a local property as the maximum number of images of the special point which change the truth value of the property.

\begin{definition} \label{def:distance}
    The distance $\Delta$ of a local property $\L^{x, D}$ is defined as \begin{align*}
        \Delta(\L^{x, D}) \coloneqq |\{y \, : \, \L^{x, D}(D[x \rightarrow y]) \neq \L^{x, D}(D[x \rightarrow \bot])\}|.
    \end{align*}
    The distance $\Delta$ of a collection of local properties $\L$ is defined as \begin{align*}
        \Delta(\L) \coloneqq \max_{\L^{x, D} \in \L} \Delta(\L^{x, D}).
    \end{align*}
\end{definition}

We can bound quantum transition capacities using local properties. In particular, if we can construct a family of local properties which interpolate the transition, the distance of this local property bounds the transition capacity.

\begin{lemma} \label{lem:transition-from-classical}
    Let $\neg \P$ and $\Q$ be predicates, and suppose the local property family $\L$ interpolates the transition $\neg \P \rightarrow \Q$ (meaning we have $\Q \subseteq \L \subseteq \P$). The quantum transition capacity then satisfies \begin{align*}
        [\![\neg \P \rightarrow \Q]\!] \leq 4\sqrt{ \Delta(\L) / N}.
    \end{align*}
    \label{lem:bound-framework}
\end{lemma}

\begin{proof}
    We have $\Q \subseteq \L$ and $\neg \P \subseteq \neg \L$ by assumption, which implies that $\Pi^\Q \leq \Pi^\L$ and $\Pi^{\neg \P} \leq \Pi^{\neg \L}$ in the semi-definite order. It follows that \begin{align*}
        [\![\neg \P \rightarrow \Q]\!] =& \norm{\Pi^{\Q}_F C_F \algo P_{XYF} C_F^\dagger \Pi^{\neg \P}_F} \\
        \leq& \norm{\Pi^{\L}_F C_F \algo P_{XYF} C_F^\dagger \Pi^{\neg \L}_F} \\
        =& [\![\neg \L \rightarrow \L]\!].
    \end{align*}
    We will bound the transition capacity $\neg \L \rightarrow \L$. Let us define $\mathcal Y^{x, D}=\{y \, : \, \L^{x, D}(D[x \rightarrow y]) \neq \L^{x, D}(D[x \rightarrow \bot])\}$ as the set of assignments to $x$ that cause a transition in $\L^{x, D}$. Now we have \begin{align*}
        \norm{\Pi^{\L}_F C_F \algo P_{XYF} C_F^\dagger \Pi^{\neg \L}_F} =& \norm{\Pi^{\L}_F L_{XF} \algo P_{XYF} L_{XF} \Pi^{\neg \L}_F} \\
        \leq& 2\norm{\Pi^{\L}_F L_{XF} \Pi^{\neg \L}_F} + \norm{\Pi^{\L}_F \algo P_{XYF} \Pi^{\neg \L}_F} &\text{(Orthogonality of $\Pi^{\neg \L}, \Pi^{\L}$)}
    \end{align*}
    Observe that $\algo P$ preserves the computational basis, so the second term is $0$. We will focus on bounding the first term, and drop the $Y$ register as it is not acted on. Let us define $\mathcal Y^{x, D}=\{y \, : \, \L^{x, D}(D[x \rightarrow y]) \neq \L^{x, D}(D[x \rightarrow \bot])\}$ as the set of assignments to $x$ that cause a transition in $\L^{x, D}$. We have $\inD_{XF} + \ninD_{XF} = I$, so we can analyze two subspaces separately. Noting that $\Pi^{\L}, \Pi^{\neg \L}$ are diagonal in the computational basis and $L_{XF}$ preserves the $X$ register and the database values outside of $X$. By \Cref{lem:ortho-subspaces-norm} we can WLOG consider a state of the form \begin{align*}
        \ket{\psi_{x, D}} = \sum_{z \in [N] \cup \{\bot\} \setminus \mathcal Y^{x, D}} \alpha_{z} \ket{x}_X \ket{D[x \rightarrow z]}_F,
    \end{align*}
    where $D \not\in \L$. 
    \begin{enumerate}[label=(\arabic*)]
        \item Case $x \not\in D$. We have \begin{align*}
            \norm{\Pi^{\L}_F L_{XF} \Pi^{\neg \L}_F \ninD_{XF} \ket{\psi}} =& \abs{\alpha_\bot}\norm{\Pi^{\L}_F \frac{1}{\sqrt{N}} \sum_{u \in [N]} \ket{x}_X \ket{D[x\rightarrow u]}_F} \\
            =& \abs{\alpha_{\bot}} \sqrt{\frac{\Delta(\L^{x, D})}{N}} \\
            \leq& \sqrt{\frac{\Delta(\L)}{N}}.
        \end{align*}
        \item Case $x \in D$. We have \begin{align*}
            \norm{\Pi^{\L}_F L_{XF} \Pi^{\neg \L}_F \inD_{XF} \ket{\psi}} =& \Bigg\Vert\Pi_F^{\L}\sum_{z \in [N] \setminus \mathcal Y^{x, D}}\alpha_z \ket{x}_X \otimes \\& \left(\ket{D[x\rightarrow z]}_F  - \frac{1}{N} \sum_{u \in [N]} \ket{D[x\rightarrow u]} + \frac{1}{\sqrt{N}} \ket{D[x\rightarrow \bot]} \right)\Bigg\Vert \\
            =& \norm{\frac{1}{N}\sum_{z \in [N] \setminus \mathcal Y^{x, D}, u \in \mathcal Y^{x, D}} \alpha_z \ket{x}_X \ket{D[x\rightarrow u]}} \\
            \leq& \frac{1}{N} \sqrt{\sum_{u \in \mathcal Y^{x, D}} \left(\sum_{z \in [N] \setminus \mathcal Y^{x, D}} \abs{\alpha_z}\right)^2} \\
            \leq& \sqrt{\frac{\Delta(\L^{x, D})}{N}} \\
            \leq& \sqrt{\frac{\Delta(\L)}{N}}
        \end{align*}
    \end{enumerate}
    Putting the cases together, we obtain the bound \begin{align*}
        [\![\neg \P \rightarrow \Q]\!] \leq& 4\sqrt{\frac{\Delta(\L)}{N}}
    \end{align*}
\end{proof}

We will now see how to connect the recorded set of input-output pairs to the adversary's knowledge, and how to use this to prove query lower bounds.

\subsection{Query lower bounds}
A key fact about compressed oracles is that the input-output pairs in the compressed database almost entirely capture the adversary's knowledge about the function. This is formally captured by the fundamental lemma, which says that any set of input-output pairs which is not in the compressed database has exponentially small probability to be in the uncompressed database.

\begin{lemma}[Fundamental lemma]
    Let $\mathbf x = \{(x_1, y_1), \dots (x_l, y_l)\}$ be a set of $l$ input-output pairs with distinct inputs and no $\bot$ output. Let $\Pi^{\mathbf x}$ act on a database state $\ket{D_f}$ as the projection onto databases that are consistent with $\mathbf x$. Suppose further that the database state is valid, $\Pi^v \ket{D_f} = \ket{D_f}$. Then we have \begin{align*}
        \norm{\Pi^{\mathbf x} C \ket{D_f}} \leq \norm{\Pi^{\mathbf x} \ket{D_f}} + \sqrt{\frac{l}{N}}.
    \end{align*}
\end{lemma}

\begin{proof}[Proof, adapted from \cite{chung2021compressed}]
    We have \begin{align*}
        \norm{\Pi^{\mathbf x} C \ket{D_f}} &\leq \norm{\Pi^{\mathbf x} \ket{D_f}} + \norm{(\Pi^{\mathbf x} - \Pi^{\mathbf x} C) \ket{D_f}} & \text{(Triangle Inequality)} \\
        &\leq \norm{\Pi^{\mathbf x} \ket{D_f}} + \norm{(\Pi^{\mathbf x} - \Pi^{\mathbf x} C)\Pi^v}.
    \end{align*}
    To bound the second term, we can bound the Fourier matrix elements\footnote{we take the Fourier transform register-wise over $\mathbb Z_N$, ignoring the $\bot$ symbol}, restricted to the $l$ database entries which are acted on. \begin{align*}
        \abs{\bra{\tilde {\mathbf p}} (\Pi^{\mathbf x} - \Pi^{\mathbf x} C) \ket{ \tilde{\mathbf q}}} \leq& \begin{cases}
            0 & \text{(If $\tilde{\mathbf p}$ contains $\bot$)} \\
            0 & \text{(If $\tilde{\mathbf q}$ contains no $\bot$ or $\tilde 0$)} \\
            N^{-l} & \text{(If $\tilde{\mathbf q}$ contains a $\bot$ or $\tilde 0$)}
        \end{cases}
    \end{align*}
    Finally observe that \begin{align*}
        \Pi^v \ket{\tilde q} =& \begin{cases}
            0 & \text{(If $\tilde q$ contains $\tilde 0$)} \\
            \ket{\tilde q} & \text{(Otherwise)}
        \end{cases}
    \end{align*}
    We can then bound the operator norm by the Frobenius norm to obtain \begin{align*}
        \norm{(\Pi^{\mathbf x} - \Pi^{\mathbf x} C)\Pi^v}^2 \leq& \sum_{\text{$\tilde{\mathbf p}$ has no $\bot$, $\tilde{\mathbf q}$ contains $\bot$ and no $\tilde 0$}} \abs{\bra{\tilde {\mathbf p}} (\Pi^{\mathbf x} - \Pi^{\mathbf x} C) \ket{ \tilde{\mathbf q}}}^2 \\ 
        \leq& \underbrace{N^l}_{{\text{($\tilde p$ choices)}}} \cdot \underbrace{l N^{l-1}}_{\text{($\geq \tilde q$ choices)}} \cdot N^{-2l} \\
        \leq& \frac{l}{N}
    \end{align*}
\end{proof}

This translates to the following algorithmic statement. 

\begin{corollary}[\cite{chung2021compressed}, Corollary 4.2]\label{cor:fund-lemma-with-R}
Let $R \subseteq {\cal X}^l \times {\cal Y}^l$ be a relation. 
Let $\algo A$ be an oracle quantum algorithm that outputs $(x_1,\dots,x_l) \in\mathcal{X}^l$ and $(y_1,\dots,y_l) \in {\cal Y}^l$. Let $p$ be the probability that ${y_i} = H({x_i})$ for all $i=1,\dots,l$ and $((x_1,\dots,x_l),(y_1,\dots,y_l)) \in R$ when $\algo A$ has interacted with the standard random oracle, initialized with a uniformly random function $H$. Similarly, let $p'$ be the probability that ${y_i} = D({x_i})$ and $((x_1,\dots,x_l),(y_1,\dots,y_l)) \in R$ when $\algo A$ has interacted with the compressed oracle instead and $D$ is obtained by measuring its internal state (in the computational basis). Then
$$
\sqrt{p} \leq \sqrt{p'} + \sqrt\frac{l}{N} \, .
$$
\end{corollary}

\subsection{Efficient representation}

Observe that a compressed database with at most $t$ input-output pairs can straightforwardly be stored using $O(tn+tm)$ space, by storing (say) a list of input-output pairs with non-$\bot$ outputs sorted by input, potentially with padding if there are fewer than $t$ such pairs. It turns out that one can also efficiently implement the query operator in this representation. In this section, we will describe a canonical simulator which efficiently simulates a random oracle for a quantum algorithm, based on the description of \cite{Zhandry2018}.

By \Cref{lem:only-one-comp}, it suffices to be able to implement $\cO = L \algo P L$ as a circuit, controlled on the first register, on the efficient database representation. Note that this operation adds at most a single input/output pair to the database. In this representation, the databases are ``padded'', such that the list of defined input output pairs appear at the beginning of the list, followed by dummy $(\infty, \bot)$ pairs such that the stored state is always the same size.

\begin{algorithm}
    \label{alg:compress-circ}
    \caption{Compression operator circuit}
    \textbf{Input:} $\ket{\psi}=\sum \alpha_{x,y,D} \ket{x}_X\ket{y}_Y\underbrace{\ket{(x_1,y_1)}_{D_1} \dots \ket{(x_t, y_t)}_{D_t}}_{\text{$D$}}$
    
    \textbf{Require:} $x_i \in [M] \cup \{\infty\}, y_i \in [N] \cup \{\bot\}$, $D$ sorted by $x_i$ with no duplicate non-$\infty$ inputs, $y_i = \bot$ if and only if $x_i = \infty$.
    
    \textbf{Output:} $\ket{\psi'}=\sum \alpha_{x,y,D'} \ket{x}\ket{y}\underbrace{\ket{(x_1',y_1')}_{D'_1} \dots \ket{(x_{t+1}', y_{t+1}')}_{D'_{t+1}}}_{D'}$ s.t. $\ket{\psi'}= C^x\ket{\psi}$ 
    \begin{enumerate}[label=(\arabic*)]
        \item Append $(\infty, \bot)$ to the end of register $D$ \Comment{Upper case for registers}
        \item For every $j$ from $1$ up to $t$: \begin{enumerate}[label=(\arabic*')]
            \item Compute flag $f \leftarrow (D_j[0]=X\text{ and }D_{j+1}[0] > X)$ \Comment{Zero based indexing}
            \item If $f$, swap registers $D_j \leftrightarrow D_{j+1}$
            \item Uncompute $f \leftarrow f \oplus (D_{j+1}[0]=X\text{ and }D_{j}[0] > X)$ \Comment{$\oplus$ is bitwise addition}
        \end{enumerate}
        \item If $D_{t+1}[1]=\bot$, perform $D_{t+1}[0] \leftarrow D_{t+1}[0] \oplus X \oplus \infty$ \Comment{Exchanges values $x \leftrightarrow \infty$}
        \item Perform $C^{D_{t+1}[0]}$ on register $D_{t+1}[1]$
        \item If $D_{t+1}[1]=\bot$, perform $D_{t+1}[0] \leftarrow D_{t+1}[0] \oplus X \oplus \infty$
        \item For every $j$ from $t$ down to $1$: \begin{enumerate}[label=(\arabic*')]
            \item Compute flag $f \leftarrow (D_j[0]=X\text{ and }D_{j+1}[0] > X)$
            \item If $f$, swap registers $D_j \leftrightarrow D_{j+1}$
            \item Uncompute $f \leftarrow f \oplus (D_{j+1}[0]=X\text{ and }D_{j}[0] > X)$
        \end{enumerate}
    \end{enumerate}
    
\end{algorithm}

To simulate the full query operator, we can: \begin{enumerate}[label=(\arabic*)]
    \item Uncompress the database on input $x$ using Algorithm \ref{alg:compress-circ}
    \item Answer the query using this database
    \item Run Algorithm \ref{alg:compress-circ} once again, omitting step (1), to recompress the database.
\end{enumerate}
Note that we may omit Step (1) on re-compression because the $\algo P$ operator preserves the computational basis on the input register, and we already created space for the $x$ input pair in the first step.

\newpage
\section{Developing the model}
\label{sec:oracle-proofs}

Here we formally develop the view of permutation oracles which we apply to the sponge.

\subsection{An alternative picture of permutation oracles}\label{sec:alt-oracles}

Let $\varphi : \{0,1\}^n \rightarrow \{0,1\}^n$ be a random injective function, i.e. a permutation on $n$ bit strings. Let $r, c$ be integers such that $r + c=n$, and let $s|_a^b$ denote the substring of $s$ starting from the $a$-th (inclusive) through the $b$-th (exclusive) bit. 

Let $\mathfrak H$ and $\mathfrak K$ be subgroups of the symmetric group $S_{2^n}$ defined by the following sets. 

\begin{align*}
    \mathfrak H := \{\pi \, : \, \pi(x \Vert y) =& (x \oplus h(y))\Vert y \text{ for some } h:\bit^c \rightarrow \bit^r\} \\
    \mathfrak K := \{\pi \, : \, \pi(x \Vert y) =& x\Vert (y \oplus k(x)) \text{ for some } k:\bit^r \rightarrow \bit^c\}
\end{align*}
   
    Note that a random $\varphi \sim S_{2^{n}}$ can equivalently be sampled by choosing a random $\pi \sim S_{2^n}$, and sampling random $\omega_h \sim \mathfrak H$, and $\sigma_k, \tau_{k'} \sim \mathfrak K$ (defined by functions $h, k, k'$ respectively), and defining 
    \begin{equation}\label{eq:decomposed-permutation}
        \varphi := \omega_h \circ \tau_{k'} \circ \pi \circ \sigma_k.
    \end{equation}
When there is no risk of confusion, we will sometimes write the above permutation simply as $\varphi = hk'\pi k$. We will now consider two access models for $\varphi$.

\begin{world}
    An adversary $\algo A$ receives access to quantum oracles for both $\varphi$ and $\varphi^{-1}$.
    \label{game:real-permutation-world}
\end{world}

\begin{world}
    An adversary $\algo A$ receives access to quantum oracles for random functions $h:\bit^c \rightarrow \bit^r$ and $k, k':\bit^r \rightarrow \bit^c$ (each independently uniform random), which correspond to permutations $\omega_h \in \mathfrak H$ and $\sigma_k, \tau_{k'} \in \mathfrak K$ as described above. The adversary is also given quantum query access to a permutation $\pi:\bit^n \rightarrow \bit^n$ and its inverse. Define
    \begin{equation}
        \varphi := \omega_h \circ \tau_{k'} \circ \pi \circ \sigma_k.
    \end{equation}
    \label{game:ideal-permutation-world}
\end{world}

It is intuitive that the interface for $\varphi$ in \Cref{game:ideal-permutation-world} is stronger than the interface for $\varphi$ in \Cref{game:real-permutation-world}. This is because the adversary can use $h, k, k'$ and $\pi$ to implement $\varphi$ and $\varphi^{-1}$, but is also free to use the oracles in whatever ways it wishes. Indeed, we can formalize this with an indifferentiability proof. Consider the trivial construction $C^\varphi = \varphi$.

	\begin{lemma} \label{lem:alternate-perm}
		Let World 1' be like \cref{game:real-permutation-world}, but $\mathcal A$ additionally has oracles for $h,k,k'$ as in \cref{game:ideal-permutation-world}, and for 
		    \begin{equation}
			\pi := \tau_{k'} \circ \omega_h \circ\varphi \circ \sigma_k.
		\end{equation}
		Then
		\begin{enumerate}
			\item For all $\mathcal A$ in \cref{game:real-permutation-world} playing game $G^\varphi$ there exists $\mathcal A'$ in World 1' such that
			\begin{equation*}
				\Pr[\mathcal A\mathrm{\ wins\ }G^\varphi]=	\Pr[\mathcal A'\mathrm{\ wins\ }G[\varphi]].
			\end{equation*}
			Here the notation $G^\varphi$ means that $G$ depends on $\varphi$ but not on $h,k,k'$ or $\pi$.
			\item World 1' is indistinguishable from \cref{game:ideal-permutation-world}.
		\end{enumerate}
\end{lemma}

\begin{proof}
  Statement 1 is clear, $A'$ just ignores its additional oracles.
   For statement 2, it suffices to observe that the oracles in the two worlds have the same joint distribution.
\end{proof}

In particular, the content of \Cref{lem:alternate-perm} is that for any security game $G^\varphi$, the maximum winning probability in  \Cref{game:ideal-permutation-world} upper-bounds the maximum winning probability in \Cref{game:real-permutation-world}. Therefore, it suffices for us to show all lower bounds in the latter model. By \Cref{lem:indiff-trans}, and the fact that indistinguishability is a special case of indifferentiability, it suffices to show indifferentiability in this model as well. In certain cases we will consider adversaries which have the full truth table of $\pi$, which is clearly a stronger model still.

\subsection{An alternative picture of the sponge: the Msponge}
\label{subsec:alt-sponge}

We now define an alternative construction that is more suited to our indifferentiability proof. We call this construction ``the Msponge''. The Msponge is identical to the sponge as defined in \Cref{sec:sponge} in all aspects except how each block of the input is absorbed. Specifically, in each absorbing round, instead of XORing the next block $x_i \in \bit^r$ of the input into the rate wire, we discard the contents of the rate wire and replace them with $x_i$; we then apply the permutation $\varphi$ to the rate and capacity wires as normal. We emphasize that all other aspects of the Msponge are identical to the sponge. In this section, we show that the sponge and the Msponge are equivalent in an appropriate sense, up to an overhead of the maximum block length.

Let $f, g : (\bit^r)^* \rightarrow \bit^r$ be random functions. The function $f$ will represent the standard sponge construction, and the function $g$ will represent the Msponge. To define the two constructions $C_1^f, C_2^g : (\bit^r)^* \rightarrow \bit^r$, we will need to introduce the notion of ``fixing'' an input. Intuitively, this map describes the correspondence between inputs to $f$ and inputs to $g$, which are in bijection with one another.

\begin{definition}
    Let $f: (\bit^r)^* \rightarrow \bit^r$ be a function. For an input of the form $x = x_1 \Vert \dots \Vert x_m \in (\bit^r)^*$, we define $\mathsf{fix}(x)$ recursively as follows. \begin{enumerate}[label=(\arabic*)]
        \item The fix of the empty string $\epsilon$ is \begin{align*}
            \mathsf{fix}(\epsilon)=\epsilon
        \end{align*}
        \item The fix of a one-block string $x_1$ is\begin{align*}
            \mathsf{fix}(x_1)=x_1
        \end{align*}
        \item The fix of a multi-block string $t \Vert x_m$, where $t$ is a non-empty string and $x_m \in \bit^r$, is \begin{align*}
            \mathsf{fix}(t \Vert x_m) = \mathsf{fix}(t) \Vert x_m \oplus f(\mathsf{fix}(t)).
        \end{align*}
    \end{enumerate}
    \label{defn:tail-fixing}
\end{definition}
It is worth observing that $\mathsf{fix}$ is a bijective function, for any function $f$. Computing the inverse is straightforward, by reversing the above recursion. We can now define our constructions.

\begin{definition}
    Let us define $C_1^f$ as $f \circ \mathsf{fix}$. Let us define $C_2^g$ as $g \circ \mathsf{fix}^{-1}$.
\end{definition}

The key point is that these constructions are indifferentiable from $g$ and $f$ respectively, as formalized in the following lemma. This justifies our choice to prove that the Msponge, where we replace the top wire with the next input rather than XORing it in, is indifferentiable will indeed be sufficient to prove that the standard sponge is indifferentiable from a random oracle.

\begin{lemma}
    The construction $C_1^f$ is perfectly indifferentiable from $g$, and the construction $C_2^g$ is perfectly indifferentiable from $f$. 
    \label{lem:tail-fix-ok}
\end{lemma}

\begin{proof}
    In the first case, we set the simulator as $C_2^g$. In the second case, we set the simulator as $C_1^f$. By \Cref{lem:tail-fixing-eq} (below), both worlds are perfectly indistinguishable, even given unlimited queries. Moreover, the simulator can be implemented such that it answers a query in $O(l)$ time, where $l$ is an upper bound on the block length of a query, by \Cref{lem:tail-fixing-implementation}.
\end{proof}

\begin{lemma}
    Let $f : (\bit^r)^* \rightarrow \bit^r$ be a random function, which defines a corresponding $\mathsf{fix}$ function. Then, the truth table of $f \circ \mathsf{fix}$ cannot be distinguished from $f$ with any advantage. In other words, $f \circ \mathsf{fix}$ is also a uniform random function, though one dependent on $f$.
    \label{lem:tail-fixing-eq}
\end{lemma}

\begin{proof}
    To see this, consider sampling the output values of $f$ in order of increasing input length. The fix of the empty string and $r$-length strings is the identity, so we clearly can sample i.i.d. strings from $\bit^r$ for images of such strings. Now inductively consider selecting the images of $\bit^{kr}$, where $f$ has been defined for all $\bit^{<kr}$. The value of $f$ on such strings suffices to define the $\mathsf{fix}$ of strings $\bit^{kr}$. Further, $\mathsf{fix}$ is a bijective function. The composition of a random function with a bijection on it's domain is still a random function, so we may sample each image of $\bit^{kr}$ as an i.i.d. string in $\bit^r$.
\end{proof}

\begin{lemma}
    Let $f : (\bit^r)^* \rightarrow \bit^r$ be a random function, which defines a corresponding $\mathsf{fix}$ function. Then, given query access to $g = f \circ \mathsf{fix}$, one can implement an $l$ block query to $f$ using $O(l)$ queries to $g$.
    \label{lem:tail-fixing-implementation}
\end{lemma}

\begin{proof}
    First note that, for single block inputs the $\mathsf{fix}$ operator is the identity, so a call to $f$ is a call to $g$. Furthermore, the input to $g$ is then $\mathsf{fix}^{-1}$ on the original input (trivially, because the fix is the identity). This forms our base case. We will also need the following characterization of $\mathsf{fix}^{-1}$.
    
    \begin{claim}
        We can characterize $\mathsf{fix}^{-1}$ recursively as follows. \begin{enumerate}[label=(\arabic*)]
            \item The inverse fix of the empty string $\epsilon$ is \begin{align*}
                \mathsf{fix}^{-1}(\epsilon)=\epsilon
            \end{align*}
            \item The inverse fix of a one-block string $x_1$ is\begin{align*}
                \mathsf{fix}^{-1}(x_1)=x_1
            \end{align*}
            \item The inverse fix of a multi-block string $t \Vert x_m$, where $t$ is a non-empty string and $x_m \in \bit^r$, is \begin{align*}
                \mathsf{fix}^{-1}(t \Vert x_m) = \mathsf{fix}^{-1}(t) \Vert x_m \oplus f(t).
            \end{align*}
        \end{enumerate}
        \label{claim:tail-fixing-inv}
    \end{claim}
    
    Inductively, let $u_1 \Vert \dots \Vert u_l = \mathsf{fix}^{-1}(x_1 \Vert \dots \Vert x_l)$, and consider an input $x_1 \Vert \dots \Vert x_{l}$. Let $y$ be the output register. By assumption, we can use $l-1$ queries to $g$ to compute the value of $f \circ \mathsf{fix}$ on input $x_1 \Vert \dots \Vert x_{l-1}$ into the register $x_l$, while simultaneously computing the $\mathsf{fix}$ of $x_1 \Vert \dots \Vert x_{l-1}$. We are then left with the strings $u_1 \Vert \dots \Vert u_{l-1} \Vert x_l \oplus f(x_1 \Vert \dots \Vert x_{l-1}) \Vert y = \mathsf{fix}^{-1}(x_1 \Vert \dots \Vert x_l) \Vert y$. We can then call the oracle for $g = f \circ \mathsf{fix}$ a single time on the whole input register to obtain the string $\mathsf{fix}^{-1}(x_1 \Vert \dots \Vert x_l) \Vert y \oplus f(x)$. This procedure uses $l$ calls exactly, and we can uncompute the change to the input with $2l$ calls.
\end{proof}

With this in place, by \Cref{lem:indiff-trans} it will suffice to prove that the Msponge is indifferentiable from a random oracle.

\subsection{Msponge terminology and combinatorics}

Consider the Msponge construction $\MSp^\varphi$, where we view the underlying permutation $\varphi$ in the alternative form described in Section \ref{sec:alt-oracles}. In this view, $\varphi$ is constructed from three random functions $h, k, k'$ and a random permutation $\pi$ as given in \Cref{eq:decomposed-permutation}.

We will now define various properties of compressed oracle databases $D_k, D_{k'}$ corresponding to the functions $k, k'$ above. These properties will depend on the permutation $\pi$, and will implicitly refer to the structure of the Msponge construction $\MSp^\varphi$. The properties we define will not depend on $h$.

We will always assume that $\pi$ is satisfies the following desirable property.
\begin{definition}
    We say that permutation $\pi$ is ``good'' if for any $x_1, x_2 \in \bit^r$, the number of suffixes $z \in \bit^c$ such that $\pi(x_1 \Vert z)$ begins with $x_2$ is at most $O(n + 2^{c-r})$.
    \label{def:perm-good}
\end{definition}
It follows from \Cref{lem:X-pairs-uniform-tail-bound} that a fraction $1-O(2^{-n})$ of permutations are good, validating this assumption.

\begin{definition}\label{def:tail}
    Let $z \in \bit^c$, and fix $D_k, D_{k'}$ and $\pi$. We say that $z$ has a ``tail'' under the following recursive conditions. 
    \begin{enumerate}[label=(\arabic*)]
        \item $z=0^c$ has a tail.
        \item $z$ has a tail if it can be reached as an internal state by a single round of absorption from a state $z_p$ which has a tail. That is, there exist $x_p \in D_k$, $x_i \in D_{k'}$, a $z_p \in \bit^c$ with a tail, and a $z_i \in \bit^c$ such that 
        \begin{align*}
            x_i \Vert z_i =& \pi(x_p \Vert z_p \oplus D_k(x_p)) \\
            z =& z_i \oplus D_{k'}(x_i).
        \end{align*}
    \end{enumerate}
    A tail of $z$ is a string of inputs required to reach $z$ according to the above conditions. Specifically, in case (1) the empty string is the unique tail of $z = 0^c$, and in case (2), any tail of $z_p$ concatenated with $x_p$ is a tail of $z$. We denote by $\mathsf{tail}(z)$ the set of tails of $z$.
\end{definition}

One can think of the set of internal states $z$ with a tail as those that an adversary ``knows how to reach''. We define $\mathsf{tail}_1(z)$ to be the lexicographically first tail of $z$, and will often deal with databases where tails are unique. We define the ``head'' of a tail of $z$ to be empty in case (1) above, and $x_i$ in case (2) above. We set $\mathsf{head}(z) = \mathsf{head}(\mathsf{tail}_1(z))$. Observe that both of these notions have an implicit dependence on $D_k, D_{k'}$ which we suppress here for notational convenience.

The definition here does not depend on the database for $h$. The reason for this is clearer to see in the Msponge: there, every call to $h$ during an internal round is immediately discarded, as the next input takes its place, and the final call to $h$ does not affect the state wire. In this way, the reachable states are decoupled from $h$.

We will also sometimes use this definition in the standard sponge, e.g. in \Cref{sec:bounds-proofs}. This definition still makes sense in that context, because an adversary that knows an input reaching a $z$ value could easily recover the tail of $z$, by evaluating the $\mathsf{fix}^{-1}$ function using the sponge, and in the other direction an adversary knowing a tail of $z$ could evaluate $\mathsf{fix}$ to recover an input reaching $z$.

We also require the notion of an intermediate pair. These correspond to the states an adversary ``knows how to reach'' in the intermediate stage of applying the block function, after applying $\pi$ and before applying $k'$.

\begin{definition}
\label{def:IP}
    Fix a database pair $(D_k, D_{k'})$. We say a pair $(x, z)$ is an intermediate pair if there exists $x_p \in D_k$ and a $z_p \in \bit^c$ with a tail such that 
    \begin{align*}
        x \Vert z = \pi(x_p \Vert z_p \oplus k(x_p))\,.
    \end{align*}
    We let $\mathsf{IP}(D_k, D_{k'})$ denote the set of all intermediate pairs.
\end{definition}

Note that because $\pi$ is a permutation, there is no multiplicity for intermediate pairs.  We can now define the set of good databases.

\begin{definition}\label{def:good-1}
    A database pair $(D_k, D_{k'})$ is \textbf{good} if every $z \in \bit^c$ has at most one tail, and $\mathsf{IP}(D_k, D_{k'})$ has unique $x$-values. Formally, we require both of the following conditions:
    \begin{align*}
        \forall z \in \bit^c&, \,|\mathsf{tail}(z)| \leq 1 \\
        \forall (x_1, z_1), (x_2, z_2) \in \mathsf{IP}(D_k, D_{k'})&, \,(x_1, z_1) \neq (x_2, z_2) \Leftrightarrow x_1 \neq x_2
    \end{align*}
\end{definition}

Note that the empty databases are clearly good, because the only state value with a tail is $0^c$ and the set of intermediate states is empty. A useful fact about good databases is that the number of state values with a tail is bounded by the number of queries, and the number of intermediate pairs is bounded by the number of queries squared.

\begin{lemma}
    Fix databases $(D_k, D_{k'})$ which are good and have at most $t$ non-$\bot$ inputs each. Then there are $O(t)$ values with a tail and $O(t^2)$ intermediate pairs. Formally, we have $|\{z \, : \, |\mathsf{tail}(z)| \geq 1\}| \leq O(t)$, and $|\mathsf{IP}(D_k, D_{k'})| = O(t^2)$.
    \label{lem:few-tails-ips}
\end{lemma}

\begin{proof}
    Each $z \neq 0^c$ value with a tail has a non-empty one, which means there exists a unique $(x_i, z_i) \in \mathsf{IP}(D_k^t, D_{k'}^t)$ such that $z = z_i \oplus k'(x_i)$. In this way, each $z \neq 0^c$ with a tail corresponds to a unique intermediate pair, for which $x_i \in D_{k'}$. In a good database the $x_i$'s are distinct for all distinct intermediate pairs, and we have $|D_{k'}| \leq t$, so there are at most $t$ such pairs. Note that we excluded $0^c$ from this argument, so we have at most $t+1$ values of $z$ with a tail.
    
    The set of intermediate pairs is constructed directly from $z$ values with a tail, and values of $x$ such that $x \in D_k$. There are at most $O(t)$ of each, so there are $O(t^2)$ pairs of each, and hence $|\mathsf{IP}(D_k^t, D_{k'}^t)| = O(t^2)$.
\end{proof}

Towards bounding the transitions from good to bad, we will need the notion of a bad attach. This is the ``bad event'' for a query to $k$.

\begin{definition}\label{def:badattach}
    Fix good databases $(D_k, D_{k'})$, and fix input $x \in \bit^r$ s.t. $x \not\in D_k$. We say that image $y$ causes a bad attach if either there is a collision in intermediate pairs, or if the newly created intermediate pairs attach to a value in $D_{k'}$. Formally, for an output value $y$, define the intermediate pairs of $(x, y)$ as \begin{align*}
        \mathsf{IP}_{(x, y)}(D_k, D_{k'}) =& \{ (x_i, z_i) \, : \, \exists z \in \bit^c \text{ s.t. } |\mathsf{tail}(z)| \geq 1, x_i \Vert z_i = \pi(x \Vert z \oplus y)\}.
    \end{align*}
    We say that $y$ causes a bad attach (on preimage $x$) if either of the following hold: \begin{enumerate}[label=(\arabic*)]
        \item There exists $(x_1, z_1) \in \mathsf{IP}_{(x, y)}(D_k, D_{k'})$ and $(x_2, z_2) \in \mathsf{IP}(D_k, D_{k'})$ such that $x_1 = x_2$.
        \item There exists $(x_1, z_1) \in \mathsf{IP}_{(x, y)}(D_k, D_{k'})$ such that $x_1 \in D_{k'}$.
    \end{enumerate}
\end{definition}

\begin{lemma}
    Fix databases $(D_k, D_{k'})$ which are good and have at most $t$ non-$\bot$ inputs each, and fix input $x \in \bit^r$ s.t. $x \not\in D_k$. Let $B$ be the set of images which causes a bad attach on preimage $x$. Then we have $|B| \leq O(t^3 n + t^3 2^{c-r})$ when $\pi$ is good.
    \label{lem:few-bad-attach}
\end{lemma}

\begin{proof}
    Let us denote by $X_{bad} \subset \bit^r$ the set of $x \in \bit^r$ such that either $x$ is the first half of an intermediate pair in $\mathsf{IP}(D_k, D_{k'})$, or $x \in D_{k'}$, or both. Note that $|X_{bad}| = O(t^2)$ by \Cref{lem:few-tails-ips}. Let us consider some $z \in \bit^c$ which has a tail, and $x_i \in X_{bad}$. We will use $N_{x_i, z}$ to count the number of images $y$ such that the value $$\pi(x \Vert y \oplus z)|_0^{r} = x_i,$$ i.e. we obtain a bad attach corresponding to $(x_i, z)$. If $y$ causes a bad attach, then it causes it for some $(x_i, z)$. We then have
    \begin{align*}
        |B| \leq& \sum_{z, \mathsf{tail}(z) \neq \emptyset} \sum_{x_i \in X_{bad}} N_{x_i, z}\\
        \leq& O(t^3 n + t^3 2^{c-r}) & \text{($\pi$ good)}
    \end{align*}
\end{proof}

The other case we must consider is a query to $k'$. Here, we will refer to a ``bad completion'' event.

\begin{definition}\label{def:badcomp}
    Fix  good databases $(D_k, D_{k'})$ which have at most $t$ non-$\bot$ inputs each, and fix input $x' \in \bit^r$ s.t. $x' \not\in D_{k'}$. We say that image $y'$ causes a bad completion if $x'$ appears in an intermediate pair $(x', y_i')$, and further there is either a created collision in intermediate pairs or state values with a tail, or if the newly created intermediate pairs attach to a value in $D_{k'}$. Formally, for an output value $y'$, define the completed pairs of $(x', y')$ as \begin{align*}
        \mathsf{CP}_{(x', y')}(D_k, D_{k'}) =& \begin{cases} \{ (x_i, z_i) \, : \, \exists x \in D_k \text{ s.t. } x_i \Vert z_i = \pi(x \Vert y_i' \oplus y' \oplus D_k(x))\} & \text{(If $\exists (x', y_i') \in \mathsf{IP}(D_k, D_{k'})$)} \\
        \emptyset & \text{(Otherwise)}
        \end{cases}
    \end{align*}
    We say that $y$ causes a bad attach (on preimage $x$) if either of the following hold: \begin{enumerate}[label=(\arabic*)]
        \item There exists $(x_1, z_1) \in \mathsf{CP}_{(x', y')}(D_k, D_{k'})$ and $(x_2, z_2) \in \mathsf{IP}(D_k, D_{k'})$ such that $x_1 = x_2$.
        \item There exists $(x_1, z_1) \in \mathsf{CP}_{(x', y')}(D_k, D_{k'})$ such that $x_1 \in D_{k'}$.
        \item There exists $(x', y_i') \in \mathsf{IP}(D_k, D_{k'})$ s.t. $\mathsf{tail}(y_i' \oplus y') \neq \emptyset$ (where here the tail is computed before assigning $[x'\rightarrow y']$).
    \end{enumerate}
\end{definition}

\begin{lemma}
    Fix databases $(D_k, D_{k'})$ which are good and have at most $t$ non-$\bot$ inputs each, and fix input $x' \in \bit^r$ s.t. $x' \not\in D_{k'}$. Let $B$ be the set of images which causes a bad completion on preimage $x'$. Then we have $|B| \leq O(t^3 n + t^3 2^{c-r})$ when $\pi$ is good.
    \label{lem:few-bad-comps}
\end{lemma}

\begin{proof}
    If there is not a $z_i' \in \bit^c$ such that $(x', z_i') \in \mathsf{IP}(D_k^t, D_{k'}^t)$, then the claim trivially holds as the completed pairs set is empty and no new value has a tail. If there is such a $z_i' \in \bit^c$, then the value $z_i' \oplus y'$ will now have a tail on image $y'$. This may create a bad completion if this value already has a tail, of which there are $O(t)$ possible values by \Cref{lem:few-tails-ips}. This gives $O(t)$ values of $y'$ that are in $B$. For the values of $y'$ where this does not happen, we then analyze the completed set.
    
    As in the proof of \Cref{lem:few-bad-attach}, let us denote by $X_{bad} \subset \bit^r$ the set of $x \in \bit^r$ such that either $x$ is the first half of an intermediate pair in $\mathsf{IP}(D_k, D_{k'})$, or $x \in D_{k'}$, or both. Note that $|X_{bad}| = O(t^2)$ by \Cref{lem:few-tails-ips}. Let us consider some $x \in D_k$ and $x_i \in X_{bad}$. We will use $N_{x_i, x}$ to count the number of images $y'$ such that the value $$\pi(x \Vert z_i' \oplus y' \oplus D_k(x))|_0^{r} = x_i,$$ i.e. we obtain a bad completion. If $y'$ causes a bad completion, then it causes it for some $(x_i, x)$. We then have
    \begin{align*}
        |B| \leq& \sum_{x \in D_k} \sum_{x_i \in X_{bad}} N_{x_i, x}\\
        \leq& O(t^3 n + t^3 2^{c-r}) & \text{($\pi$ good)}
    \end{align*}
\end{proof}

Putting these results together, we can construct a local property which will enable us to bound the transition to a bad database from a good one.

\begin{definition}
    Fix databases $D = (D_k, D_{k'})$, and let $x$ be an input to either $k$ s.t. $x \not\in D_k$, or an input to $k'$ s.t. $x \not\in D_{k'}$. Define the local property \begin{align*}
        \L^{x, D}_g =& \begin{cases}
            \{(D'_k[x \rightarrow y], D'_{k'}) \, : \, (D'_k, D'_{k'}) \in (D_k|_x, D_{k'}), &y \in \bit^c \cup \{\bot\} \text{ and } D'_k[x\rightarrow y], D'_{k'} \text{ is good}\} \\ & \text{(if $x$ input to $k$)} \\
            \{(D'_k, D'_{k'}[x \rightarrow y]) \, : \, (D'_k, D'_{k'}) \in (D_k|_x, D_{k'}), &y \in \bit^c \cup \{\bot\} \text{ and } D'_k, D'_{k'}[x\rightarrow y] \text{ is good}\}\\ & \text{(if $x$ input to $k'$)} 
        \end{cases}
    \end{align*}
    and let $\L_g$ be the family of such local properties; in other words $\L_g$ is the family of local properties which recognize good databases. We use the notation $\L_g^t$ to denote the restriction of $\L_g$ to databases with at most $t$ inputs each.
\end{definition}

\begin{lemma}\label{lem:local-bad}
    The local property $\L_g^t$ satisfies 
    $\Delta(\L_g^t) \leq O(nt^3 + 2^{c-r} t^3).$
    \label{lem:bad-transition-bound}
\end{lemma}

\begin{proof}
    It suffices to bound the distance for any of the constituent $\L_g^{x, D, t}$. Observe that if $D[x\rightarrow \bot]$ (where $x$ is either an input to $k$ or $k'$) is bad, then the property becomes trivial as the bad predicate is monotone, and $\Delta(\L_g^{x, D, t}) = 0$. We therefore assume $D[x\rightarrow \bot]$ is good. Let us split into cases. \begin{enumerate}[label=(\arabic*)]
        \item Suppose that $x$ is an input to $D_k$. In order for the database to become bad by assigning a value to $x$, the image of $x$ must cause a collision in the intermediate pair prefixes, or result in a new intermediate pair whose prefix already lies in $D_{k'}$. In other words, the image $y$ must cause a bad attach. By \Cref{lem:few-bad-attach}, there are $O(nt^3 + 2^{c-r} t^3)$, which implies \begin{align*}
            \Delta(\L_g^{x, D, t}) \leq& O(nt^3 + 2^{c-r} t^3).
        \end{align*}
        \item Suppose that $x$ is an input to $D_{k'}$. In order for the database to become bad by assigning a value to $x$, the image of $x$ must cause a collision in the intermediate pair prefixes or tail values, or result in a new intermediate pair whose prefix already lies in $D_{k'}$. In other words, the image $y$ must cause a bad completion. By \Cref{lem:few-bad-comps}, there are $O(nt^3 + 2^{c-r} t^3)$, which implies \begin{align*}
            \Delta(\L_g^{x, D, t}) \leq& O(nt^3 + 2^{c-r} t^3).
        \end{align*}
    \end{enumerate}
\end{proof}

Finally, the following lemma will be helpful in our indifferentiability proof, and is based on a similar analysis. It essentially states that, for a fixed input $x$ and state value $z$, the number of images of $x$ which cause the tail of $z$ to be different from when the image of $x$ is $\bot$ is small.

\begin{lemma}
    Fix databases $D = (D_k, D_{k'})$ which are good, let $x$ be an input to either $k$ s.t. $x \not\in D_k$, or an input to $k'$ s.t. $x \not\in D_{k'}$, and let $z \in \bit^c$ be some possible state. Denote by $\mathsf{tail}_1(z)$ the lexicographically first tail of $z$ in $D$ (or a failure symbol $\bot$ if $z$ has no tail), and $\mathsf{tail}_{1,[x\rightarrow y]}(z)$ as the lexicographically first tail of $z$ in $D[x\rightarrow y]$ (or again a failure symbol $\bot$ if $z$ has no tail). Then we have \begin{align*}
        | \{ y \, : \, \mathsf{tail}_{1,[x\rightarrow y]}(z) \neq \mathsf{tail}_1(z)\} | \leq& O(nt^3 + 2^{c-r} t^3).
    \end{align*}
    \label{lem:tail-insensitive}
\end{lemma}

\begin{proof}
    The case where $\mathsf{tail}_1(z) \neq \emptyset$ is trivial, as it means that $z$ has a tail in $D_k$ and $D_{k'}$. This tail will remain valid if a new tail of $z$ involving $x$ is created, which must happen if $\mathsf{tail}_{1,[x\rightarrow y]}(z) \neq \mathsf{tail}_1(z)$. There is then a $z$ value with multiple tails, and so the database is bad; the claim now follows from \Cref{lem:bad-transition-bound}. Let us instead focus on the case where $\mathsf{tail}_{1}(z) = t \neq \emptyset$. We have two cases. \begin{enumerate}[label=(Case \arabic*)]
        \item $x$ is an input to $k$. The only image assignments to $x$ which could possibly involve $x$ in a tail are ones where the prefix of one of the resultant intermediate pairs appears in $D_{k'}$, which is a bad attach. From \Cref{lem:few-bad-attach} there are $O(nt^3 + t^3 2^{c-r})$ values.
        \item $x$ is an input to $k'$. In the case where there is no $z'$ such that $(x, z') \in \mathsf{IP}(D_k, D_{k'})$, no new tail is created for any image assignment, so the claim holds trivially. Suppose the opposite, and let us exclude the values of $y$ which lead to a bad completion; by \Cref{lem:few-bad-comps} there are $O(nt^3 + t^3 2^{c-r})$ values. There is then at most one new value of $z$ with a tail (recalling that intermediate pair prefixes are unique in a good database), which in particular has value $z' \oplus y$. This can only change the tail of $z$ if $y = z' \oplus z$, which is a single value.
    \end{enumerate}
\end{proof}

\newpage
\section{Query lower bound proofs}
\label{sec:bounds-proofs}

In this section, we derive bounds for collision resistance and preimage resistance of the sponge construction. We start by re-deriving classical lower bounds for these problems within our framework, using lazy sampling arguments. We then use the tools developed along the way to give quantum query lower bounds (in \cref{subsec:QQB}) via the framework of \cite{Zhandry2018,chung2021compressed}. We will show lower bounds for preimage and collision resistance in the Msponge, which implies via black-box reduction preimage and collision resistance of the standard sponge.

\subsection{Classical lower bounds}

Consider a classical adversary with query access to $\varphi$ and $\varphi^{-1}$. The probability that this adversary creates a bad database with fewer than $q$ queries is upper bounded by \Cref{lem:bad-transition-bound}. Now consider the probability that this same adversary outputs a collision in $\Sp^\varphi$, and separate this event into two disjoint events: one where the database is good, and one where it is bad. We derive a bound for the probability of outputting a collision conditioned on the database being good, as this considerably simplifies the analysis. We take a similar approach to preimage finding. One downside of this strategy is that it does not yield a tight bound for either of the considered problems. 

We begin by defining a notion of reachable outputs given some database $D = \{D_{k},D_{k'},D_{h}\}$ and permutation $\pi$. Under this notion, an output $y \in \bit^r$ is reachable if there's a collection of input-output pairs in the databases $D$ that, when put together with $\pi$ in the manner defined by the Msponge, would yield $y$ as an output. In particular, this implies that $D$ determines an input $m$ such that $\MSp(m) = y$. Note that, if an output $y$ is reachable in this sense, then one can use $D$ and block-length-many queries to $h$ to construct another input $m'$ such that $y = \Sp^\varphi(m')$.
\begin{definition}[reachable output] \label{def:reachable-output}
Let $y\in \bit ^r$. We say that $y$ is a \textit{reachable output} of $\Sp^\varphi$ in $D$ if there exists $z\in \bit^c$ such that \begin{enumerate}[label=(\arabic*)]
    \item $(z,h(z))\in D_{h},$
    \item \label{cond:reachable-2} There exists a tail $x$ of $z$\text{ such that } $y=\mathsf{head}(x) \oplus h(z)$.
\end{enumerate}
We refer to the tail-capacity pair $(x,z)$ as a path. We say that the path $(x,z)$ reaches $y$.
\end{definition}
Condition \ref{cond:reachable-2} implicitly means that every $k$ and $k'$ query appearing when we feed $\Sp^\varphi$ with input $\mathsf{fix}(x)$ in the framework of \Cref{game:ideal-permutation-world} must be in the databases $D_k$ and $D_{k'}$, respectively. A path is an analog of input for $\Sp^\varphi$. Also, note that the tail $x$ can appear in the path of one output at most, but an output might be reachable via several paths each having a distinct tail $x$.
\begin{lemma} \label{lem:input-path}
    Let $\mathcal{A}^{\varphi,\varphi^{-1}}$ be a (classical or quantum) $q$-query algorithm outputting colliding inputs $m,m'\in (\bit^r)^{\leq l}$ for $\Sp^\varphi$ with probability $\epsilon$. Then there exists a $(3q+6l)$-query algorithm $\mathcal{B}^{h,k,k'}(\pi)$ outputting colliding paths $s,s'$ for $\Sp^{h k' \pi k}$ with the same probability, i.e., 
    \begin{equation}\label{eq:input-path} 
        \underset{ m,m'\leftarrow \mathcal{A}^{\varphi,\varphi^{-1}}}{\Pr}[\Sp^{\varphi}(m) = \Sp^\varphi(m')]
        =\underset{s,s'\leftarrow \mathcal{B}^{h,k,k'}(\pi)}{\Pr}[s,s' \text{ reach the same output}].
    \end{equation} 
\end{lemma}
The probabilities in \Cref{eq:input-path} are taken over the appropriate sampling of the oracles, i.e., uniformly random $\varphi$ for the left-hand side and uniformly random $h, k, k', \pi$ for the right-hand side.
\begin{proof}
     The algorithm $\mathcal{B}$ will begin by running $\mathcal{A}$ as a subroutine, simulating the oracles $\varphi, \varphi^{-1}$ using $h, k', k, \pi$ until $\mathcal{A}$ terminates and outputs $m, m'$. Let $y := \Sp^{h k' \pi  k}(m)$ and $y' := \Sp^{h k' \pi  k}(m')$. By the indifferentiability result in \Cref{lem:alternate-perm}, the probability that $y=y'$ after this simulated run of $\algo A$ is exactly equal to 
     \begin{equation}\label{eq:sim-coll-success}
         \underset{ m,m'\leftarrow \mathcal{A}^{\varphi,\varphi^{-1}}}{\Pr}[\Sp^{\varphi}(m) = \Sp^\varphi(m')]
     \end{equation}
     where $\varphi$ is sampled uniformly at random.
     
     The above simulation of $\algo A$ by $\algo B$ works as follows. When $\mathcal{A}$ queries on $x$, $\mathcal{B}$ queries on
     \begin{enumerate}
         \item $x|_0^r$ to $k$,
         \item $x'=\pi(x|_0^r||(x|_r^n\oplus k(x|_0^r)))|_0^r$ to $k'$, and
         \item $z=\pi(x|_0^r||(x|_r^n\oplus k(x|_0^r)))_r^n \oplus k'(x')$ to $h$,
     \end{enumerate}
    then returns $(x'\oplus h(z))||z$. Each query of $\mathcal{A}$ costs 3 queries for $\mathcal{B}$. After $q$ queries, $\mathcal{A}$ terminates and outputs $m,m'$. 
    
    Next, algorithm $\mathcal{B}$ makes the necessary queries to form the paths $s, s'$ that reach $y,y'$ (respectively) with the following procedure. We will WLOG write the set of necessary queries for the case of $m$ explicitly, assuming the block size of $m$ is bounded by $l$. The relevant query input-output set is then
    \begin{align} \label{query-set-for-message-m}
        \{(x_i,k(x_i)),(x'_i,k'(x'_i)),(z_i,h(z_i))\}_{i\in \{1,\dots,l\}}
    \end{align}
    such that 
    \begin{align*} 
        x_1:= m|_0^r &&x'_1:=\pi(x_1||k(x_1))|_0^r && z_1:=\pi(x_1||k(x_1))|_r^n \oplus k'(x'_1) \\
        x_2:= m|_r^{2r} \oplus x'_1 \oplus h(z_1)&&x'_2:=\pi(x_2||k(x_2)\oplus z_1)|_0^r && z_2:=\pi(x_2||k(x_2))|_r^n \oplus k'(x'_2) \\
       \vdots &&\vdots&&\vdots\\
       x_l:= m|_{(l-1)r}^{lr} \oplus x_{l-1}'
 \oplus h(z_{l-1}') &&x'_l:=\pi(x_l||k(x_l)\oplus z_{l-1})|_0^r &&z_l:=\pi(x_l||k(x_l)|_r^n \oplus k'(x'_l) .
    \end{align*} 
    This procedure costs at most $6l$ many queries for $\mathcal{B}$. Note that these $6l$ queries made by $\mathcal{B}$ are classical even if the queries of $\algo A$ to its oracle are quantum. The path that reaches $y$ becomes $s=(x_1||\cdots||x_l,z_l)$. The path $s'$ that reaches $y'$ is defined similarly. $\algo B$ completes by outputting $s, s'$ and terminating. 
    
    By construction, $s$ and $s'$ reach the same output if and only if $y=y'$. As observed above, the latter occurs with probability exactly equal to \Cref{eq:sim-coll-success}.
\end{proof}

\begin{remark} \label{remark:outputting-col-containing-col}
The proof of \cref{lem:input-path} shows that the success probability of any classical  $q$-query adversary $\mathcal{A}^{\varphi,\varphi^{-1}}$ outputting a collision is equal to the probability that, after the classical $3q+6l$-query adversary $\mathcal{B}^{h,k,k'}(\pi)$ finishes, the database contains a collision. 
\end{remark}

It's straightforward to adapt the construction of $\algo B$ in the proof of \cref{lem:input-path} so that, if $\algo A$ is some algorithm that outputs only a single $m$, then the result is a single path $s$ that reaches the output of the sponge evaluated on $m$. This yields the following.

\begin{lemma}\label{lem:input-path-for-OW}
    Let $\mathcal{A}^{\varphi,\varphi^{-1}}$ be a (classical or quantum) $q$-query algorithm that, given $y\sim \bit^r$, outputs $m\in (\bit^r)^{\leq l}$ such that $y=\Sp^\varphi(m)$ with probability $\epsilon$. Then there exists a $(3q+3l)$-query algorithm $\mathcal{B}^{h,k,k'}(\pi)$ that, given $y\sim \bit^r$, outputs a path $s$ that reaches $y$ for $\Sp^{h k'  \pi  k}$ with the same probability, i.e., 
    \begin{align*} 
    \Pr_{\substack{y \sim \bit^r \\ m'\leftarrow \mathcal{A}^{\varphi,\varphi^{-1}}}}
       [y = \Sp^\varphi(m')] 
       = \Pr_{\substack{y \sim \bit^r \\ s\leftarrow \mathcal{B}^{h,k,k'}(\pi)}}[s \text{ reaches y}]
    \end{align*}  
\end{lemma}

When we consider the framework of \Cref{game:ideal-permutation-world}, we will refer to the existence of paths reaching the same output as a collision. We will derive the collision resistance bound for $\Sp^\varphi$ in \Cref{thm:classical-CR} by combining \Cref{lem:input-path} and the bound for the database containing a collision as in \Cref{lem:collision-in-db}.

\begin{definition}
    We say that there is a collision in the database (in the framework of \Cref{game:ideal-permutation-world}) if at least two paths reach the same output.
    \end{definition}

\begin{lemma}[collision in the database] \label{lem:collision-in-db} Given classical query access to $h,k,k'$, the probability of the database containing a collision with bounded block length $l$ after $q$ queries is bounded by $O(q^4n2^{-\min(r,c)}).$ 
    Consequently, for constant success probability, the number of queries required is $\Omega(\sqrt[4]{2^{\min(r,c)}/n})$.
\end{lemma}
\begin{proof}
Let $D_{k}, D_{k'}, D_{h}$, denoted by $D$, be a set of databases recording queries to functions $k,k',h$, respectively. Define the event \begin{align*}
        \CC_t :=& \exists \text{ a collision in } D^t
    \end{align*} where $D^t$ is a database such that the total number of defined input/output points is at most $t$. Then, $\Pr[\CC_q]$ denotes the probability of a database containing a collision after $q$ queries.
 We separate the event into the cases that database is \textit{good}, as in \Cref{def:good-1}, and \textit{bad} as follows \begin{align}
        \Pr[\CC_q] &= \Pr[\CC_q \cap D^{q-1} \text{ good}] + \Pr[\CC_q\cap D^{q-1} \text{ bad}] \label{good-bad}.
    \end{align}
We will work with good databases after bounding the second term in \Cref{good-bad} by \begin{align*}
       \Pr[\CC_q\cap D^{q-1} \text{ bad}]
       &\leq \Pr[D^{q} \text{ bad}] \\ 
       &\leq \sum_{t=1}^q \Pr[D^t \text{ bad}|D^{t-1} \text{ good}]+\Pr[D^t \text{ bad}|D^{t-1} \text{ good}]\\
       &\leq O(nq^4 2^{-\min(r,c)})
     \end{align*}
     where the last inequality follows from Lemma \ref{lem:bad-transition-bound}.
Observe that $D^t$ good implies $D^{t-1}$ good and $\CC_t$ implies $\CC_{t-1}$. Using these, we can bound the first term in \Cref{good-bad} by
    \begin{align*}
    \Pr[\CC_q \cap D^{q-1} \text{ good}] &=\Pr[\CC_{q-1} \cap \CC_q \cap D^{q-1} \text{ good} ] + \Pr[\neg \CC_{q-1} \cap \CC_q \cap D^{q-1} \text{ good}]\\
    &\leq \Pr[\CC_{q-1} \cap D^{q-2} \text{ good}] + \Pr[\CC_q|  \neg \CC_{q-1}, D^{q-1} \text{ good}]\\
     &\leq \sum_{t=1}^q \Pr[\CC_{t} | \neg \CC_{t-1},D^{t-1} \text{ good}]\end{align*} 
    where the second inequality followed from the observation and the third from repeatedly applying the procedure in the first two lines.
    
Notice that adding a collision to the database with a single query, assuming there is none, requires adding a new path. This can be done by either
\begin{itemize}
    \item querying to $k$ or $k'$ to create a new tail $x\in(\bit^r)^*$ for some $z\in \bit^c$, or
    \item querying to $h$ on $z\in \bit^c$ with $\mathsf{tail}(z)\not=\emptyset$.
\end{itemize}
These are just necessary, but not sufficient, conditions. In the $t$-th query, we have three options: the adversary can query to the function $k,k'$, or $h$. We consider these events separately; the event "$f$ query" means the queried function is $f\in\{k,k',h \}$.
\begin{itemize}
    \item [(1)] A query to $k$ on input $x$. If a collision occurs after this query, a new path, in particular a new tail, must have been created in a round starting with the state $x||z$ where $z\in\bit^c$ and $\mathsf{tail}(z)\not=\emptyset$. This means the $k'$ query in this round must be in the database, i.e., $\pi(x||k(x)\oplus z)|_0^r\in D_{k'}^{t-1}$. Then, with the union bound and the fact that $|\mathsf{tail}(z)|\leq O(t)$ in a good database, we get 
    \begin{align*}
    \Pr[\CC_{t} |k \text{ query}, \neg E_{t-1},D^{t-1}
    \text{ good}] & \leq \sum_{z:\mathsf{tail}(z)\not=\emptyset}  \sum_{x'\in D_{k'}^{t-1}} \underset{k(x)}{\Pr}[\pi(x||k(x)\oplus z)|_0^r=x'] \\
    & \leq \sum_{z:\mathsf{tail}(z)\not=\emptyset}  \sum_{x'\in D_{k'}^{t-1}} N_{x_i,z} 2^{-c} \\
    & \leq O(t^2 2^{-r} + t^2n2^{-c}) 
    \end{align*} where $N_{x_i,z}$ is the number of values $k(x)$ such that $\pi(x||k(x)\oplus z)|_0^r=x'$.

    \item [(2)] A query to $k'$ on input $x'$. Similar to the previous one, creating a collision requires creating a new tail, which must have been created by querying the first entry of an element in the set of intermediate pairs $\mathsf{IP}$. Since all the first entries in $\mathsf{IP}$ are distinct, we can extend only one tail. Say the input to be queried is $\pi(x||k(x)\oplus z)|_0^r=x'$ for some $z\in\bit^c$ with $\mathsf{tail}(z)\not=\emptyset$ and $(x,k(x))\in D_k^{t-1}$. After this query, the value $\bar{z}=\pi(x||k(x)\oplus z)|_r^n \oplus k'(x')$ will have a tail. We can separate the collision event into disjoint events while skipping writing the conditions on the database for simplicity:
    \begin{align*}
        \Pr[\CC_t]= \Pr[\CC_t \cap \mathsf{tail}(\bar{z})\not = \emptyset \text{ in } D^{t-1}] + \Pr[\CC_t \cap \mathsf{tail}(\bar{z})= \emptyset \text{ in } D^{t-1}]
    \end{align*} If $\bar{z}$ has a tail before this query, then collision is inevitable. So, \begin{align*}
     \Pr[\CC_t \cap \mathsf{tail}(\bar{z})\not = \emptyset \text{ in } D^{t-1}] =\Pr[ \mathsf{tail}(\bar{z})\not = \emptyset \text{ in } D^{t-1}]\leq O(t2^{-c})
    \end{align*} since $\bar{z}$ is uniformly random in 
$\bit^c$ and number of $z$ with a tail is $O(t)$ in $D^{t-1}$. If $\bar{z}$ does not have a tail before this query, then having a collision implies \begin{enumerate}
    \item \label{k'-col-case1} $\bar{z}\in D_h^{t-1}$ and $x'\oplus h(\bar{z})$ collides with a reachable output via $D^{t-1}$, or
    \item \label{k'-col-case2} $k'$-query following a round ending with $\bar{z}$ is in $D_{k'}^{t-1}$.
\end{enumerate}
Case \ref{k'-col-case1} is upper bounded by $O(t2^{-c})$ due to the uniformity of $\bar{z}$. Case \ref{k'-col-case2} can be upper bounded as \begin{align*}
     \sum_{\bar{x}\in D_k^{t-1}} \sum_{\bar{x}'\in D_{k'}^{t-1}} \underset{k'(x')}{\Pr}[\pi(\bar{x}||k(\bar{x}) \oplus \bar{z})|_0^r=\bar{x}'] &\leq \sum_{\bar{x}\in D_k^{t-1}} \sum_{\bar{x}'\in D_{k'}^{t-1}} N_{\bar{x}',\bar{z}}2^{-c}\\
     &\leq O(t^2 2^{-r} + t^2n2^{-c}) 
\end{align*} 
where $N_{\bar{x}',\bar{z}}$ is the number of values $k'(x')$ such that $\pi(\bar{x}||k(\bar{x})\oplus \bar{z})|_0^r=\bar{x}'$.
\item [(3)] A query to $h$ on input $z$. Let $R$ denote the set of reachable outputs via $D^{t-1}$, then $|R|\leq O(t)$ since $D^{t-1}$ is good. A collision can be formed if the query $z$ has a tail $x$ in $D^{t-1}$ and $x,z$ constructs a path for an element in $R$. So,
    \begin{align*}
    \Pr[\CC_{t} |h \text{ query}, \neg E_{t-1},D^{t-1} \text{ good}] & \leq \sum_{y\in R} \underset{h(z)}{\Pr}[y=x'_z \oplus h(z)] \\
    & \leq O(t 2^{-c})
    \end{align*} where $x'_z$ head of the tail of $z$.    
\end{itemize}
\end{proof}
This is our main theorem for the classical collision resistance of the sponge construction. This will be the main ingredient while proving the quantum collision resistance using quantum transition capacities.

The following result is obtained immediately from the proof of \Cref{lem:collision-in-db} since the bound for $k,k'$ query cases are argued by the bound for creating a new tail, which is also required to create a path reaching some output conditioned on there is no path for this image. Also, the bound for the $h$ query case is argued by querying a capacity value with a tail, which is again necessary to create a path reaching some output. The success probability of the $h$ query case is less than the one in \Cref{lem:collision-in-db}. 
\begin{corollary}\label{cor:preimage-in-db} Given $y\sim \bit^r$ and classical query access to $h,k,k'$, the probability of the database containing $m$ with bounded block length $l$ such that $y=\Sp^\varphi(m)$ after $q$ queries is bounded by $O(q^4n2^{-\min(r,c)}).$ 
    Consequently, for constant success probability, the number of queries required is $\Omega(\sqrt[4]{2^{\min(r,c)}/n})$.
\end{corollary}
\begin{theorem}[classical collision resistance]\label{thm:classical-CR}
Given classical query access to $\varphi$ and $\varphi^{-1}$, the probability of $q$-query adversary outputting a collision, i.e., an $m, m' \in ({\bit}^r)^{\leq l}$ for some $l\le q$ such that $m \neq m'$ and $\Sp^{\varphi}(m) = \Sp^\varphi(m')$  is bounded by
 \begin{align*}
        \underset{ m,m'\leftarrow \mathcal{A}^{\varphi,\varphi^{-1}}}{\Pr}[\Sp^{\varphi}(m) = \Sp^\varphi(m')]\le O(q^4n2^{-\min(r,c)}).
   \end{align*}
    Consequently, for constant success probability, the number of queries required is $\Omega(\sqrt[4]{2^{\min(r,c)}/n})$.
\end{theorem}
\begin{proof}
We will consider adversaries in the framework of \Cref{game:ideal-permutation-world}, which by \Cref{lem:alternate-perm} suffices. By \Cref{lem:input-path}, for each adversary $\mathcal{A}^{\varphi,\varphi^{-1}}$, there exists an adversary $\mathcal{B}^{h,k,k'}$ such that the probabilities of $\mathcal{A}$ outputting a colliding inputs with $q$ queries and $\mathcal{B}$ outputting colliding paths with $3q+6l$ queries are the same. Moreover, as $l\leq q $ by assumption, and using \Cref{lem:collision-in-db}, we get the desired bound. 
\end{proof}
\begin{corollary}\label{cor:classical-OW}
Given $y\sim\bit^r$ and classical query access to $\varphi$ and $\varphi^{-1}$, the probability of $q$-query adversary outputting $m \in ({\bit}^r)^{\leq l}$ such that $y=\Sp^\varphi(m)$ and some $l\leq q$, namely \begin{align*}
         \Pr_{\substack{y \sim \bit^r \\ m'\leftarrow \mathcal{A}^{\varphi,\varphi^{-1}}}}[y =Sp^\varphi(m')] 
    \end{align*}
    is bounded by $O(q^4n2^{-\min(r,c)}).$ 
    Consequently, for constant success probability, the number of queries required is $\Omega(\sqrt[4]{2^{\min(r,c)}/n})$.
\end{corollary}

\subsection{Quantum lower bounds}\label{subsec:QQB}

With the result in \Cref{thm:classical-CR} and the notion of quantum transition capacities, we can prove the quantum collision resistance of the sponge construction.

By definition, the value $[\![ \neg \P \xrightarrow{q} \Q]\!]$ is equal to the square-root of the maximal probability that the internal state of the compressed oracle, when supported only on databases $D \in \neg \P$, is measured to be in a database $D' \in \Q$ after a quantum query algorithm performs $q$ sequential queries. The special case $[\![ \emptyset \xrightarrow{q} \Q]\!]$ is the square-root of the maximal probability of $D$ satisfying $\Q$ when $D$ is obtained by measuring the internal state of the compressed oracle after the interaction with $\mathcal{A}$, maximized over all $q$-query quantum algorithms $\mathcal{A}$, i.e., 
\begin{align} \label{transition-to-P}
[\![ \emptyset \xrightarrow{q} \Q]\!]:= \underset{\mathcal{A}}{\max} \sqrt{\Pr[D\in \Q]}.     
\end{align}
Before deriving the quantum collision resistance proof, we need a bound for having a bad database after the oracle calls since we always separate the good and bad databases. 
\begin{theorem}\label{thm:quantum-bad-bound}
    Define the following predicate 
    \begin{align*}
   \mathrm{BAD}= &\Big\{\left(D\in\mathfrak{D}|\exists z \in \bit^c: |\mathsf{tail}^D(z)| \geq 2 \right)\\
   &\quad \vee\left(\exists (x_1,z_1),(x_2,z_2)\in \mathsf{IP}^D:x_1=x_2,z_1\neq z_2 \right)\Big\} 
   \end{align*}
    where $D$ in the superscripts indicates the database the notion is defined for. Then, the probability $p'$ that algorithm $\mathcal{A}$ creates a bad database (i.e., applying the binary measurement defined by $\Pi^\mathrm{BAD}$ returns $\mathrm{BAD}$) after $q$ queries to the compressed oracle is bounded as
    \begin{align*}
    p' \leq [\![ \emptyset \xrightarrow{q} \mathrm{BAD}]\!]^2 \leq O(q^5n2^{-\min(r,c)}).
    \end{align*}
\end{theorem}
\begin{proof}
By definition of the quantum transition capacity, we have $\sqrt{p'} \leq [\![ \emptyset \xrightarrow{q} \mathrm{BAD}]\!]$. Then, using its properties, we get
\begin{align*}
[\![ \emptyset \xrightarrow{q} \mathrm{BAD}]\!]&\leq [\![ \emptyset \xrightarrow{q} \mathrm{BAD} \cup \neg D^q]\!] &\text{(\Cref{quant-trans-cap-union})}\\
&\leq \sum_{t=1}^q  [\![  \mathrm{GOOD} \cap \D^{t-1}\xrightarrow{} \mathrm{BAD} \cup \neg D^t]\!] &\text{(\Cref{quant-trans-cap-splitting-queries})}\\
&\leq \sum_{t=1}^q  [\![  \mathrm{GOOD} \cap \D^{t-1}\xrightarrow{} \mathrm{BAD} ]\!] &
\end{align*}where the last inequality follows from the fact that a single query cannot increase the size of the database by more than one.

We can bound $[\![ \mathrm{GOOD} \cap D^{t-1}\xrightarrow{} \mathrm{BAD} ]\!]$ using the classical reasoning. By setting $\L^{x,D}:=\mathrm{BAD}|_{D|^x}$ and applying \Cref{lem:transition-from-classical}, we get
 \begin{align*}
 [\![ \mathrm{GOOD} \cap D^{t-1}\xrightarrow{} \mathrm{BAD} ]\!] &\leq 4\sqrt{ \Delta(\L) / N} 
\end{align*} where $\L \coloneqq \bigcup_{x, D} \{\L^{x, D} \, : \, x \in \mathcal X, D \in \mathfrak D\}$.

Notice that  $\L$ corresponds to $\L_g^t$ in \Cref{lem:local-bad}, hence $ \Delta(\L) \leq O(t^3n + t^3 2^{c-r})$ and  \begin{align*}
   \Delta(\L) / N \leq O((t^3n + t^3 2^{c-r})/2^c ) \leq O\left(t^3n2^{-\min(r,c)}\right)
\end{align*} where $2^c$ in denominator appears because the non-trivial queries $x$ for $\L$ are either $k$ or $k'$ queries and their output length is $c$.
\end{proof}

\begin{theorem}\label{thm:quantum-collision}
    Define the following predicate
    $$
    \mathrm{C}= \{D\in\mathfrak{D}|\exists s,s' \in (\bit^{\leq rl}, \bit^c): s\neq s' \text{ and } s,s'\text{ reaches the same output} \},
     $$
     i.e., it is the collection of databases that contains the necessary queries that can construct at least two distinct paths reaching the same output. Then, the probability $p'$ that algorithm $\mathcal{B}$ creates a database in $\mathrm{C}$ (i.e., applying the binary measurement defined by $\Pi^\mathrm{C}$ returns $\mathrm{C}$) after $q$ queries to the compressed oracle is bounded as
    \begin{align*}
    p' \leq [\![ \emptyset \xrightarrow{q} \mathrm{C}]\!]^2 \leq O\left(q^5n2^{-\min(r,c)}\right).
    \end{align*}
\end{theorem}
\begin{proof}
The maximality in the definition of the quantum transition capacity implies $$\sqrt{p'}\leq [\![ \emptyset \xrightarrow{q} \mathrm{C}]\!].$$ Next, we will use the properties of the quantum transition capacity to upper bound this as follows 
\begin{align*}
[\![ \emptyset \xrightarrow{q} \mathrm{C}]\!]&\leq [\![ \emptyset\xrightarrow{q} \mathrm{C} \cup \mathrm{BAD} \cup \neg D^{q}]\!]  &\text{(\Cref{quant-trans-cap-union})}\\
    &\leq \sum_{t=1}^q  [\![ \neg \P_{t-1} \xrightarrow{} \P_t]\!] &\text{(\Cref{quant-trans-cap-splitting-queries})}   
\end{align*} where $\P_0=\neg \bot$ and $\P_t=  \neg D^{t} \cup\mathrm{C} \cup \mathrm{BAD} $ for $t \in \{1, \dots, q\}$. Moreover, to make it more similar to the classical proofs for finding a collision in the sponge and having a bad database, we will do 
 \begin{align*}
  [\![ \neg \P_{t-1} \xrightarrow{} \P_t]\!] &=[\![ \neg \P_{t-1} \xrightarrow{} \mathrm{C} \cup \mathrm{BAD} ]\!]\\   
  &\leq [\![ \neg \P_{t-1} \xrightarrow{} \mathrm{C} ]\!] + [\![ \neg \P_{t-1} \xrightarrow{} \mathrm{BAD}]\!]\\
  &\leq [\![ \neg \P_{t-1} \xrightarrow{} \mathrm{C} ]\!] + [\![ D^{t-1} \cap \mathrm{GOOD} \xrightarrow{} \mathrm{BAD}]\!]
 \end{align*} where the first line follows from the fact that a single query cannot increase the size of the database by more than one and the rest follows from \Cref{quant-trans-cap-union}. 
 
 The latter term is already bounded as in \Cref{thm:quantum-bad-bound}. The former one, namely $[\![ D^{t-1}\cap \neg \mathrm{C} \cap \mathrm{GOOD} \xrightarrow{} \mathrm{C} ]\!] $,  will be bounded using the classical reasoning. Setting $\L^{x,D}:= \mathrm{C}|_{D|^x}$, applying \Cref{lem:transition-from-classical} and \Cref{thm:classical-CR}, we get
  \begin{align*}
[\![ D^{t-1}\cap \neg \mathrm{C} \cap \mathrm{GOOD} \xrightarrow{} \mathrm{C} ]\!] &\leq \sqrt{O\left(t^3n2^{-\min(r,c)}\right)}. 
 \end{align*}    
Combining all these gives \begin{align*}
    \sqrt{p'} &\leq [\![ \emptyset \xrightarrow{q} \mathrm{C}]\\
    &\leq \sum_{t=1}^q [\![ D^{t-1}\cap \neg \mathrm{C} \cap \mathrm{GOOD} \xrightarrow{} \mathrm{C} ]\!] + [\![ D^{t-1} \cap \mathrm{GOOD} \xrightarrow{} \mathrm{BAD}]\!]\\
    & \leq \sum_{t=1}^q \sqrt{O\left(t^3n2^{-\min(r,c)}\right)}\\
    &\leq q\sqrt{O\left(q^3n2^{-\min(r,c)}\right)}.
\end{align*}
\end{proof}
\begin{corollary}\label{cor:quantum-preimage}
    For $y\sim \bit^r$, define the following predicate 
    $$
    \P_y= \{D\in\mathfrak{D}|\exists s \in (\bit^{\leq rl}, \bit^c): s\text{ reaches the output } y \},
     $$
     i.e., the set of databases that contains the necessary queries that can construct a path that reaches $y$. Then, the probability $p'$ that algorithm $\mathcal{B}$ creates a database in $\P_y$ (i.e., applying the binary measurement defined by $\Pi^\mathrm{P}_y$ returns $\P_y$) after $q$ queries to the compressed oracle is bounded as
    \begin{align*}
    p' \leq [\![ \emptyset \xrightarrow{q} \P_y]\!]^2 \leq O\left(q^5n2^{-\min(r,c)}\right).
    \end{align*}
\end{corollary}
Now, we can give the corollary that bounds the success probability of an actual sponge collision finding adversary, i.e., one that outputs $m,m'\in\left(\{0,1\}^r\right)^{\leq l}$ such that $\Sp^\varphi(m)=\Sp^\varphi(m')$.
\begin{corollary}[quantum collision resistance] \label{cor:main-sponge-quantum CR}
The probability that a quantum algorithm $\mathcal{A}$ with quantum query access to a random permutation $\varphi \in S_{\{0,1\}^{r+c}}$ and its inverse, making at most a total of $q$ queries, returns $m,m'\in\left(\{0,1\}^r\right)^{\le l}$ for $l\le q$ 
such that $m\neq m'$ and  $\Sp^\varphi(m)=\Sp^\varphi(m')$, can be upper bounded as
\begin{align*}
\Pr_{\substack{\varphi \sim S_{\{0,1\}^{r+c}} \\ m,m'\leftarrow \mathcal{A}^{\varphi,\varphi^{-1}}}}
    [\Sp^\varphi(m)=\Sp^\varphi(m')]\leq O\left(q^5n2^{-\min(r,c)}\right).
\end{align*}
\end{corollary} 
\begin{proof}
    For each such algorithm $\mathcal{A}$ construct an algorithm $\mathcal{B}$ with query complexity $O(q)$ that simulates $\varphi,\varphi^{-1}$ to $\mathcal{A}$ as in \Cref{lem:input-path} and outputs paths $s,s'$ corresponding to the outputs $m,m'$ of $\mathcal{A}$, respectively. Suppose $m$ and $m'$ are $\tilde l\le l$ and $\tilde l'\le l$ blocks, respectively.
    More explicitly, a path $s$ is derived from the set of queries corresponding to $m$, namely 
    $$
    \{(x_i,k(x_i)),(x'_i,k'(x'_i)),(z_i,h(z_i))\}_{i\in \{1,\dots,\tilde l\}}
    $$
     as described in \Cref{lem:input-path}. Similarly, one can form another set of queries 
     $$
     \{(x_i,k(x_i)),(x'_i,k'(x'_i)),(z_i,h(z_i))\}_{i\in \{\tilde l+1,\dots,\tilde l+\tilde l'\}}
     $$
      that defines a path $s'$ corresponding to $m'$.  By \Cref{lem:input-path}, we get 
      \begin{align} \label{prob-p-quantum collision}
        \underset{ m,m'\leftarrow \mathcal{A}^{\varphi,\varphi^{-1}}}{\Pr}[Sp^{\varphi}(m) = Sp^\varphi(m')]=\underset{s,s'\leftarrow \mathcal{B}^{h,k,k'}}{\Pr}[s,s' \text{ reaches the same output}].
    \end{align}

By taking the union of the two sets of queries forming the paths $s$ and $s'$, we form $$\mathbf{x}=\{(x_i,k(x_i)),(x'_i,k'(x'_i)),(z_i,h(z_i))\}_{i\in \{1,\dots,\tilde l+\tilde l'\}}.$$ If the set $\mathbf{x}$ has a size smaller than $6l$ we can complete it set to have size $6l$ by adding random input-output pairs. Let $\mathcal B'$ be the slightly modified algorithm that outputs that $\mathbf x$.
    We define the following relation\footnote{We slightly abused the notation here. In \Cref{cor:fund-lemma-with-R} elements of the relation are pairs of input and output tuples of size $6l$. However, in the relation described here, for each element of $R$, the element obtained by applying the same permutation on the indices of both input and output tuples is also in $R$. Moreover, any element obtained this way satisfies the conditions in the description of the probabilities of $p$ and $p'$, we consider elements of $R$ modulo the permutation action on the indices.}
    \begin{align*}
        R=\{\mathbf{x} \in \mathcal{X}^{6l}  \times \mathcal{Y}^{6l}: \mathbf{x} \text{ contains the queries constructing at least two colliding paths} \}.
    \end{align*} 
    Note that
    \[
    \Pr_{\mathbf x\leftarrow{\mathcal{B}'}^{h,k,k'} }[\mathbf x\in R]=\underset{s,s'\leftarrow \mathcal{B}^{h,k,k'}}{\Pr}[s,s' \text{ reaches the same output}]\eqqcolon p.
    \]
    Applying \Cref{cor:fund-lemma-with-R}, we get    
    \begin{align*}
        \sqrt{p}\leq \sqrt{p'}+\sqrt{\frac{6l}{N}}
    \end{align*}
    with $p'$ as defined in \cref{cor:fund-lemma-with-R}. 
    We have the bound for $p'$ as $O(q\sqrt{q^3n2^{-\min(r,c)}})$ by \Cref{thm:quantum-collision}. Hence, \begin{align*}
        \underset{m,m'\leftarrow \mathcal{A}^{\varphi,\varphi^{-1}}}{\Pr}[\Sp^\varphi(m)=\Sp^\varphi(m')] \leq O(q^5n2^{-\min(r,c)}).
    \end{align*}
    \end{proof}
\begin{corollary}\label{cor:main-sponge-quantum-P}
Given $y\sim \bit^r$, the probability that a quantum algorithm $\mathcal{A}$ with quantum query access to a random permutation $\varphi \in S_{\{0,1\}^{r+c}}$ and its inverse, making at most a total of $q$ queries, returns $m\in\left(\{0,1\}^r\right)^{\leq l}$ for $l\leq q$ such that $y=\Sp^\varphi(m)$, can be upper bounded as
\begin{align*}
    \Pr_{\substack{y \sim \bit^r \\ m'\leftarrow \mathcal{A}^{\varphi,\varphi^{-1}}}}[y=\Sp^\varphi(m)]\leq O\left(q^5n2^{-\min(r,c)}\right).
\end{align*}
\end{corollary}

We now briefly sketch how to adapt the above approach to the case of preimage finding. The detailed derivation is omitted since it is very similar to collision finding, and the actual bound for both is dominated by the bad database predicate. The proof of \Cref{lem:input-path-for-OW} is only technical and almost the same with the one derived for the collision resistance, namely \Cref{lem:input-path}, except that in the former one fewer output is queried. Also, the proof of \Cref{cor:preimage-in-db} uses the same reasoning with its analogous statement for collision resistance, which is \Cref{lem:collision-in-db}. If we are in a good database, the proof argues that---conditioned on no collisions in the database---creating a collision requires creating a new path. This requires creating a new tail via a $k$-query or $k'$-query or reaching one of the already reachable outputs in the database via an $h$-query. In a good database, the same argument works for finding a preimage of some value, with a slight difference in the $h$ query case (now, one has to reach a specific output, which is slightly more difficult). If we are in a bad database, without considering the preimage or collision events, we can use the bound for database being bad. Eventually, the bound for having a database dominates the ones for having collision or preimage in the good database. Hence, the preimage and collision finding bounds become equal. 

These lemmas are sufficient to derive the bound for classical collision resistance in our framework. Consequently, they are sufficient to derive the same bound for classical preimage resistance. Moreover, the quantum bound for collision resistance is a direct consequence of the classical bound due to quantum transition capacity and employment of \Cref{cor:fund-lemma-with-R}, which are technical tools and apply similarly to preimage predicate. 

\newpage
\section{Sponge indifferentiability proofs}
\label{sec:indiff-proofs}

Let us consider here the Msponge construction with functions $h, k, k'$ and permutation $\pi$ determining the sponge permutation $\varphi$, as described in \Cref{subsec:alt-sponge}. When using compressed oracles here, we will consider the efficient representation. We consider two worlds.

\textbf{Real world:} The experiment proceeds as: \begin{enumerate}[label=(\arabic*)]
    \item A permutation $\pi :\bit^n\rightarrow\bit^n$ is selected at uniform random.
    \item A database $D_h$ and $D_k, D_{k'}$ are initialized as empty
    \item A $q$ query quantum algorithm $\algo A$ receives oracle access to $\pi$, $\pi^{-1}$ and (compressed) oracle access to $h, k, k'$ using $\cO$ on databases $D_h,D_k,D_{k'}$, and an oracle for $\Sp^{hk'\pi k}$ implemented using the prior oracles.
    \item The algorithm $\algo A$ outputs a bit $b$.
\end{enumerate}

We can write the real world via a set of quantum registers: \begin{align}
    \underbrace{\ket{A}_A}_{\text{distinguisher}} \ket{D_h}_H \ket{D_k}_K \ket{D_{k'}}_{K'}.
\end{align}

\textbf{Ideal world:} The experiment proceeds as: \begin{enumerate}[label=(\arabic*)]
    \item A permutation $\pi :\bit^n\rightarrow\bit^n$ is selected at uniform random.
    \item A database $D_f$ is initialized as empty
    \item A simulator $\algo S$ is given query access to $\pi, \pi^{-1}$ and $f$ using $\cO$ on $D_f$.
    \item A $q$ query quantum algorithm $\algo A$ receives oracle access to $\pi, \pi^{-1}$, access to simulated oracles for $h, k,k'$ by querying the simulator, and an oracle for $f$ implemented using the compressed oracle on $D_f$.
    \item The algorithm $\algo A$ outputs a bit $b$.
\end{enumerate}

We can write the ideal world via a set of quantum registers: \begin{align}
    \underbrace{\ket{A}_A}_{\text{distinguisher}} \underbrace{\ket{D_h}_H \ket{D_k}_K \ket{D_{k'}}_{K'}}_{\text{simulator}} \ket{D_f}_F,
\end{align}

where the underbraces denotes the quantum registers maintained by the distinguisher and the simulator. We will in fact consider indistinguishability of the two worlds given a fixed permutation $\pi$, for any good $\pi$ (as in \Cref{def:perm-good}). We know $1-O(2^{-n})$ permutations are good, so this incurs a negligible loss which we will ignore. Note that our simulator will in fact be secure even against adversaries that see the whole truth table of $\pi$, so long as $\pi$ is good.

\subsection{Defining the simulator}\label{subsec:sim}

To define the action of the simulator, we will need to define a ``find tail'' operation, $\mathsf{fT}$. This operation takes a $z \in \bit^c$, and examines $D_k, D_{k'}$ to find the $\mathsf{tail}(z)$, finding the first such tail if many exist, or a fail flag if one cannot be found. If a tail is found, then this operation also returns the corresponding head $x_i$. Note that the tail uniquely determines the corresponding head.

\begin{definition}
    The operation $\mathsf{find\mbox{-}tail}$, which we write as $\mathsf{fT}$, is a unitary acting on a tail input register $Z$, databases $K, K'$ of functions $\bit^r \rightarrow \bit^c$ with at most $t$ input points, and output registers $T$ for the tail and associated head and $S$ for the success flag. We let $\mathsf{tail}(z)_1$ denote the lexicographically first tail of $z$ in $K, K'$. The operation then acts like \begin{align*}
        \mathsf{fT} \ket{z}_Z \ket{tl}_T \ket{s}_S \ket{D_k}_{K} \ket{D_{k'}}_{K'} \coloneqq& \begin{cases}
            \ket{z}_Z \ket{tl}_T \ket{s}_S \ket{D_k}_{K} \ket{D_{k'}}_{K'} & \text{($\mathsf{tail}(z) = \emptyset$)} \\
            \ket{z}_Z \ket{tl \oplus \mathsf{tail}_1(z)}_T \ket{s \oplus 1}_S \ket{D_k}_{K} \ket{D_{k'}}_{K'} & \text{(Otherwise)}
        \end{cases}
    \end{align*}
    \label{def:find-tail}
\end{definition}

We can now define the action of the simulator. Queries to $k, k'$ are answered using the compressed oracle on $D_k, D_{k'}$. Queries to $h$ on input $z$ are answered using the following procedure.
\begin{enumerate}[label=(\arabic*)]
    \item First, compute $\mathsf{fT}$, to find the tail $tl$ of $z$ if one exists
    \item If a tail does not exist, answer using the compressed oracle and $D_h$.
    \item If a tail does exist, then determine it's corresponding head $x_i$.
    \item Then, query $f(tl)$, and return $x_i \oplus f(tl)$.
    \item Finally, uncompute intermediate variables used above.
\end{enumerate}

\subsection{Good databases}

For this section, we will require a refinement of good databases. In the real world, good databases satisfy the definition we have been using all along. In the ideal world, we additionally stipulate that no $z$ value with a tail appears in the $H$ database. Intuitively, this will be satisfied because the simulator will never query a $z$ value to $h$ which has a tail, and it is unlikely a query to $k'$ attaches to a $z$ value already queried. We will follow the strategy outlined by \Cref{lem:cons-and-ind}, prooving separately indistinguishability and consistency of our simulator. Our indisitinguishability proof requires both notions, and our consistency proof requires only the ideal notion.

\begin{definition}
    We say that databases $D_h, D_k, D_{k'}$ are ``ideal good'' if they satisfy \Cref{def:good-1} and $D_h$ is undefined on every $z$ value that has a tail. We use $\Pi^{\ig}_{KK'H}$ as a projector onto ideal good databases, and $\mathsf{IG}$ as the predicate for ideal good databases.
\end{definition}

\begin{definition}
    We say that database $D_k, D_{k'}$ are ``real good'' if they satisfy \Cref{def:good-1}. We use $\Pi^{\rg}_{KK'}$ as a projector onto real good databases, and $\mathsf{RG}$ as the predicate for real good databases.
\end{definition}

It will be an important fact that the states in our experiment are close to ideal good and real good, in the corresponding worlds. After $q$ queries to the above simulator, the norm of the state on bad databases satisfies $\norm{\Pi^{\ig\perp}_{HKK'} \ket{\psi_I^q}} = \Tilde O(\sqrt{q^5 2^{-\min(r, c)}})$. After $q$ queries to any oracle in the real world, the norm of the state on bad databases satisfies $\norm{\Pi^{\rg\perp}_{KK'} \ket{\psi_R^q}} = \Tilde O(\sqrt{q^5 2^{-\min(r, c)}})$. This can be expressed formally as follows.

\begin{remark}
    We have transition capacities under queries to $k, k',$ and $h$\footnote{Observe that in the ideal world queries to $h$ are not implemented using just $\cO$, and we technically did not define the transition capacity in this case. It will turn out that the simulators action on $h$ queries perfectly preserves ideal goodness, so this point is moot.} given by
    \begin{align*}
        [\![\mathsf{IG}^t \rightarrow \neg \mathsf{IG}]\!] \leq \tilde O(\sqrt{t^3 2^{-\min(r, c)}}) && [\![\mathsf{RG}^t \rightarrow \neg \mathsf{RG}]\!] \leq \tilde O(\sqrt{t^3 2^{-\min(r, c)}}).
    \end{align*}
    These further imply
    \begin{align*}
        [\![\emptyset \xrightarrow[]{q} \neg \mathsf{IG}]\!] \leq \tilde O(\sqrt{q^5 2^{-\min(r, c)}}) && [\![\emptyset \xrightarrow[]{q} \neg \mathsf{RG}]\!] \leq \tilde O(\sqrt{q^5 2^{-\min(r, c)}}).
    \end{align*}
    \label{rem:good-trans-cap}
\end{remark}

\begin{proof}
    The statements concerning $\mathsf{RG}$ follow from combining \Cref{lem:bad-transition-bound} and \Cref{lem:bound-framework}. 
    
    For the statements involving $\mathsf{IG}$, two observations are in order. First, observe that the simulator's action $\algo S^h$ does not affect the database register corresponding to any input $z \in D_h$ where $\mathsf{tail}(z) \neq \emptyset$; this is by definition, since in this case the simulator answers using a compressed query to database register $F$. Thus, a query to $h$ cannot cause a bad event in our simulator, and $\algo S^h$ preserves the ideal good subspace.

    It now only remains to analyze queries to $k$ and $k'$. A transition from ideal good to ideal bad can happen on such queries in the case of a bad attach or a bad completion, whose distances are bounded by \Cref{lem:few-bad-attach,lem:few-bad-comps}. It may also happen in the case where a query causes a value $z$ to have a tail where $(z, y) \in D_h$; note that there are at most $t$ such values by assumption, so using an argument similar to \Cref{lem:few-bad-attach,lem:few-bad-comps}, one can construct a local property of similarly bounded distance governing this translition.
    The statements then follow from \Cref{lem:bad-transition-bound} and \Cref{lem:bound-framework}.
\end{proof}

\subsection{Indistingushability}
\label{subsec:indist}

We first define an isometry $V : \mathcal H_{HKK'} \rightarrow \mathcal H_{HKK'F}$ which will map three function databases to four function databases. $V$ will map computational basis vectors to computational basis vectors, so it is easiest to define $V$ in terms of an injective function $V_c$ on such databases (we consider $\pi$ here to be fixed between both experiments):

\begin{align*}
    V_c(D_h^R, D_k^R, D_{k'}^R) = D_h^I, D_k^I, D_{k'}^I, D_f^I \text{ such that:} \\
     D_k^I(x) &= D_k^R(x) \\
    D_{k'}^I(x) &= D_{k'}^R(x) \\
    D_h^I(z) &= \begin{cases}
        D_h^R(z) & \text{(If $z$ has no tail in $D_k^R, D_{k'}^R$)} \\
        \bot & \text{(Otherwise)}
    \end{cases} \\
    D_f^I(x_1 \Vert \dots \Vert x_m) &= \begin{cases}
        D_h^R(z) \oplus \mathsf{head}(z) & \text{(If $z$'s first tail is $x_1 \Vert \dots \Vert x_m$)} \\
        \bot & \text{(Otherwise)}
    \end{cases}
\end{align*}
When we say the first tail of $z$, we mean the first tail under lexicographic ordering. Observe that $V_c$ is a bijection from real good databases to ideal good databases. We define the full isometry $V$ as the linear continuation of $V_c$, noting that $V$ is an isometry because $V_c$ is an injective mapping on basis states.

\paragraph{Commutation relations of $V$.}

It turns out that the isometry $V$ nearly commutes with the compression operator on both $K$ and $K'$, notated as $L$ in \Cref{eq:local-compression}, at least on the good subspace. Here and going forward, we use the projector $\Pi^t$ to denote the projection onto databases with at most $t$ non-$\bot$ outputs.  \Cref{def:gen-comm} defines the commutator notion of a unitary and an isometry used here.

\begin{lemma}
    The commutator between the isometry $V$ and the local compression operators on $K$ almost commute on the good subspace: \begin{align*}
        \norm{[V_{KK'H}, L_{XK}]\Pi^t_{KK'H}\Pi^{\rg}_{KK'}} \leq \tilde O(\left(\sqrt{t^3 2^{-\min(r, c)}}\right).
    \end{align*}
    \label{lem:comp-V-comm-K}
\end{lemma}

\begin{proof}
    The key observation is essentially that $V$ is a bijective function on good databases which leaves the $k, k'$ databases invariant, and a query to either $k$ or $k'$ remains mostly within the set of good databases. In more detail, consider projectors $\inD_{XK}, \ninD_{XK}$ which sum to the identity. Using the triangle inequality, we split into two cases. \begin{enumerate}[label=(Case \arabic*)]
        \item $x \not\in D_k$, or $\norm{[V_{KK'H}, L_{XK}]\Pi^t_{KK'H}\Pi^{\rg}_{KK'}\ninD_{XK}}$. Note that $\L_{XK}$ preserves computational basis states, i.e. is a controlled operator, on everything except the $x$-th register of $D_k$, and on distinct databases the images of $V$ are orthogonal. Additionally, $\ninD_{XK}\ket x_X=\proj\bot_{D_x}$.  By \Cref{lem:ortho-subspaces-norm}, it suffices to fix $x, D_k, D_{k'}, D_h$ as computational basis states with good databases of size at most $t$ and $x \not\in D_k$\footnote{the cleanest argument is obtained by applying \cref{lem:ortho-subspaces-norm} first wrt. register $X$, and then for the remaining registers}, and consider the action on the state \begin{align*}
            \ket{\psi} =& \ket{x}_X \ket{D_k}_K \ket{D_{k'}}_{K'} \ket{D_h^R}_{H}.
        \end{align*}
        
        Let us define $D_{h}^I$ (in the ideal world) as the database of input/output pairs in ${D_h^R}$ (in the real world) without a tail under $D_k, D_{k'}$, and ${D_{f}^I}$ (in the ideal world) as the database for $f$ constructed from the set of input/output pairs in ${D_h^R}$ with a tail. These are such that \begin{align*}
            V \ket{D_k}_K\ket{D_{k'}}_{K'}\ket{D_h^R}_H = \ket{D_k}_K\ket{D_{k'}}_{K'}\ket{D_h^I}_H \ket{D_f^I}_F.
        \end{align*}
        Recall that databases without a superscript do not change between the ideal and real world.
        We may now compute \begin{align*}
                \ket{\psi^I} =& L_{XK} V \ket{\psi} \\
                =& \ket{A}_A \ket{x}_X \ket{D_h^I}_H \ket{D_{k'}}_{K'}\ket{D_f^I}_F \left( \sum_{y \in \bit^c} 2^{-c/2} \ket{D_k[x\rightarrow y]}_K\right)
            \end{align*}
        as well as \begin{align*}
            \ket{\psi^{R}} =& L_{XK} \ket{\psi} \\
            =& \ket{A}_A \ket{x}_X \ket{D_h^R}_H \ket{D_{k'}}_{K'} \left( \sum_{y \in \bit^c} 2^{-c/2} \ket{D_k[x\rightarrow y]}_K\right).
        \end{align*}
        
        Let $B$ be the set of images $y$ of $x$ such that assigning $x$ to $y$ will cause a bad attach. By \Cref{lem:few-bad-attach}, we have $|B| \leq O(t^3 n + t^3 2^{c-r})$. For any value $y\not\in B$, we have the identity
       \begin{align*}
            V \ket{D_k[x\rightarrow y]}_K \ket{D_{k'}}_{K'} \ket{D_h^I}_H =& \ket{D_k[x\rightarrow y]}_K \ket{D_{k'}}_{K'} \ket{D_h^R}_H \ket{D_f^R}_F,
        \end{align*}
        because in such cases no new $z$ values will have a tail under the assignment $[x\rightarrow y]$. It follows that \begin{align}
            V \Pi^{B\perp} \ket{\psi^R} = \Pi^{B\perp} \ket{\psi^I}. \label{eqn:k-xnin-nobad}
        \end{align}
        We can write \begin{align*}
            \norm{\ket{\psi^R} - \Pi^{B\perp}\ket{\psi^R}} =& \norm{\ket{A}_A \ket{x}_X \ket{D_h^R}_H \ket{D_{k'}}_{K'} \left( \sum_{y \in B} 2^{-c/2} \ket{D_k[x\rightarrow y]}_K\right)} \\
            =& \sqrt{|B| 2^{-c}} \numberthis \label{eqn:k-xnin-real} \\
            \leq& \tilde O(\sqrt{t^3 2^{-\min(r, c)}}), \\
            \norm{\ket{\psi^I} - \Pi^{B\perp}\ket{\psi^I}} =& \norm{\ket{A}_A \ket{x}_X \ket{D_h^I}_H \ket{D_{k'}}_{K'} \ket{D_f^I}_F \left( \sum_{y \in B} 2^{-c/2} \ket{D_k[x\rightarrow y]}_K\right)} \\
            =& \sqrt{|B| 2^{-c}} \\
            \leq& \tilde O(\sqrt{t^3 2^{-\min(r, c)}}). \numberthis \label{eqn:k-xnin-ideal}
        \end{align*}
        Putting everything together, we have \begin{align*}
            \norm{\ket{\psi^I} - V \ket{\psi^{R}}} \leq& \norm{\ket{\psi^I} - \Pi^{B\perp} \ket{\psi^{I}}} + \norm{\Pi^{B\perp} \ket{\psi^I} - V\Pi^{B\perp} \ket{\psi^R}}\\& + \norm{V\ket{\psi^R} - V \Pi^{B \perp}\ket{\psi^{R}}} & \text{(Triangle inequality)} \\
            =& \norm{\ket{\psi^I} - \Pi^{B\perp} \ket{\psi^{I}}} + \norm{V\ket{\psi^R} - V \Pi^{B \perp}\ket{\psi^{R}}} & \text{(\Cref{eqn:k-xnin-nobad})} \\ 
            =& \norm{\ket{\psi^I} - \Pi^{B\perp} \ket{\psi^{I}}} + \norm{\ket{\psi^R} - \Pi^{B \perp}\ket{\psi^{R}}} & \text{($V$ an isometry)} \\
            \leq& \tilde O(\sqrt{t^3 2^{-\min(r, c)}}) & \text{(\Cref{eqn:k-xnin-real,eqn:k-xnin-ideal})}
        \end{align*}
        \item $x \in D_k$, or $\norm{[V_{KK'H}, L_{XK}]\Pi^t_{KK'}\Pi^{\rg}_{KK'}\inD_{XK}}$. Once again, note that $\L_{XK}$ preserves the computational basis everywhere except the $x$-th register of $D_k$, and on distinct databases the images of $V$ are orthogonal. By \Cref{lem:ortho-subspaces-norm}, it suffices to fix $x, D_k, D_{k'}, D_h$ as computational basis states with good databases of size at most $t$ and where $x \not\in D_k$, and consider the action on a state of the form \begin{align*}
            \ket{\psi} =& \sum_{z \text{ s.t. } D_k[x\rightarrow z], D_{k'} \text{ is good}} \alpha_z \ket{x}_X \ket{D_k[x\rightarrow z]}_K \ket{D_{k'}}_{K'} \ket{D_h}_H.
        \end{align*}
        Let us similarly define ${D_{h}^{I,y}}$ (in the ideal world) as the database of input/output pairs in ${D_h^R}$ (in the real world) without a tail under $D_k[x\rightarrow y], D_{k'}$, and ${D_{f}^{I, y}}$ (in the ideal world) as the database for $f$ constructed from the set of input/output pairs in ${D_h^R}$ with a tail in $D_k[x\rightarrow y], D_{k'}$. These are such that \begin{align*}
            V \ket{D_k[x\rightarrow y]}_K\ket{D_{k'}}_{K'}\ket{D_h^R}_H = \ket{D_k[x\rightarrow y]}_K\ket{D_{k'}}_{K'}\ket{D_h^{I, y}}_H \ket{D_f^{I, y}}_F.
        \end{align*}
        Recall that databases without a superscript do not change between the ideal and real world. We may now compute \begin{align*}
                \ket{\psi^I} =& L_{XK} V \ket{\psi} \\
                =& \sum_{z \text{ s.t. } D_k[x\rightarrow z], D_{k'} \text{ is good}} \alpha_z \ket{A}_A \ket{x}_X \ket{D_h^{I, z}}_H \ket{D_{k'}}_{K'}\ket{D_f^{I, z}}_F \otimes \\
                &\left(\ket{D_k[x\rightarrow z]}_K - 2^{-c} \sum_{u \in \bit^c} \ket{D_k[x\rightarrow u]}_K + 2^{-c/2} \ket{D_k}_K\right)
            \end{align*}
        as well as \begin{align*}
            \ket{\psi^{R}} =& L_{XK} \ket{\psi} \\
                =& \sum_{z \text{ s.t. } D_k[x\rightarrow z], D_{k'} \text{ is good}} \alpha_z \ket{A}_A \ket{x}_X \ket{D_h^R}_H \ket{D_{k'}}_{K'}\otimes \\
                &\left(\ket{D_k[x\rightarrow z]}_K - 2^{-c} \sum_{u \in \bit^c} \ket{D_k[x\rightarrow u]}_K + 2^{-c/2} \ket{D_k}_K\right)
        \end{align*}
        Let $B$ be the set of images $y$ of $x$ such that assigning $x$ to $y$ will cause a bad attach. Observe that, from the analysis of \Cref{lem:few-bad-attach}, we have $|B| \leq O(t^3 n + t^3 2^{c-r})$. For any value $y\not\in B$, we have the identity \begin{align*}
            V \ket{D_k[x\rightarrow y]}_X \ket{D_{k'}}_{K'} \ket{D_h^I}_H =& \ket{D_k[x\rightarrow y]}_X \ket{D_{k'}}_{K'} \ket{D_h^R}_H \ket{D_f^R},
        \end{align*}
        because in such cases no new state values will have a tail under the assignment $[x\rightarrow y]$. Observe here that the initial state $\ket{\psi}$ may be entirely supported on images in $B$ that lead to a ``bad attach'', for instance if $x$ is part of a tail in $D_k[x\rightarrow z], D_{k'}$. This prevents applying the strategy from the previous case where we just project out those values, so instead we have to explicitly write down the difference. Let us analyze the difference \begin{align*}
            \ket{\psi^I} - V\ket{\psi^R} =& \sum_{z \text{ s.t. } D_k[x\rightarrow z], D_{k'} \text{ is good}} \alpha_z \ket{A}_A \ket{x}_X \ket{D_h^{I, z}}_H \ket{D_{k'}}_{K'}\ket{D_f^{I, z}}_F \otimes \\
                &\left(\ket{D_k[x\rightarrow z]}_K - 2^{-c} \sum_{u \in \bit^c} \ket{D_k[x\rightarrow u]}_K + 2^{-c/2} \ket{D_k}_K\right)- \\
                &\sum_{z \text{ s.t. } D_k[x\rightarrow z], D_{k'} \text{ is good}} \alpha_z \ket{A}_A \ket{x}_X \ket{D_{k'}}_{K'}\otimes \\
                &\Bigg(\ket{D_k[x\rightarrow z]}_K \ket{D_h^{I, z}}_H \ket{D_f^{I, z}}_F  - 2^{-c} \sum_{u \in \bit^c} \ket{D_k[x\rightarrow u]}_K\ket{D_h^{I, u}}_H \ket{D_f^{I, u}}_F +\\& 2^{-c/2} \ket{D_k}_K \ket{D_h^{I, \bot}}_H \ket{D_f^{I, \bot}}_F\Bigg)
        \end{align*}
        Which, after collapsing terms, can be written as \begin{align*}
            \ket{\psi^I} - V\ket{\psi^R} =& \sum_{z \text{ s.t. } D_k[x\rightarrow z], D_{k'} \text{ is good}} \alpha_z \ket{A}_A \ket{x}_X \ket{D_{k'}}_{K'} \otimes \\
                &\left(2^{-c} \sum_{u \in \bit^c} \ket{D_k[x\rightarrow u]}_K(\ket{D_h^{I, u}}_H \ket{D_f^{I, u}}_F - \ket{D_h^{I, z}}_H\ket{D_f^{I, z}}_F)\right) + \\
                &\sum_{z \text{ s.t. } D_k[x\rightarrow z], D_{k'} \text{ is good}} \alpha_z \ket{A}_A \ket{x}_X \ket{D_{k'}}_{K'}\otimes \\
                &\left(\ket{D_k}_K 2^{-c/2}(\ket{D_h^{I, z}}_H \ket{D_f^{I, z}}_F  - \ket{D_h^{I, \bot}}_H \ket{D_f^{I, \bot}}_F)\right).
        \end{align*}
        We can then write \begin{align*}
            \norm{VL\ket{\psi} - LV \ket{\psi}} =& \norm{V\ket{\psi^R} - \ket{\psi^I}} \\
            \leq& 2\underbrace{\norm{\sum_{z \in \bit^c} \alpha_z \sum_{u \in \bit^c, (D^{I, u}_h D^{I, u}_f) \neq (D^{I, z}_h D^{I, z}_f)} 2^{-c} \ket{D_k[x\rightarrow u]}}}_{T_1} +\\& \underbrace{2\norm{\sum_{z \in \bit^c, (D^{I, z}_h D^{I, z}_f) \neq (D^{I, \bot}_h D^{I, \bot}_f)} \alpha_z 2^{-c/2}}}_{T_2} & \text{(Triangle Inequality)}
        \end{align*}
        Let us focus on bounding each term individually. We begin with \begin{align*}
            T_1 =& \norm{\sum_{z \in \bit^c \quad} \sum_{u \in \bit^c, (D^{I, u}_h D^{I, u}_f) \neq (D^{I, z}_h D^{I, z}_f)} \alpha_z 2^{-c} \ket{D_k[x\rightarrow u]}} \\
            \leq& \underbrace{\norm{\sum_{z \in \bit^c, (D^{I, z}_h D^{I, z}_f) \neq (D^{I, \bot}_h D^{I, \bot}_f) \quad} \sum_{u \in \bit^c, (D^{I, u}_h D^{I, u}_f) \neq (D^{I, z}_h D^{I, z}_f)} \alpha_z 2^{-c} \ket{D_k[x\rightarrow u]}}}_{T_{11}} + \\& \underbrace{\norm{\sum_{z \in \bit^c, (D^{I, z}_h D^{I, z}_f) = (D^{I, \bot}_h D^{I, \bot}_f)\quad} \sum_{u \in \bit^c, (D^{I, u}_h D^{I, u}_f) \neq (D^{I, z}_h D^{I, z}_f)} \alpha_z 2^{-c} \ket{D_k[x\rightarrow u]}}}_{T_{12}} & \text{(Triangle inequality)}
        \end{align*}
        Observe that the second sum in $T_{11}$ is over at most $2^c$ terms, which gives an upper bound of $\abs{\alpha_z}2^{-c/2}$ for it's norm. The first sum in $T_{11}$ is over a set of size $O(t^3 n + t^3 2^{c-r})$ by \Cref{lem:tail-insensitive,lem:few-bad-comps}, so by a standard inequality between $L_1$ and $L_2$ norm we have \begin{align*}
            \sum_{z \in \bit^c, (D^{I, z}_h D^{I, z}_f) \neq (D^{I, \bot}_h D^{I, \bot}_f)} \abs{\alpha_z} \leq O(\sqrt{t^3 n + t^3 2^{c-r}}). \numberthis \label{eqn:sum-alphs}
        \end{align*}
        It follows that $T_{11} \leq \tilde O(\sqrt{t^3 2^{-\min(r, c)}})$.
        
        Observe that the second sum in $T_{12}$ is over a set of size $O(t^3 n + t^3 2^{c-r})$ by \Cref{lem:tail-insensitive,lem:few-bad-comps}, which gives an upper bound of $\abs{\alpha_z}\sqrt{O(t^3n + t^3 2^{c-r})} \cdot 2^{-c}$ for it's norm. The first sum in $T_{12}$ is over a set of size at most $2^{-c}$, so by the relation between $L_1$ and $L_2$ norm we have \begin{align*}
            \sum_{z \in \bit^c, (D^{I, z}_h D^{I, z}_f) = (D^{I, \bot}_h D^{I, \bot}_f)} \abs{\alpha_z} \leq 2^{c/2}.
        \end{align*}
        It follows that $T_{12} \leq \tilde O(\sqrt{t^3 2^{-\min(r, c)}})$.

        Finally, it follows from \Cref{eqn:sum-alphs} that $T_2 \leq \tilde O(\sqrt{t^3 2^{-\min(r, c)}})$.
    \end{enumerate}
\end{proof}

\begin{lemma}
    The commutator between the isometry $V$ and the local compression operators on $K'$ almost commutes on the good subspace: \begin{align*}
        \norm{[V_{KK'H}, L_{XK'}]\Pi^t_{KK'H}\Pi^{\rg}_{KK'}} \leq \tilde O(\left(\sqrt{t^3 2^{-\min(r, c)}}\right).
    \end{align*}
    \label{lem:comp-V-comm-K'}
\end{lemma}

The proof for the $k'$ case is similar to the one for $k$. For completeness, we give it in \Cref{sec:def-proofs}.

\paragraph{Query closeness.}

We can apply the previous commutator bounds to give a bound on how much the distinguisher's view can diverge as a result of queries to $k, k',$ and $h$. We do this by showing that $V$ acts as an approximate \emph{intertwiner} between real world and ideal world queries. Consider, e.g., the query unitary $\cO_{AK}$ and the operator $\algo S^k$ the simulator applies upon a $k$ query. Then $S^kV\approx V\cO_{AK}$ for an appropriate notion of $\approx$. Note that in this case $\algo S^k_{AHKK'F} = \cO_{AK}$, as the simulator simply answers by a compressed database call on database $K$, so here we have, in fact, a certain approximate commutation relation.

\begin{lemma}\label{lem:k-query-close}
    Let $\algo S^k$ denote the action of the simulator on a $k$ query. Consider the following two operators, \begin{align*}
        O^I =&  \algo S^k_{AHKK'F} V_{HKK'} \Pi^{\rg}_{HKK'}\\
        O^R =& V_{HKK'} \Pi^{\rg}_{HKK'} \cO_{AK} \Pi^{\rg}_{KK'}
    \end{align*}
    and let $\Pi^{t}_{HKK'}$ be the projector onto databases with at most $t$ query points in each. Then we have \begin{align*}
        \norm{(O^I - O^R)\Pi^{t}} \leq O(\sqrt{t^3 2^{-\min(r, c)}})
    \end{align*}
    \label{lem:k-queries-close}
\end{lemma}

\begin{proof}
    Since the simulator replies to $k$-queries simply by using the corresponding compressed oracle, we can write
    \begin{align}
        O^I =& \algo S^k_{AHKK'F} V_{HKK'} \Pi^{\rg}_{HKK'}\\
        =& L_{XK} \algo P_{XYK} L_{XK} V_{HKK'} \Pi^{\rg}_{HKK'},
    \end{align}
    observing that \begin{align*}
        V_{HKK'}\Pi^{\rg}_{KK'} = \Pi^{\ig}_{HKK'} V_{HKK'} \numberthis\label{eqn:V-good-int}
    \end{align*}
    we further have \begin{align*}
        O^R =& \Pi^{\ig}_{HKK'} V_{HKK'} L_{XK} \algo P_{XYK} L_{XK} \Pi^{\rg}_{KK'}.
    \end{align*}
    Making heavy use of \Cref{eqn:V-good-int}, and the fact that $\algo P_{XYK}$ commutes with good projectors and the $V$ operation, we have \begin{align*}
        \norm{(O^I - O^R) \Pi^t} =& \norm{(L \algo P L V \Pi^{\rg} - \Pi^{\ig} V L \algo P L \Pi^{\rg}) \Pi^t} \\
        \leq& \norm{(L \algo P L V \Pi^{\rg} - L \algo P \Pi^{\ig} V  L \Pi^{\rg})\Pi^t} + \underbrace{\norm{(L \algo P V \Pi^{\rg}  L \Pi^{\rg} - L V \Pi^{\rg} \algo P L \Pi^{\rg})\Pi^t}}_{=0} + \\
        &\norm{(L V \Pi^{\rg} \algo P L \Pi^{\rg} - \Pi^{\ig} V L \Pi^{\rg} \algo P L \Pi^{\rg})\Pi^t} + \\& \norm{V L \Pi^{\rg} \algo P L \Pi^{\rg})\Pi^t - \Pi^{\ig} V L \Pi^{\rg} \algo P L \Pi^{\rg})\Pi^t}\,\,\,\,\,\,\, \text{(Triangle inequality)} \\
        \leq& \norm{[L_{XK}, V] \Pi^{\rg} \Pi^t} + \norm{[L_{XK}, V] \Pi^{\rg} \Pi^{t+1}} + \\& \tilde O(\sqrt{t^3 2^{-\min(r, c)}}) \,\,\,\,\,\,\,\,\,\,\,\,\,\,\,\,\,\,\,\,\,\,\,\,\,\,\,\,\,\,\,\,\,\,\,\,\,\,\,\,\,\,\,\,\,\,\,\,\,\,\,\,\,\,\,\,\,\,\,\,\,\,\,\,\,\,\,\,\,\,\,\,\,\,\,\, \text{(\Cref{rem:good-trans-cap})} \\
        \leq& \tilde O(\sqrt{t^3 2^{-\min(r, c)}}) \,\,\,\,\,\,\,\,\,\,\,\,\,\,\,\,\,\,\,\,\,\,\,\,\,\,\,\,\,\,\,\,\,\,\,\,\,\,\,\,\,\,\,\,\,\,\,\,\,\,\,\,\,\,\,\,\,\,\,\,\,\,\,\,\,\,\,\,\,\,\,\,\,\,\,\, \text{(\Cref{lem:comp-V-comm-K})}
    \end{align*}
\end{proof}

Once again, note that $\algo S^{k'}_{AHKK'F} = \cO_{AK'}$, as the simulator simply answers by a compressed database call on database $K$.

\begin{lemma}
    Let $\algo S^{k'}$ denote the action of the simulator on a $k'$ query. Consider the following two operators, \begin{align*}
        O^I =&  \algo S^{k'}_{AHKK'F} V_{HKK'} \Pi^{\rg}_{HKK'}\\
        O^R =& V_{HKK'} \Pi^{\rg}_{KK'} \cO_{AK'} \Pi^{\rg}_{KK'},
    \end{align*}
    and let $\Pi^{t}_{HKK'}$ be the projector onto databases with at most $t$ query points. Then we have \begin{align*}
        \norm{(O^I - O^R)\Pi^{t}} \leq O(\sqrt{t^3 2^{-\min(r, c)}})
    \end{align*}
    \label{lem:k'-queries-close}
\end{lemma}

\begin{proof}
    The proof follows similarly to that of \Cref{lem:k-queries-close}, except with \Cref{lem:comp-V-comm-K'} taking the place of \Cref{lem:comp-V-comm-K}.
\end{proof}

\begin{lemma}
    Let $\algo S^{h}$ denote the action of the simulator on a $h$ query. Consider the following two operators, \begin{align*}
        O^I =&  \algo S^{h}_{AHKK'F} V_{HKK'} \Pi^{\rg}_{HKK'}\\
        O^R =& V_{HKK'} \Pi^{\rg}_{KK'} \cO_{AH} \Pi^{\rg}_{KK'}.
    \end{align*}
    Then we have \begin{align*}
        O^I = O^R.
    \end{align*}
    \label{lem:h-queries-close}
\end{lemma}

\begin{proof}
    Recall that $V$ is a bijection from good real databases to good ideal databases, and an injection from real databases to ideal. Observe that the construction of $D_{f}^I$ and $D_{h}^I$ is simply a re-labeling of the input/output pairs contained in $D_{h}^R$, depending on which have a tail. Further, the simulator answers according to a random function which depends only on the input value of such pairs. Let us work out the action of $O^I$ and $O^R$ on good databases. Let \begin{align*}
        \ket{\psi^R} = \ket{A}_A \ket{x}_X \ket{y}_Y \ket{D_k}_K \ket{D_{k'}}_{K'} \ket{D_h}_h
    \end{align*}
    denote a computational basis state in the real world which is good.
    \begin{itemize}
        \item[$(O^I)$.] We will consider the action on $\ket{\psi^R}$. We have $\mathsf{tail}(z)$ (with respect to $D_k$, $D_{k'}$) is either an empty or a singleton set, and so is $\mathsf{head}(z)$. Let us define the sets \begin{align*}
            Z_t &= \{z | z \in D_h, \mathsf{tail}(z) \neq \emptyset\} \\
            Z_{\neg t} &= \{z | z \in D_h, \mathsf{tail}(z) = \emptyset\}.
        \end{align*}
        We can then work out \begin{align*}
            V\Pi^{\rg} \ket{\psi_R} =& \ket{A}_A \ket{x}_x \ket{y}_y \ket{D_k}_k \ket{D_{k'}}_{k'} \ket{\{(z, h(z)) | z \in Z_{\neg t}\}}_h \ket{\{(\mathsf{tail}(z), h(z) \oplus \mathsf{head}(z)) | z \in Z_{t}}_f 
        \end{align*}
        The action of the simulator can then be split into two cases. \begin{enumerate}[label=(\arabic*)]
            \item $x \in Z_t$. In this case, the query is answered using a compressed oracle call on input $\mathsf{tail}_1(x)$ to $D_f^I$, and XORed with $\mathsf{head}(x)$.
            \item $x \in Z_{\neg t}$. In this case, the query is answered using a compressed oracle call on input $x$ to $D_h$.
        \end{enumerate}
        \item[$(O^R)$.] We will consider the action on $\ket{\psi^R}$. We have $\mathsf{tail}(z)$ (with respect to $D_k$, $D_{k'}$) is of size at most one. Let us again consider the sets \begin{align*}
            Z_t &= \{z | z \in D_h, \mathsf{tail}(z) \neq \emptyset\} \\
            Z_{\neg t} &= \{z | z \in D_h, \mathsf{tail}(z) = \emptyset\}.
        \end{align*}
        The query to $h$ will be answered using $\cO_{XYH}$, which is a controlled operation on $x$ that targets the $x$-th output in the database for $h$. Note that (because $D_h$ is good) we have that $V$ will simply move and permute the labels of input output pairs. 
        It follows that 
         \begin{align*}
            V L_{XH} \ket{x}_X =& \begin{cases}
                C^{\mathsf{tail}(x)}_f V \ket{x}_X & \text{(If $\mathsf{tail}(x) \neq \emptyset$)} \\
                C^{x}_H V \ket{x}_X & \text{(Otherwise)}
            \end{cases}
        \end{align*}
        Using the notation $\algo P_{x,y,H}$ to denote $\algo P$ with input value $x$, output value $y$, and database register $H$, we then have for any $x,y$ that
        \begin{align*}
            V \algo P_{XYH}\ket x_X\ket y_Y =& \begin{cases}
                \algo P_{\mathsf{tail}(x), y \oplus \mathsf{head}(x), F} V \ket x_X\ket y_Y& \text{(If $\mathsf{tail}(x) \neq \emptyset$)} \\
                \algo P_{x,y,H} V\ket x_X\ket y_Y & \text{(Otherwise)}
            \end{cases}
        \end{align*}
        It follows that queries are answered using the same procedure as by the simulator.
    \end{itemize}
    The action of $O^I$ and $O^R$ is the same on all good computational basis states, so it is the same on the subspace spanned by good computational basis states. Note that $\Pi^{\rg}$ does not depend on the $h$ database, and so clearly commutes with $\cO_{AH}$; this justifies the final $\Pi^{\rg}$ in $O^R$.
\end{proof}

\paragraph{Putting the pieces together.}
\begin{theorem}
    The simulator defined in \cref{subsec:sim} is indistinguishable. In particular, for a $q$-query distinguisher we have \begin{align*}
        |\Pr[\algo A^{\pi,\pi^{-1},h,k,k'}() = 1] - \Pr[\algo A^{\pi, \pi^{-1}, \algo S^{f}}() = 1]| = \tilde O\left(\sqrt{q^5 2^{-\min(r, c)}}\right)
    \end{align*}
    \label{thm:sim-indist}
\end{theorem}

\begin{proof}
    Let us consider some initial state for the real experiment, \begin{align*}
        \ket{\psi^0} = \ket{\alpha^0}_A \ket{\emptyset}_H \ket{\emptyset}_K \ket{\emptyset}_{K'},
    \end{align*}
    where $\ket{\alpha^0}$ is an arbitrary initial state of the adversary. Let us consider in fact a strengthened adversary which is time unbounded, and which holds the entire truth table of $\pi$. The only constraint we require is that $\pi$ is good, which happens with probability $1 - O(2^{-n})$. Our proof shows that for any fixed, good $\pi$, the real and ideal worlds are indistinguishable. We will do this by showing that the reduced density matrix of an adversary which queries compressed random oracles for $h, k, k'$, and one which queries the simulator are close, which we do by showing an isometry mapping the purification of the real state to an approximate purification of the ideal state.
    
    We can WLOG take our adversary to alternate queries to the various oracles (at a loss of an $O(1)$ factor), suppose in the order $k, k', h$. We can parameterize a $q$-query adversary of this form by unitaries $U^1, \dots, U^q$ acting on the internal register $A$. The final state in the real experiment is given by \begin{align*}
        \ket{\psi^R} = U^{q}_A \dots U^3_A \cO_{AH} U^2_A\cO_{AK'} U^1_A \cO_{AK} \ket{\psi^0}.
    \end{align*}
    We can write the state at the end of the ideal experiment as \begin{align*}
        \ket{\psi^I} =& U^{q}_A \dots U^3_A \algo S^h_{AHKK'F} U^2_A\algo S^{k'}_{AHKK'F} U^1_A \algo S^k_{AHKK'F} V_{KK'H}\ket{\psi^0}.
    \end{align*}
    Recall that we would like show that $V_{KK'H} \ket{\psi^R}$ is close to $\ket{\psi^I}$. Let us define $\ket{\psi^{R, g}}$ similar to $\ket{\psi^R}$, but with projectors onto good placed after each query: \begin{align*}
        \ket{\psi^{R, g}} \coloneqq U^{q}_A  \Pi^{\rg} \dots U^3_A \Pi^{\rg} \cO_{AH} U^2_A \Pi^{\rg} \cO_{AK'} U^1_A  \Pi^{\rg} \cO_{AK}  \Pi^{\rg} \ket{\psi^0}.
    \end{align*}
    Note that in the real world, we have \begin{align*}
        \norm{V_{KK'H} \ket{\psi^R} - V_{KK'H} \ket{\psi^{R, g}}} =& \norm{\ket{\psi^R} - \ket{\psi^{R, g}}} & \text{($V$ an isometry)} \\
        \leq& \tilde O(\sqrt{q^5 2^{-\min(r, c)}}) & \text{(\Cref{rem:good-trans-cap})}
    \end{align*}
    and in the ideal, \begin{align*}
        \ket{\psi^I} =& U^{q}_A \dots U^3_A \algo S^h_{AHKK'F} U^2_A\algo S^{k'}_{AHKK'F} U^1_A \algo S^k_{AHKK'F} V_{KK'H}\ket{\psi^0} \\
        =& U^{q}_A \dots U^3_A \algo S^h_{AHKK'F} U^2_A\algo S^{k'}_{AHKK'F} U^1_A \algo S^k_{AHKK'F} V_{KK'H}  \Pi^{\rg} \ket{\psi^0}.
    \end{align*}
    Let us now define \begin{align*}
        \ket{\psi^I_0} \coloneqq& U^{q}_A \dots U^3_A \algo S^h_{AHKK'F} U^2_A\algo S^{k'}_{AHKK'F} U^1_A \algo S^k_{AHKK'F} V_{KK'H}  \Pi^{\rg} \ket{\psi^0} \\
        \ket{\psi^I_1} \coloneqq& U^{q}_A \dots U^3_A \algo S^h_{AHKK'F} U^2_A\algo S^{k'}_{AHKK'F} V_{KK'H} \Pi^{\rg} U^1_A \cO_{AK} \Pi^{\rg}  \ket{\psi^0} \\
        \ket{\psi^I_2} \coloneqq& U^{q}_A \dots U^3_A \algo S^h_{AHKK'F} V_{KK'H}  \Pi^{\rg} U^2_A  \cO_{AK'} \Pi^{\rg} U^1_A \cO_{AK} \Pi^{\rg} \ket{\psi^0} \\ 
        \vdots& \\
        \ket{\psi^I_q} \coloneqq& V_{KK'H} \Pi^{\rg} U^{q}_A \dots \Pi^{\rg} U^3_A \algo \cO_{AH}  \Pi^{\rg} U^2_A  \Pi^{\rg} \cO_{AK'} U^1_A \cO_{AK} \Pi^{\rg} \ket{\psi^0},
    \end{align*}
    where we have $\ket{\psi^I_0} = \ket{\psi^I}$ and $\ket{\psi^I_q} = V \ket{\psi^{R, g}}$ (since the $U$'s acting on $A$ commute with $\Pi^{\rg}_{KK'}$), meaning  $\norm{\ket{\psi^I_q} - V\ket{\psi^R}} \leq \tilde O(\sqrt{q^5 2^{-\min(r, c)}})$ by the argument prior. To complete the proof, observe that we have \begin{align*}
        \norm{\ket{\psi^I_1} - \ket{\psi^I_0}} \leq& \tilde O(\sqrt{1^3 2^{-\min(r, c)}}) & \text{(By \Cref{lem:k-queries-close})} \\
        \norm{\ket{\psi^I_2} - \ket{\psi^I_1}} \leq& \tilde O(\sqrt{2^3 2^{-\min(r, c)}}) & \text{(By \Cref{lem:k'-queries-close})} \\ 
        \vdots& \\
        \norm{\ket{\psi^I_{t+1}} - \ket{\psi^I_t}} \leq& \tilde O(\sqrt{t^3 2^{-\min(r, c)}}) & \text{(By \Cref{lem:k-queries-close,lem:k'-queries-close,lem:h-queries-close})}
    \end{align*}
    and by triangle inequality, \begin{align*}
        \norm{\ket{\psi^I} - V_{KK'H}\ket{\psi^R}} \leq& \norm{\ket{\psi^I_q} - V_{KK'H}\ket{\psi^R}} + \sum_{t=1}^q \norm{\ket{\psi^I_{t}} - \ket{\psi^I_{t-1}}} \\
        \leq& \tilde O(\sqrt{q^5 2^{-\min(r, c)}}).
    \end{align*}
    The claim now follows from \Cref{lem:approx-uhlmann}.
\end{proof}

\subsection{Consistency}
\label{subsec:cons}

We would like to argue that the procedure for answering Msponge queries using the simulator are close to the ideal functionality. We will compute the Msponge using an out-of-place circuit in each round, as depicted in \Cref{fig:sponge-round-comp,fig:sponge-state-comp,fig:sponge-value-comp}. Let $l$ be an upper bound on the block length of an Msponge input, and we will use $t$ to represent the number of queries made so far to the oracles. Let us first define some operations which can be seen as building blocks of the full Msponge.

Throughout this section, let $\ket{x}_X  \ket{D_k}_K \ket{D_{k'}}_{K'} \ket{D_h}_H$ denote input registers and function databases respectively. We split $x=x_1 \Vert \dots \Vert x_l$ into $r$ bit blocks.

\begin{definition}
    Define isometry $A$ as the isometry which simply appends registers $\ket{0}_Z \ket{0}_{Hd} \ket{0}_W$ for the state, head, and intermediate outputs registers, as well as $\ket{0}_T \ket{0}_S$ for the tail, and success flag.
    \label{def:iso-A}
\end{definition}

\begin{definition}
    Define unitary $U$ as the unitary which out-of-place computes the sponge state in input $x$, depicted in \Cref{fig:sponge-round-comp,fig:sponge-state-comp}. In particular, compute the sponge state from $X$ into $Z$ and ${Hd}$ registers, as in \Cref{fig:sponge-state-comp}, using the $W$ register for intermediate computations.
    \label{def:sponge-U}
\end{definition}

\begin{definition}
    We define isometries $O^R, O^I$ as follows. Each operation will make use of $A$ and $U$, and we refer to the $Z, Hd, W$ registers created by $A$ using the same labels. We will refer to the tail $T$ as both a tail portion and a recovered head (potentially distinct from the actual head $Hd$) portion. The operators both begin by applying $A$ to create the commensurate registers, and continue as follows.
    \begin{itemize}
        \item[$O^R$.] Continue with the following operations: \begin{enumerate}[label=(\arabic*)]
            \item Use $U$ to compute the sponge state from $X$ into $Z$ and ${Hd}$ registers, as in \Cref{fig:sponge-state-comp}, using the $W$ register for intermediate computations
            \item Call $\mathsf{find\mbox{-}tail}$ with input register $Z$, on databases $K, K'$, with target register $T$ (both portions) and success flag target $S$.
        \end{enumerate}
        
        \item[$O^I$.] Continue with the following operations: \begin{enumerate}[label=(\arabic*)]
            \item Flip on the success flag $S$ 
            \item XOR the $X$ register into the tail portion of $T$
            \item Use $U$ to compute the sponge state from $X$ into $Z$ and ${Hd}$, as in \Cref{fig:sponge-state-comp}, using the $W$ register for intermediate computations
            \item XOR the head ${Hd}$ into the recovered head portion of $T$
        \end{enumerate} 
    \end{itemize}
    \label{defn:real-ideal-tail-isoms}
\end{definition}

To start, we will show that the databases for $K, K'$ almost always contain all of the input-output points computed by $O^R$, so long as the input state is valid. Noting that this is true for the uncompressed databases, this essentially follows from the fundamental lemma. Let us begin by defining a projector onto states where $K, K'$ indeed contain all of the input-output points computed by $O^R$. Here $g$ denotes a placeholder data for workspace values which do not correspond to input-output pairs.

\begin{definition}
    Define the recorded projector $\Pi_{X, W, K, K'}^{r}$ that checks whether every $(x, k(x))$ and $(x', k(x'))$ pair in the $X, W$ registers appears in the $K, K'$ database. In particular, we have \begin{align*}
        \ket{\psi} =& \ket{\mathbf x}_X \ket{\mathbf w \Vert \mathbf x' \Vert \mathbf w' \Vert \mathbf g}_W \ket{D_k}_K \ket{D_{k'}}_K \\
        \Pi^r \ket{\psi} \coloneqq& \begin{cases}
            0 & \text{(If $\exists i$ s.t. $(x_i, w_i) \not\in D_k$)} \\
            0 & \text{(If $\exists i$ s.t. $(x_i', w_i') \not\in D_{k'}$)} \\
            \ket{\psi} & \text{(Otherwise)}
        \end{cases}
    \end{align*}
\end{definition}

\begin{lemma}
    On a valid initial state $\ket{\psi}$, the input to the $\MSp$ is almost always recorded in the databases after the out-of-place circuit. Formally, \begin{align*}
        \norm{\Pi^r O^R \ket{\psi}} \geq 1 - O(\sqrt{l2^{-c}}),
    \end{align*}
    \label{lem:always-tail}
\end{lemma}

\begin{proof}
    Let $\ket{\psi}_{XKK'H}$ be a state such that $\Pi^v\ket{\psi} = \ket{\psi}$. Observe that the fully uncompressed databases of $K, K'$ would be guaranteed to satisfy $\Pi^r$ at the end; this is because the $X$ and $W$ registers contain an exact record of the queries made to $k, k'$. In other words, we can write \begin{align*}
        \Pi^r C_K C_{K'} O^R \ket{\psi} = C_K C_{K'} O^R \ket{\psi}, \numberthis\label{eqn:valid-has-tail-assump}
    \end{align*}
    where here the $C$ operator denotes full decompression. Let us define \begin{align*}
        \ket{\psi'} =& O^R \ket{\psi} 
        =& \sum_{\mathbf x, \mathbf w, \mathbf x', \mathbf w'} \alpha_{\mathbf x, \mathbf w, \mathbf x', \mathbf w'} \ket{\mathbf x}_X \ket{\mathbf w \Vert \mathbf x' \Vert \mathbf w' \Vert \mathbf g}_W \ket{D_{k,\mathbf x, \mathbf w, \mathbf x', \mathbf w'}}_K \ket{D_{k', \mathbf x, \mathbf w, \mathbf x', \mathbf w'}}_{K'}
    \end{align*}
    where $\mathbf x, \mathbf w, \mathbf x', \mathbf w'$ are the recorded input-output pairs at the end of $O^R$, and $\mathbf g$ captures all other non-database wires. We know that $\ket{\psi'}$ satisfies \Cref{eqn:valid-has-tail-assump}, which implies \begin{align*}
        \forall \mathbf x, \mathbf w, \mathbf x', \mathbf w':&& \\
        &&\Pi^{(\mathbf x, \mathbf w)}C_K\ket{D_{k, \mathbf x, \mathbf w, \mathbf x', \mathbf w'}}_K =& C_K\ket{D_{k, \mathbf x, \mathbf w, \mathbf x', \mathbf w'}}_K\\
        && \Pi^{(\mathbf x', \mathbf w')}C_{K'}\ket{D_{k', \mathbf x, \mathbf w, \mathbf x', \mathbf w'}}_{K'} =&  C_{K'}\ket{D_{k', \mathbf x, \mathbf w, \mathbf x', \mathbf w'}}_{K'}.
    \end{align*}
    We can now write \begin{align*}
        \norm{\Pi^r O^R \ket{\psi}} =& \Bigg\Vert\Pi^r \sum_{\mathbf x, \mathbf w, \mathbf x', \mathbf w'} \alpha_{\mathbf x, \mathbf w, \mathbf x', \mathbf w'} \ket{\mathbf x}_X \ket{\mathbf w \Vert \mathbf x \Vert \mathbf w' \Vert \mathbf g}_W \otimes \\& \ket{D_{k,\mathbf x, \mathbf w, \mathbf x', \mathbf w'}}_K \ket{D_{k', \mathbf x, \mathbf w, \mathbf x', \mathbf w'}}_{K'}\Bigg \Vert \\
        =& \Bigg\Vert\sum_{\mathbf x, \mathbf w, \mathbf x', \mathbf w'} \alpha_{\mathbf x, \mathbf w, \mathbf x', \mathbf w'} \ket{\mathbf x}_X \ket{\mathbf w \Vert \mathbf x \Vert \mathbf w' \Vert \mathbf g}_W \otimes\\& \left(\Pi^{(\mathbf x, \mathbf w)}\ket{D_{k, \mathbf x, \mathbf w, \mathbf x', \mathbf w'}}_K\right) \otimes \left(\Pi^{(\mathbf x', \mathbf w')}\ket{D_{k', \mathbf x, \mathbf w, \mathbf x', \mathbf w'}}_{K'}\right)\Bigg\Vert \\
        =& \Bigg(\sum_{\mathbf x, \mathbf w, \mathbf x', \mathbf w'} \abs{\alpha_{\mathbf x, \mathbf w, \mathbf x', \mathbf w'}}^2 \cdot \\& \norm{\left(\Pi^{(\mathbf x, \mathbf w)}\ket{D_{k, \mathbf x, \mathbf w, \mathbf x', \mathbf w'}}_K\right) \otimes \left(\Pi^{(\mathbf x', \mathbf w')}\ket{D_{k', \mathbf x, \mathbf w, \mathbf x', \mathbf w'}}_{K'}\right)}^2\Bigg)^{-1/2} & \text{(Orthogonal states)} \\
        \geq& \min_{\mathbf x, \mathbf w, \mathbf x', \mathbf w'} \norm{\left(\Pi^{(\mathbf x, \mathbf w)}\ket{D_{k, \mathbf x, \mathbf w, \mathbf x', \mathbf w'}}_K\right) \otimes \left(\Pi^{(\mathbf x', \mathbf w')}\ket{D_{k', \mathbf x, \mathbf w, \mathbf x', \mathbf w'}}_{K'}\right)} & \text{(Convex combination)} \\
        \geq& \min_{\mathbf x, \mathbf w, \mathbf x', \mathbf w'} \norm{\left(\Pi^{(\mathbf x, \mathbf w)}C_K\ket{D_{k, \mathbf x, \mathbf w, \mathbf x', \mathbf w'}}_K\right) \otimes \left(\Pi^{(\mathbf x', \mathbf w')}C_{K'}\ket{D_{k', \mathbf x, \mathbf w, \mathbf x', \mathbf w'}}_{K'}\right)}\\& - \sqrt{l2^{-c+1}} &\text{(Fundamental lemma)} \\
        \geq& 1 - O(\sqrt{l2^{-c}}) & \text{(\Cref{eqn:valid-has-tail-assump})}
    \end{align*}
\end{proof}

\begin{corollary}
    From \Cref{lem:always-tail} and \Cref{lem:proj-no-close}, we have \begin{align*}
        \norm{\Pi^{\neg r} O^R \Pi^v} \leq O(\sqrt[4]{l2^{-c}}).
    \end{align*}
    \label{cor:always-tail}
\end{corollary}

We can continue with the main lemma, which essentially states that feeding an input into the Msponge and then calling $\mathsf{find\mbox{-}tail}$ on the state will nearly always return the original input and the correct head. This will allow us to show that the action of the simulator is essentially to compute $f$ on the original input, and is the technical core. Note that intermediate wires which are later uncomputed are not mentioned in this statement, to prevent clutter.

Note that the formulation of this lemma is slightly different from the formulation in \Cref{subsec:indist}, in that we do not project onto good databases. This is because we require the input state to be valid, and the projector $\Pi^\ig$ does not commute with the projector $\Pi^v$. We instead explicitly write the form of the state on which this lemma will apply, which we will then show are the kinds of states that arise in our experiment. 

\begin{lemma}
    Consider a (possibly sub-normalized) state $\ket{\gamma}$ over registers $XKK'H$, which is valid s.t. $\Pi^v \ket{\gamma} = \ket{\gamma}$, and whose databases have size at most $t$ s.t. $\Pi^t \ket{\gamma} = \ket{\gamma}$. Let $\beta = \norm{\Pi^{\neg \ig} \ket{\gamma}}$. Then we have \begin{align*}
        \norm{(O^I-O^R)\ket{\gamma}} \leq& \tilde O\left(\beta + l\sqrt{t^32^{-\min(r, c)}} + \sqrt[4]{l2^{-c}}\right)
    \end{align*}

    \label{lem:pull-out-tail}
\end{lemma}

\begin{proof}
    First observe that steps (3) and (4) of $O^R$ preserve the computational basis of the $X$ and $S$ registers. They therefore commute with steps (1) and (2), so we can move steps (1) and (2) to be at the end. Let us call this new procedure $O^{R'}$ with steps (1'), \dots, (4'). Let us use the letter $F$ to refer to the operator implementing the first step of $O^I$ (which is the same first step as $(O^{R'})$.

    Now observe that $\Pi^r$ and $\Pi^{\ig}$ are diagonal in the computational basis, meaning they commute. This implies that $\Pi^{r \cap \ig} = \Pi^r \Pi^\ig$, i.e. the projector onto recorded and good databases is the product of the projector onto good with the projector onto recorded. We can write \begin{align*}
        F \ket{\gamma} &= \Pi^{r \cap \ig} F \ket{\gamma} + \Pi^{\neg r \cup \neg \ig} F \ket{\gamma} \\
        &= \Pi^{r \cap \ig} F \ket{\gamma} + \underbrace{\Pi^{\neg r} F \ket{\gamma}}_{\ket{\xi_1}} + \underbrace{\Pi^{r \cap \neg \ig} F \ket{\gamma}}_{\ket{\xi_2}}.
    \end{align*}
    We have that $\norm{\ket{\xi_1}} \leq \tilde O(\sqrt[4]{l2^{-c}})$ from \Cref{lem:always-tail}, and $\norm{\ket{\xi_2}} \leq \tilde O(\beta + l\sqrt{t^3 2^{-\min(r, c)}})$ from \Cref{lem:bad-transition-bound}. It follows that there exists a state $\ket{\xi}$ of norm $O(\beta + l\sqrt{t^3 2^{-\min(r, c)}} + \sqrt[4]{l2^{-c}})$ s.t. \begin{align*}
        F \ket{\gamma} &= \Pi^{r \cap \ig} F \ket{\gamma} + \ket{\xi}.
    \end{align*}
    Observe that the action of steps (2', 3', 4') of $O^{R'}$ and step (2) of $O^I$ act identically on states within the $\Pi^{r \cap \ig}$ subspace; this is because such states have $\mathbf x$ as the unique recorded tail of the $z$ value appearing as the final state. In other words, $O^IF^\dagger$ and $O^RF^\dagger$ act identically on such states.
    
    From the fact that unitaries preserve distance, we then have \begin{align*}
        \norm{O^I \ket{\gamma} - O^R \ket{\gamma}} \leq& \norm{O^I F^\dagger \Pi^{r \cap \ig} F \ket{\gamma} - O^R F^\dagger \Pi^{r \cap \ig} F \ket{\gamma}} + \norm{O^I F^\dagger \Pi^{\neg r \cup \neg \ig} F \ket{\gamma} - O^R F^\dagger \Pi^{\neg r \cup \neg \ig} F \ket{\gamma}} \\
        \leq& \underbrace{\norm{\Pi^{r \cap \ig} F \ket{\gamma} - \Pi^{r \cap \ig} F \ket{\gamma}}}_{=0} + 2\norm{\ket{\xi}} \\
        \leq& \tilde O\left(\beta + \sqrt{t^32^{-\min(r, c)}} + \sqrt[4]{l2^{-c}}\right).
    \end{align*}
\end{proof}

\paragraph{Preserving the valid subspace.}
\label{subsec:usually-valid}
We will often use the fact that the state in the consistency experiment is close to valid. Here we justify that assumption. The first and easy case is interactions via compressed oracle calls $\cO$. As was shown in \Cref{cor:always-valid}, these preserve the valid subspace. The only other way in which the simulator and adversary interact with the compressed databases is via the $\mathsf{find\mbox{-}tail}$ operation. When we say a state is ``valid'', here we mean that all the compressed databases for $K, K', H, F$ are valid.

Showing that the $\mathsf{find\mbox{-}tail}$ operation essentially preserves validity on good databases will be somewhat more involved. We will often refer back to the formal definition of $\mathsf{find\mbox{-}tail}$, \Cref{def:find-tail}. 

\begin{lemma}
    Let $\ket{\psi}_{ZTSKK'}$ be a valid state be such that $\norm{\Pi^{\neg \ig}_{KK'}\ket{\psi}_{ZTSKK'} } = \beta$. Then we have \begin{align*}
        \norm{\Pi^{\neg v}_{KK'} \mathsf{fT} \ket{\psi}} \leq \tilde O(t\beta + \sqrt{t^5 2^{-\min(r, c)}})
    \end{align*}
    \label{lem:ft-valid}
\end{lemma}

\begin{proof}
    Let us say that the input vector $\mathbf x$ of a database $D_k$ is the set of $\mathbf x$ values on which $D_k$ is defined, and similarly $\mathbf x'$ for $D_{k'}$. The output vector $\mathbf y$ and $\mathbf y'$ are similarly defined by the outputs of $D_k$ and $D_{k'}$ respectively. Databases with different input vectors are clearly orthogonal. Now observe that $\Pi^{\neg v}$ and $\mathsf{fT}$ preserve the input vectors of $D_k$ and $D_{k'}$, as well as the value of $z$. We can write \begin{align*}
        \ket{\psi} =& \sum_{\mathbf x, \mathbf x', z} \alpha_{\mathbf x, z} \ket{\psi_{\mathbf x, \mathbf x', z}},
    \end{align*}
    where $\ket{\psi_{\mathbf x, \mathbf x', z}}$ is an arbitrary state with input vectors $\mathbf x, \mathbf x'$ for registers $K, K'$ and fixed value $z$ for register $Z$. From the previous observation, we have \begin{align*}
        \norm{\Pi^{\neg v}_{KK'} \mathsf{fT} \ket{\psi}} =& \sqrt{\sum_{\mathbf x, \mathbf x', z} \abs{\alpha_{\mathbf x, z}}^2 \norm{\Pi^{\neg v}_{KK'} \mathsf{fT} \ket{\psi_{\mathbf x, \mathbf x', z}}}^2}.
    \end{align*}
    Let us now suppose the following claim, which we defer the proof of. \begin{claim}
        Let $\ket{\psi_{\mathbf x, \mathbf x', z}} \in \Pi^v_{KK'}$ be as above with fixed input vectors and $z$ value, such that $\norm{\Pi^{\neg \ig}_{KK'}\ket{\psi_{\mathbf x, \mathbf x', z}}} = \beta$. Then we have \begin{align*}
            \norm{\Pi^{\neg v}_{KK'} \mathsf{fT} \psi_{\mathbf x, \mathbf x', z}} \leq \tilde O(t\beta + \sqrt{t^5 2^{-\min(r, c)}}).
        \end{align*}
        \label{clm:ft-valid-special}
    \end{claim}

    Assuming such a claim, we can write  \begin{align*}
        \norm{\Pi^{\neg v}_{KK'} \mathsf{fT} \ket{\psi}} =& \sqrt{\sum_{\mathbf x, \mathbf x', z} \abs{\alpha_{\mathbf x, z}}^2 \norm{\Pi^{\neg v}_{KK'} \mathsf{fT} \ket{\psi_{\mathbf x, \mathbf x', z}}}^2} \\
        \leq& \tilde O\left(\sqrt{\sum_{\mathbf x, \mathbf x', z} \abs{\alpha_{\mathbf x, z}}^2 (t \norm{\Pi^{\neg \ig}_{KK'} \ket{\psi_{\mathbf x, \mathbf x', z}}})^2} + \sqrt{\sum_{\mathbf x, \mathbf x', z}\abs{\alpha_{\mathbf x, z}}^2 t^5 2^{-\min(r, c)}}\right) \\
        \leq& \tilde O\left(t \norm{\Pi^{\neg \ig}_{KK'} \ket{\psi}} + \sqrt{t^5 2^{-\min(r, c)}}\right).
    \end{align*}

    It simply remains to prove the claim. \end{proof}
\begin{proof}[Proof of \Cref{clm:ft-valid-special}]
    Consider a state of the form \begin{align}
        \ket{\psi_{\mathbf x, \mathbf x', z}} =& \sum_{\mathbf y, \mathbf y', tl, s} \alpha_{\mathbf y, \mathbf y', tl, s} \ket{D_k[\mathbf x \rightarrow \mathbf y]}_K \ket{D_{k'}[\mathbf x' \rightarrow \mathbf y']}_{K'} \ket{z}_Z \ket{tl}_T \ket{s}_S, \label{eqn:sform-tail-valid}
    \end{align}
    where $\mathbf x,\mathbf x',$ and $z$ are fixed strings. For ease of notation, for most of this proof we will drop the subscripts on $\ket{\psi}$. Recalling that database validity means every output has zero overlap with the uniform superposition, we can write the validity condition in this notation as \begin{align*}
        \forall tl, s, \mathbf x, \mathbf x', \mathbf y, \mathbf y': \sum_{y_i} \alpha_{\mathbf y|_{y_i}, \mathbf y', tl, s} &= 0, \\
        \sum_{y_i'} \alpha_{\mathbf y, \mathbf y'|_{y_i'}, tl, s} &= 0,
    \end{align*}
    where the notation $\mathbf y |_{y_i}$ means replacing the $i$-th index of $\mathbf y$ with $y_i$. Now let us define $\Pi^{\neg v, x}$ as the projector onto span of databases that are invalid on input $x$. For a general database, this projector can be written as \begin{align*}
        \Pi^{\neg v, x} = \sum_{D, x \not\in D} \frac{1}{N} \sum_{u, v} \ket{D[x\rightarrow u]}\bra{D[x\rightarrow v]}.
    \end{align*}
    In our case, we will use $\Pi^{\neg v, x}_K$ to denote such a projector on the $K$ register and $\Pi^{\neg v, x'}_{K'}$ to denote such a projector on the $K'$ register, with the subscripts sometimes dropped.
    Observe that in general \begin{align*}
        \norm{\Pi^{\neg v}_{KK'} \ket{\phi}} \leq& \sum_{x} \norm{\Pi^{\neg v, x}_{K} \ket{\phi}} + \sum_{x'} \norm{\Pi^{\neg v, x'}_{K'} \ket{\phi}}.
    \end{align*}
    For a state of the form in \Cref{eqn:sform-tail-valid}, at most $t$ of these terms will be non-zero, corresponding to the terms in the input vectors. Let us analyze the norm of one such term, say corresponding to $\Pi^{\neg v, x_1}_K$. Note that this operation (as well as $\mathsf{fT}$) preserves all other input-output points, except the value of $y_1$. We will for convenience simply use $D_k$ and $D_{k'}$ to refer to the databases, and $x$ and $y$ to refer to the selected input-output pair. We can write 
    \begin{align*}
        \ket{\psi^g} =& \sum_{\mathbf y, \mathbf y', tl, s; [\mathbf x \rightarrow (\bot, y_2, \dots)]_K, [\mathbf x' \rightarrow \mathbf y']_{K'} \text{ is good}} \alpha_{\mathbf y, \mathbf y', tl, s} \ket{D_k[\mathbf x \rightarrow \mathbf y]}_K \ket{D_{k'}[\mathbf x' \rightarrow \mathbf y']}_{K'} \ket{z}_Z \ket{tl}_T \ket{s}_S \\
        \ket{\psi^b} =& \sum_{\mathbf y, \mathbf y', tl, s; [\mathbf x \rightarrow (\bot, y_2, \dots)]_K, [\mathbf x' \rightarrow \mathbf y']_{K'} \text{ is bad}} \alpha_{\mathbf y, \mathbf y', tl, s} \ket{D_k[\mathbf x \rightarrow \mathbf y]}_K \ket{D_{k'}[\mathbf x' \rightarrow \mathbf y']}_{K'} \ket{z}_Z \ket{tl}_T \ket{s}_S,
    \end{align*}
    or in words $\ket{\psi^g}$ are the terms in which the database with $x$ assigned to $\bot$ is good, and $\ket{\psi^b}$ are the terms in which the database with $x$ assigned to $\bot$ is bad. Observe that we have \begin{align*}
        \norm{\Pi^{\neg v, x} \mathsf{fT} \ket{\psi}} =& \norm{\Pi^{\neg v, x} \mathsf{fT} (\ket{\psi^g} + \ket{\psi^b})} \\
        \leq& \norm{\Pi^{\neg v, x} \mathsf{fT} \ket{\psi^g}} + \norm{\Pi^{\neg v, x} \mathsf{fT} \ket{\psi^b}} & \text{(Triangle inequality)}\\
        \leq& \norm{\Pi^{\neg v, x} \mathsf{fT} \ket{\psi^g}} + \beta &\text{(Monotonicity of bad)}
    \end{align*}
    
    Note that the last inequality follows because adding a new input-output pair cannot convert a bad database into a good one; hence every term with non-zero amplitude in $\ket{\psi^b}$ is bad. It therefore remains to bound the first term. Observe that $\ket{\psi^g}$ is valid on input $x$, meaning $\Pi^{\neg v, x} \ket{\psi^g} = 0$. To see this, we can write the amplitudes \begin{align*}
        \ket{\psi^g} =& \sum_{\mathbf y, \mathbf y', tl, s} \beta_{\mathbf y, \mathbf y', tl, s} \ket{D_k[\mathbf x \rightarrow \mathbf y]}_K \ket{D_{k'}[\mathbf x' \rightarrow \mathbf y']}_{K'} \ket{z}_Z \ket{tl}_T \ket{s}_S \\
        \beta_{\mathbf y, \mathbf y', tl, s} \coloneqq& \begin{cases}
            \alpha_{\mathbf y, \mathbf y', tl, s} & \text{$([\mathbf x \rightarrow (\bot, y_2, \dots)]_K, [\mathbf x' \rightarrow \mathbf y']_{K'}$ is good)} \\
            0 & \text{($[\mathbf x \rightarrow (\bot, y_2, \dots)]_K, [\mathbf x' \rightarrow \mathbf y']_{K'}$ is bad)}
        \end{cases}
    \end{align*}
    and we can write the sums \begin{align*}
        \sum_{y_1} \beta_{\mathbf y|_{y_1}, \mathbf y', tl, s} =& \begin{cases}
            \sum_{y_1} \alpha_{\mathbf y|_{y_1}, \mathbf y', tl, s} & \text{($[\mathbf x \rightarrow (\bot, y_2, \dots)]_K, [\mathbf x' \rightarrow \mathbf y']_{K'}$ is good)} \\
            0 & \text{($[\mathbf x \rightarrow (\bot, y_2, \dots)]_K, [\mathbf x' \rightarrow \mathbf y']_{K'}$ is bad)}
        \end{cases} \\
        =& 0.
    \end{align*}
    
    Continuing on, denote by $t_y$ the output (first) tail on $D_k[x \rightarrow y], D_{k'}$. Let $S$ denote the possible values of $y$ such that the tail of $z$ will be different on $D_k[x \rightarrow y], D_{k'}$ than on $D_k[x \rightarrow \bot], D_{k'}$. Note that by \Cref{lem:tail-insensitive}, we have $|S|=O(t^3(n + 2^{c-r}))$. Finally, observe that $\Pi^{\neg v, x}$ annihilates databases not defined on $x$, so we may drop such terms. We now split into two cases. \begin{enumerate}[label=(\arabic*)]
        \item Suppose that $z$ has a tail $t_\bot$ in $D_k[x \rightarrow \bot], D_{k'}$. It follows that, for any $y$, $z$ will have a tail in $D_k[x \rightarrow y], D_{k'}$. We have \begin{align*}
            & \norm{\Pi^{\neg v, x_1}_K \mathsf{fT} \ket{\psi}} \\=& \norm{\Pi^{\neg v, x_1} \sum_{y, tl, s} \alpha_{y, tl, s} \ket{D_k[x \rightarrow y]}_K \ket{D_{k'}}_{K'} \ket{tl \oplus t_y}_T \ket{s\oplus 1}_S} \\
            =& \norm{\sum_{y, u, tl, s} \alpha_{y, tl, s}2^{-c} \ket{D_k[x \rightarrow u]}_K \ket{D_{k'}}_{K'} \ket{tl \oplus t_y}_T \ket{s\oplus 1}_S} \\
            =& \norm{\sum_{u, tl, s}  2^{-c}\left(\sum_{y \not\in S} \alpha_{y, tl\oplus t_{\bot}, s \oplus 1} + \sum_{y \in S} \alpha_{y, tl \oplus t_y, s\oplus 1}\right) \ket{u}\ket{tl}\ket{s}} \\
            \leq& \norm{\sum_{u, tl, s}  2^{-c}\left(-\sum_{y \in S} \alpha_{y, tl\oplus t_\bot, s \oplus 1} + \sum_{y \in S} \alpha_{y, tl \oplus t_y, s\oplus 1}\right) \ket{u}\ket{tl}\ket{s}} & \text{(Validity of $\ket{\psi}$)} \\
            \leq& \norm{\sum_{u, tl, s}  2^{-c}\sqrt{2|S|\left(\sum_{y \in S} |\alpha_{y, tl\oplus t_{\bot}, s \oplus 1}|^2 + \sum_{y \in S} |\alpha_{y, tl \oplus t_y, s\oplus 1}|^2\right)} \ket{u}\ket{tl}\ket{s}} & \text{($L_2$ bound on $L_1$)} \\
            =& 2^{-c}\sqrt{2|S|\sum_{y \in S, u, tl, s} \left(|\alpha_{y, tl\oplus t_{\bot}, s \oplus 1}|^2 + |\alpha_{y, tl \oplus t_y, s\oplus 1}|^2\right)} \\
            \leq& 2^{-c} \sqrt{2|S|\sum_{u} 2} & \text{(Normalization of $\ket{\psi}$)} \\
            \leq& \sqrt{4|S|2^{-c}} & \text{(Cauchy-Schwarz)} \\
            \leq& \tilde O(\sqrt{t^32^{-\min(r,c)}}).
        \end{align*}
        \item Suppose that $z$ has no tail in $D_k[x \rightarrow \bot], D_{k'}$. It follows by definition of $S$ that $z$ has a tail only when $y \in S$. We write \begin{align*}
            &\norm{\Pi^{\neg v, x_1}_K \mathsf{fT} \ket{\psi}} \\=& \norm{\Pi^{\neg v, x_1} \sum_{y, tl, s} \alpha_{y, tl, s} \ket{D_k[x \rightarrow y]}_K \ket{D_{k'}}_{K'} \ket{tl \oplus t_y}_T \ket{s\oplus \mathbb I[y \in S]}_S} \\
            =& \norm{\sum_{y, u, tl, s} \alpha_{y, tl, s}2^{-c} \ket{D_k[x \rightarrow u]}_K \ket{D_{k'}}_{K'} \ket{tl \oplus t_y}_T \ket{s\oplus \mathbb I[y \in S]}_S} \\
            =& \norm{\sum_{u, tl, s}  2^{-c}\left(\sum_{y \not\in S} \alpha_{y, tl, s \oplus 1} + \sum_{y \in S} \alpha_{y, tl \oplus t_y, s}\right) \ket{u}\ket{tl}\ket{s}} \\
            \leq& \norm{\sum_{u, tl, s}  2^{-c}\left(-\sum_{y \in S} \alpha_{y, tl, s \oplus 1} + \sum_{y \in S} \alpha_{y, tl \oplus t_y, s}\right) \ket{u}\ket{tl}\ket{s}} & \text{(Validity of $\ket{\psi}$)} \\
            \leq& \norm{\sum_{u, tl, s}  2^{-c}\sqrt{2|S|\left(\sum_{y \in S} |\alpha_{y, tl, s \oplus 1}|^2 + \sum_{y \in S} |\alpha_{y, tl \oplus t_y, s}|^2\right)} \ket{u}\ket{tl}\ket{s}} & \text{($L_2$ bound on $L_1$)} \\
            =& 2^{-c}\sqrt{2|S|\sum_{y \in S, u, tl, s} \left(|\alpha_{y, tl, s \oplus 1}|^2 + |\alpha_{y, tl \oplus t_y, s}|^2\right)} \\
            \leq& 2^{-c} \sqrt{2|S|\sum_{u} 1} & \text{(Normalization of $\ket{\psi}$)} \\
            \leq& \sqrt{2|S|2^{-c}} & \text{(Cauchy-Schwarz)} \\
            \leq& \tilde O(\sqrt{t^32^{-\min(r, c)}}).
        \end{align*}
    \end{enumerate}
    Putting everything together, we have \begin{align*}
        {\Pi^{\neg v, x_1}_K \mathsf{fT} \ket{\psi_{\mathbf x, \mathbf x', z}}} \leq& \tilde O(\beta + \sqrt{t^3 2^{-\min(r, c)}}).
    \end{align*}
    A similar argument implies the same bound for any $x_i \in \mathbf x$ and $x_i' \in \mathbf x'$, so we have \begin{align*}
        \norm{\Pi^{\neg v}_{K, K'} \mathsf{fT} \ket{\psi_{\mathbf x, \mathbf x', z}}} \leq& \sum_{x_i \in \mathbf x} \norm{\Pi^{\neg v, x_i}_K \mathsf{fT} \ket{\psi_{\mathbf x, \mathbf x', z}}} + \sum_{x_i' \in \mathbf x'} \norm{\Pi^{\neg v, x_i'}_{K'} \mathsf{fT} \ket{\psi_{\mathbf x, \mathbf x', z}}} \\
        \leq& \tilde O(t\beta + \sqrt{t^5 2^{-\min(r, c)}}).
    \end{align*}
\end{proof}

\paragraph{Putting the pieces together.}

We are now ready to argue that answering a query using the simulated Msponge (i.e. $\MSp^{\algo S^f}$) and using the ideal functionality (i.e. $f$) are indistinguishable. We will show that this is the case for any state which has a bounded number of input output points in each database, and is close to both valid and ideal good.

\begin{lemma}
    Consider a state $\ket{\psi}$ on adversary register $A$ and database registers $K, K', H, F$ that satisfies \begin{align*}
        \norm{\Pi^{\neg t} \ket{\psi}} = 0 && \norm{\Pi^{\neg v} \ket{\psi}} = \gamma && \norm{\Pi^{\neg \ig} \ket{\psi}} = \beta.
    \end{align*}
    If we define the ideal and real states after an $l$-block query to the sponge, we may write \begin{align*}
        \ket{\psi^I} \coloneqq& \cO_{AF} A \ket{\psi} \\
        \ket{\psi^R} \coloneqq& \Sp^{\algo S^f}_{AHKK'F} A \ket{\psi},
    \end{align*}
    where here the isometry $A$ creates the workspace used in the out-of-place sponge circuit---which is untouched in the ideal $\ket{\psi^I}$, and (approximately) uncomputed in the real $\ket{\psi^R}$. Then we have \begin{align*}
        \norm{\ket{\psi^I} - \ket{\psi^R}} \leq \tilde O( \beta + \gamma + l\sqrt{t^3 2^{-\min(r, c)}} + \sqrt[4]{l2^{-c}}).
    \end{align*}
    \label{lem:indis-one-query}
\end{lemma}

\begin{proof}
    Recalling that we break the adversary state $A$ into an $X$ and $Y$ component (plus leftover storage), and using $X$ to denote a bitflip and $XOR_{A \rightarrow B}$ to denote bitwise xor from $A$ to $B$, we have \begin{align*}
        \ket{\psi^I} =& U^\dagger X_S XOR_{Hd \rightarrow T} XOR_{X \rightarrow T} \cO_{XYF} O^I \ket{\psi} \\
        \ket{\psi^R} =& U^\dagger \mathsf{fT}_{KK'TS} \algo S^h_{XYHTSF} O^R \ket{\psi}.
    \end{align*}
    Note that we are being somewhat sloppy with notation in the second equation; in fact, the calls to $\mathsf{find\mbox{-}tail}$ made by $\algo S^h$ are here being absorbed into $O^R$ and the latter $\mathsf{fT}$ call.
    
    We have \begin{align}
        \ket{\phi} =& \Pi^{v} \ket{\psi}, &&
        \norm{\ket{\psi} - \ket{\phi}} = \gamma, \label{eqn:cons-phi-near-psi}
    \end{align}
    so we will consider instead the states \begin{align*}
        \ket{\phi^I} =& U^\dagger X_S XOR_{Hd \rightarrow T} XOR_{X \rightarrow T} \cO_{XYF} O^I \ket{\phi} \\
        \ket{\phi^R} =& U^\dagger \mathsf{fT}_{KK'TS} \algo S^h_{XYHTF} XOR_{H\rightarrow Y} O^R \ket{\phi}.
    \end{align*}
    From the fact that $\ket{\phi}$ is valid, we have \begin{align}
        \norm{O^I\ket{\phi}-O^R\ket{\phi}} \leq& \tilde O\left(\beta + l\sqrt{t^32^{-\min(r, c)}} + \sqrt[4]{l2^{-c}}\right) & \text{(From \Cref{lem:pull-out-tail})}. \label{eqn:OI-to-OR-cons}
    \end{align}
    Observe that we have \begin{align}
        \algo S^h_{XYHTF} XOR_{Hd \rightarrow Y} O^I \ket{\phi} =& \cO_{XYF} O^I \ket{\phi}. \label{eqn:cons-IR-same}
    \end{align}
    To see this, recall that $O^I$ simply places the input $X$ into the tail register $T$, and sets the flag $S$ to be success. On states of this form, the simulator will implement the call to $h$ via a (compressed oracle) call to $f$, and then XOR in the recovered head. The head $Hd$ is already XORed into the $Y$ register, and so the two cancel out and the remaining operation is simply an oracle call to $f$ on input $X$ and output $Y$. Finally, we have \begin{align}
        \norm{\mathsf{fT}_{KK'TS}\cO_{XYF} O^I \ket{\phi} - X_S XOR_{Hd \rightarrow T} XOR_{X \rightarrow T} \cO_{XYF} O^I \ket{\phi}} \leq& O(\sqrt[4]{l2^{-c}}) \label{eqn:cons-always-tail}
    \end{align}
    from \Cref{lem:always-tail}.
    We can therefore write \begin{align*}
        \norm{\ket{\psi^I} - \ket{\psi^R}} \leq& \norm{\ket{\phi^I} - \ket{\psi^I}} + \norm{\ket{\phi^R} - \ket{\psi^R}} + \norm{\ket{\phi^I} - \ket{\phi^R}} & \text{(Triangle Inequality)} \\
        \leq& 2\gamma + \norm{(O^I - O^R)\ket{\phi}} + \norm{\algo S^h XOR_{Hd\rightarrow Y}O^I \ket{\phi} - \cO_{XYF} O^I \ket{\phi}} + \\& \Bigg\Vert X_S XOR_{Hd \rightarrow T} XOR_{X \rightarrow T} \cO_{XYF} O^I \ket{\phi} - \\& \mathsf{fT}_{KK'TS}\cO_{XYF} O^I \ket{\phi}\Bigg\Vert & \text{(Triangle Inequality)} \\
        \leq& \tilde O\left(\gamma + \beta + l\sqrt{t^32^{-\min(r, c)}} + \sqrt[4]{l2^{-c}}\right) & \text{(\Cref{eqn:cons-IR-same,eqn:OI-to-OR-cons,eqn:cons-phi-near-psi,eqn:cons-always-tail})}
    \end{align*}
\end{proof}

With this in place, we can argue that the simulator we describe is consistent.
        
\begin{theorem}
    The simulator described in \cref{subsec:sim} is consistent. In particular, for a $q$-query distinguisher which makes queries of block length at most $l$ we have \begin{align*}
    |\Pr[\algo A^{\pi,\pi^{-1},\algo S^f, f}() = 1] - \Pr[\algo A^{\pi, \pi^{-1}, \algo S^{f}, \Sp^{\algo S^f}}() = 1]| = O\left(\sqrt{q^9 2^{-\min(r, c)}} + l\sqrt[4]{q^5 2^{-\min(r, c)}}\right)
\end{align*}
    \label{thm:sim-equiv-compute}
\end{theorem}

\begin{proof}
    We will refer to the left world as the ideal (i.e. the one with $\algo A^{\pi,\pi^{-1},\algo S^f, f}()$), and the right as real. Observe that the oracle calls to $\Sp^{\algo S^f}$, as depicted in \Cref{fig:sponge-value-comp}, use intermediate workspace initialized to a fixed $0$ state. We will consider the experiment where this workspace is maintained in a separate garbage register $G$ for the remainder of the real experiment. In the ideal experiment, we simply initialize a corresponding amount of auxiliary space in the fixed $0$ state. The isometry which corresponds to this is $A$, as in \Cref{def:iso-A}. In other words, we purify both experiments; we will show that the purified states remain close in between queries.
    
    The initial state for each experiment is the same, \begin{align*}
        \ket{\psi^0} \coloneqq \ket{\alpha^0}_A \ket{\emptyset}_H \ket{\emptyset}_K \ket{\emptyset}_{K'} \ket{\emptyset}_F \ket{\emptyset}_G,
    \end{align*}
    where $\ket{\alpha^0}$ is an arbitrary initial state of the adversary. As in the proof of \Cref{thm:sim-indist}, let us consider in fact a strengthened adversary which is time unbounded, and which holds the entire truth table of $\pi$. The only constraint we require is that $\pi$ is good, which happens with probability $1 - O(2^{-n})$. Our proof will show that for any fixed, good $\pi$, the final real and ideal states are close. We will let $l$ be an upper bound on the block length of any query to the sponge/ideal functionality.
    
    We can WLOG take our adversary to alternate queries to the $\algo S^f$ oracle and the hash function ($f$ in ideal, and $\Sp^{\algo S^f}$ in real), increasing the total number of queries by at most  $2$. We can parameterize a $q$-query adversary of this form by unitaries $W^1, \dots, W^q$ acting on the internal register $A$. The final state in the real experiment is given by \begin{align*}
        \ket{\psi^R} \coloneqq W^{q}_A \dots W^2_A \Sp^{\algo S^f}_{AHKK'FG} A_G W^1_A \algo S^f_{AHKK'F} \ket{\psi^0}.
    \end{align*}
    The final state in the ideal is given by
    \begin{align*}
        \ket{\psi^I} \coloneqq W^{q}_A \dots W^2_A \cO_{AF} A_G W^1_A \algo S^f_{AHKK'F} \ket{\psi^0}.
    \end{align*}
    We define the ``hybrid'' states $\ket{\psi_t^H}$, that begins like the ideal world but switches after the $t$-th query to answering using the procedure from the real world. We have \begin{align*}
        \ket{\psi^H_t} \coloneqq&  W^q \dots \Sp^{\algo S^f}_{AHKK'FG} A_G W^{2t+1}_A  \algo S^f_{AHKK'F} W^{2t}_A \cO_{AF} A_G W^{2t-1}_A \dots W^2_A \cO_{AF} A_G W^1_A \algo S^f_{AHKK'F} \ket{\psi^0},
    \end{align*}
    where $\ket{\psi^H_0} = \ket{\psi^I}$ and $\ket{\psi^H_{t/2}} = \ket{\psi^R}$. 
    Let us also define the intermediate states in the ideal world, \begin{align*}
        \ket{\psi^I_t} \coloneqq& \cO_{AF} A_G W^{2t-1}_A \dots W^2_A \cO_{AF} A_G W^1_A \algo S^f_{AHKK'F} \ket{\psi^0},
    \end{align*}
    in other words the ideal states right after the $t+1$-th query to the hash function.
    Note that we have $\ket{\psi^I_{q/2}} = \ket{\psi^I}$, by definition. Observe that we have \begin{align}
        \norm{\Pi^{\neg \ig} \ket{\psi^I_t}} \leq& \tilde O(\sqrt{t^5 2^{-\min(r, c)}}) & \text{(By \Cref{rem:good-trans-cap})} \label{eqn:indis-good}\\
        \norm{\Pi^{\neg v} \ket{\psi^I_t}} \leq& \tilde O(\sqrt{t^7 2^{-\min(r, c)}}) & \text{(By \Cref{lem:ft-valid})} \label{eqn:indis-valid}
    \end{align}
    The first inequality comes from the fact that the $\mathsf{find\mbox{-}tail}$ operation preserves the computational basis for $D_k, D_{k'}, D_h$, and hence preserves the good subspace; the only other interaction with these databases is via compressed oracle calls. The second inequality comes from the fact that only $O(t)$ calls to $\mathsf{find\mbox{-}tail}$ are made to reach $\ket{\psi^I_t}$, and the bad component on each call can be bounded by the first inequality. This is the only operation which does not preserve the valid subspace.

    Now observe that \begin{align*}
        \norm{\ket{\psi^H_t} - \ket{\psi^H_{t-1}}} =& \norm{\ket{\psi^I_t} - \Sp^{\algo S^f}_{AHKK'FG} A_G W^{2t+1}_A  \algo S^f_{AHKK'F} W^{2t}_A \ket{\psi^I_{t-1}}},
    \end{align*}
    following from unitary and isometry invariance of norms. If we define \begin{align*}
        \ket{\phi_t} = W^{2t+1}_A  \algo S^f_{AHKK'F} W^{2t}_A \ket{\psi^I_{t-1}},
    \end{align*}
    this can be written as \begin{align*}
        \norm{\ket{\psi^H_t} - \ket{\psi^H_{t-1}}} =& \norm{\cO_{AF} A_G \ket{\phi} - \Sp^{\algo S^f}_{AHKK'FG} A_G \ket{\phi}} \\
        \leq& \tilde O\left(\sqrt{t^7 2^{-\min(r, c)}} + l\sqrt{t^3 2^{-\min(r, c)}} + \sqrt[4]{l2^{-c}}\right) & \text{(By \Cref{lem:indis-one-query})}. \numberthis \label{eqn:indis-hybrid}
    \end{align*}
    where we observe that $\ket{\phi_t}$ satisfies \Cref{eqn:indis-good,eqn:indis-valid} for the same reason that $\ket{\psi^I_t}$ does. We can then write \begin{align*}
        \norm{\ket{\psi^R} - \ket{\psi^I}} \leq& \sum_{t=1}^{q/2} \norm{\ket{\psi^H_t} - \ket{\psi^H_{t-1}}} & \text{(Triangle inequality)} \\
        \leq& O\left(\sqrt{q^9 2^{-\min(r, c)}} + l\sqrt{q^5 2^{-\min(r, c)}} + \sqrt[4]{lq^42^{-c}}\right) & \text{(By \Cref{eqn:indis-hybrid})} \\
        \leq& O\left(\sqrt{q^9 2^{-\min(r, c)}} + l\sqrt[4]{q^5 2^{-\min(r, c)}}\right)
    \end{align*}
\end{proof}

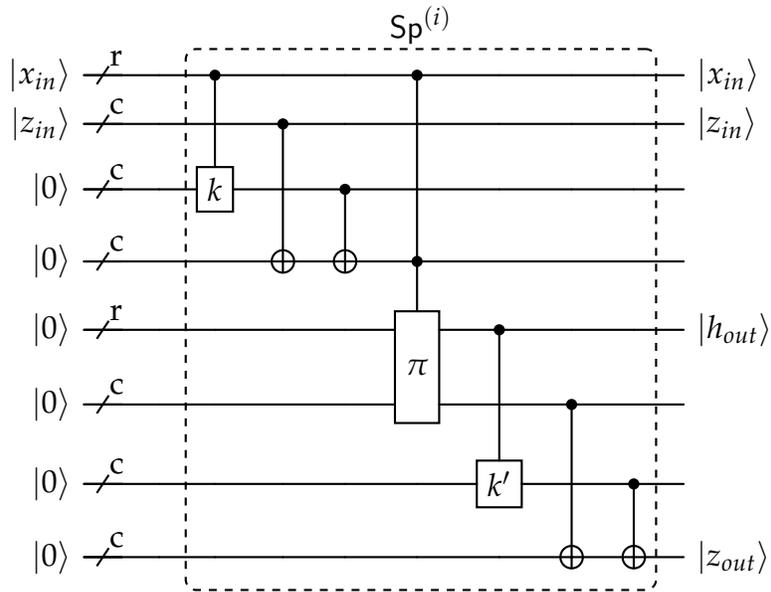
\begin{figure}[H]
    \centering
    \begin{quantikz}[transparent]
        \lstick{$\ket{x_{in}}$} & \qw\qwbundle{$r$} & & \ctrl{2}\gategroup[wires=8,steps=7,style={dashed,rounded corners,fill=white,inner xsep=0pt},background,label style={label position=above,anchor=north,yshift=0.5cm}]{$\Sp^{(i)}$} & & & \ctrl{3} & & & & \rstick{$\ket{x_{in}}$} \\
        \lstick{$\ket{z_{in}}$} & \qw\qwbundle{$c$} & & & \ctrl{2} & & & & & & \rstick{$\ket{z_{in}}$}\\
        \lstick{$\ket{0}$} & \qw\qwbundle{$c$} & & \gate{k} & & \ctrl{1} &  & & & & \\
        \lstick{$\ket{0}$} & \qw\qwbundle{$c$} & & & \targ{} & \targ{} & \ctrl{1} & & & & \\
        \lstick{$\ket{0}$} & \qw\qwbundle{$r$} & & & & & \gate[2]{\pi}& \ctrl{2} & & & \rstick{$\ket{h_{out}}$}  \\
        \lstick{$\ket{0}$} & \qw\qwbundle{$c$} & & & & & & & \ctrl{2} & & \\
        \lstick{$\ket{0}$} & \qw\qwbundle{$c$} & & & & & & \gate{k'} & & \ctrl{1} & \\
        \lstick{$\ket{0}$} & \qw\qwbundle{$c$} & & & & & & & \targ{} & \targ{} & \rstick{$\ket{z_{out}}$}
    \end{quantikz}
    \caption{An out-of-place quantum circuit for intermediate sponge rounds.}
    \label{fig:sponge-round-comp}
\end{figure}

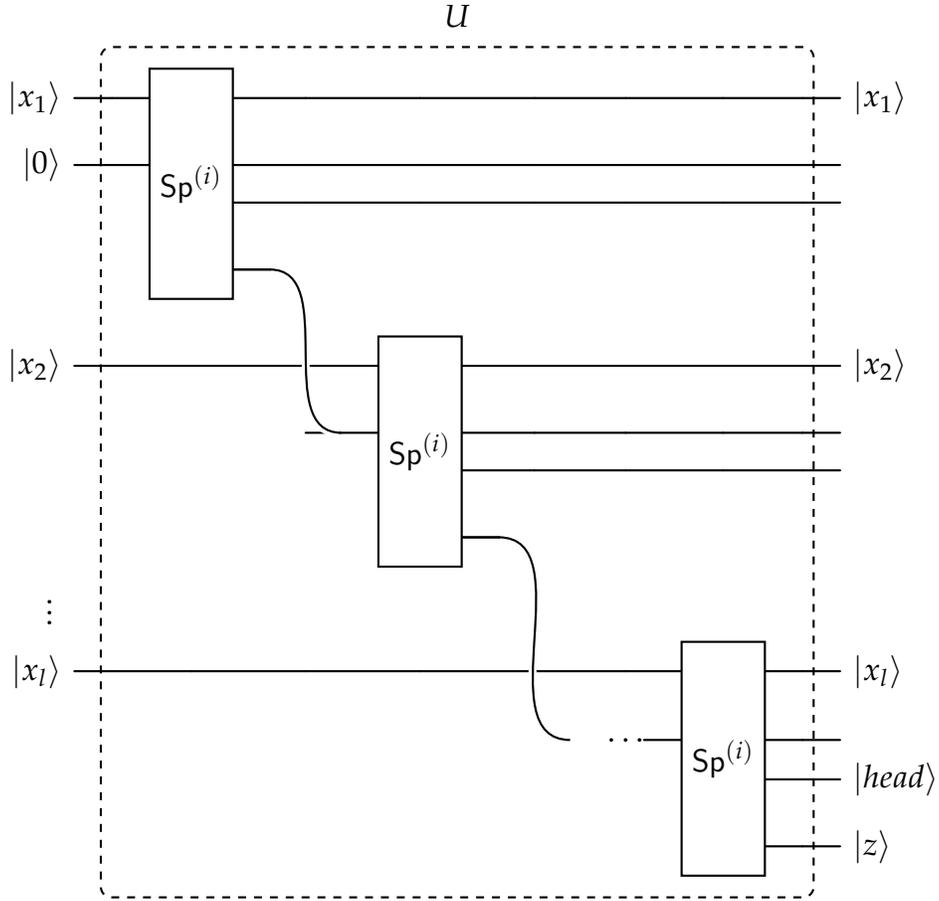
\begin{figure}[H]
    \centering
\begin{quantikz}
 \lstick{$\ket{x_1}$} &\gategroup[wires=13,steps=8,style={dashed,rounded corners,fill=white,inner xsep=0pt},background,label style={label position=above,anchor=north,yshift=0.5cm}]{$U$}&\gate[4]{\Sp^{(i)}} &&&&&&& \rstick{$\ket{x_1}$}\\
\lstick{$\ket{0}$}& &&&&&&&&\\
&\wireoverride{n} &\wireoverride{n}&&&&&&&\\
&\wireoverride{n}&\wireoverride{n} &\permute{3}&\wireoverride{n}&\wireoverride{n}&\wireoverride{n} &\wireoverride{n} &\wireoverride{n}\\
\lstick{$\ket{x_2}$} &&&&\gate[4]{\Sp^{(i)}} &&&&&\rstick{$\ket{x_2}$}\\
&\wireoverride{n}&\wireoverride{n}&\wireoverride{n}&&&&&&\\
&\wireoverride{n}&\wireoverride{n} &\wireoverride{n}&\wireoverride{n} & &&&&\\
&\wireoverride{n}&\wireoverride{n} &\wireoverride{n}&\wireoverride{n} &\permute{4}&\wireoverride{n} &\wireoverride{n}&\wireoverride{n}\\
\lstick{\myvdots}&\wireoverride{n}&\wireoverride{n}&\wireoverride{n}&\wireoverride{n}&\wireoverride{n}&\wireoverride{n} &\wireoverride{n}\\ 
\lstick{$\ket{x_l}$} &&&&&&&\gate[4]{\Sp^{(i)}} &&\rstick{$\ket{x_l}$}\\
&\wireoverride{n}&\wireoverride{n}&\wireoverride{n}&\wireoverride{n}&\wireoverride{n}&\wireoverride{n}\ldots&&&\\
&\wireoverride{n}&\wireoverride{n}&\wireoverride{n}&\wireoverride{n}&\wireoverride{n}&\wireoverride{n}&\wireoverride{n} &&\rstick{$\ket{head}$} \\
&\wireoverride{n}&\wireoverride{n}&\wireoverride{n}&\wireoverride{n}&\wireoverride{n}&\wireoverride{n}&\wireoverride{n} && \rstick{$\ket{z}$}
\end{quantikz}
    \caption{Circuit $U$, representing an out-of-place quantum circuit for computing the sponge state before the $h$ call. Because we consider the purified experiment, uncomputed workspace is not discarded.}
    \label{fig:sponge-state-comp}
\end{figure}

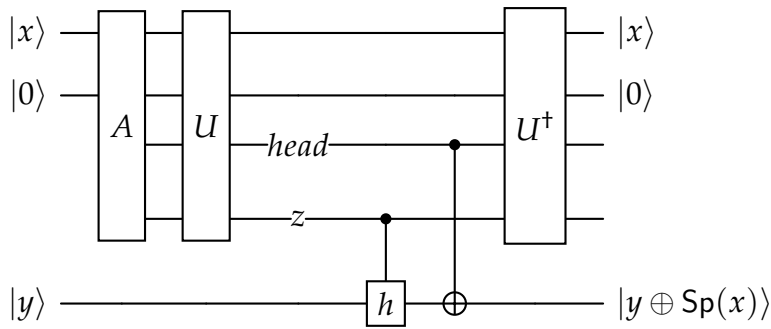
\begin{figure}[H]
    \centering
\begin{quantikz}
\lstick{$\ket{x}$} &\gate[4]{A} &\gate[4]{U} &&& &\gate[4]{U^\dagger} & \rstick{$\ket{x}$} \\
\lstick{$\ket{0}$} &&&&&&& \rstick{$\ket{0}$}\\
&\wireoverride{n}&&\push{head}&&\ctrl{2}&& \\
&\wireoverride{n}&&\push{z}&\ctrl{1}&&& \\
\lstick{$\ket{y}$} &&&&\gate{h}&\targ{}&&  \rstick{$\ket{y\oplus\Sp(x)}$}
\end{quantikz}   
    \caption{An out-of-place quantum circuit for computing the Msponge. Because we consider the purified experiment, uncomputed workspace is not discarded.}
    \label{fig:sponge-value-comp}
\end{figure}

\subsection{Main theorem}

Combining the results of \cref{subsec:indist,subsec:cons}, we obtain a bound as follows.

\begin{theorem}
    There exists an efficient simulator $\algo S$ for the sponge construction such that for all adversaries $\algo A$ making $q$ queries of block length at most $l$, they can distinguish with advantage \begin{align*}
        \norm{\Pr[\algo A^{\varphi, \varphi^{-1}, \Sp^{\varphi}}() = 1 - \algo A^{\algo S^f, f}() = 1]} = O\left(l^2 \sqrt{q^9 2^{-\min(r, c)}} + l^3 \sqrt[4]{q^5 2^{-\min(r, c)}}\right).
    \end{align*}
    \label{thm:main}
\end{theorem}

\begin{proof}
    It follows from \Cref{thm:sim-equiv-compute,thm:sim-indist,lem:cons-and-ind} that the Msponge is indifferentiable up to a bound $O\left(\sqrt{q^9 2^{-\min(r, c)}} + l \sqrt[4]{q^5 2^{-\min(r, c)}}\right)$, in the model where both the simulator and the adversary query a truly random permutation $\pi$, and the simulator provides oracles for random functions $h, k, k'$. To obtain a computationally bounded simulator, the simulator can construct a quantum secure psuedorandom permutation \cite{zhandry2016notequantumsecureprps} in place of $\pi$, and we maintain security against computationally bounded adversaries. In the computationally unbounded setting, the simulator may use something sufficiently close to a $q$-wise independent permutation, such as a $\mathsf{\poly}(q)$ round unbalanced Feistel cipher. Note also that our indifferentiability result does not directly consider superposition queries over different lengths. However, this can be handled by \Cref{lem:ortho-subspaces-norm}, noting that the query operators will be controlled on the length control register, and therefore have orthogonal ranges and domains.
    
    Plugging into \Cref{lem:tail-fix-ok}, we obtain an indifferentiability result for the proper sponge with a single round of squeezing and inputs from $(\bit^r)^*$, with an additional factor $O(l)$. In turn plugging this into \Cref{claim:sponge-approach}, we incur yet another factor $O(l)$, but now obtain an indifferentiability result of the full sponge construction with any valid $\mathsf{PAD}$ function.
\end{proof}

\subsection{On the gap in Merkle-Damg\r ard indifferentiability}
\label{subsec:MD-gap}
In the original indifferentiability proof of Merkle-Damg\r ard \cite{Zhandry2018}, the validity of compressed databases throughout the experiment is not explicitly considered. Recall that validity is the property that decompressing on any input leads to a well defined output. By \Cref{cor:always-valid} (or analogous result in \cite{Zhandry2018}), interacting with compressed databases through compressed oracle calls perfectly preserves validity. However, subroutines which inspect the database in a more direct way, such as $\mathsf{find\mbox{-}tail}$ in this work and $\mathsf{find\mbox{-}input}$ in the Merkle-Damg\r ard proof, do not perfectly preserve validity. In addition, projecting onto good databases is diagonal in the computational basis of the database register, and projecting onto valid databases is diagonal in the fourier basis; hence, these projectors also do not commute.

To see why this is a problem, consider for instance the simpler two-round domain extender of Section 5, in particular Lemma 16, of \cite{Zhandry2018}. The setting of this lemma is a hybrid in which certain bad events, namely round function collisions, have been projected out, and further the simulator has potentially made many calls to $\mathsf{find\mbox{-}input}$ while answering previous queries. The proof of Lemma 16 in \cite{Zhandry2018} implicitly assumes that, after decompressing on input $x_1$, there is an output pair in the database. This is necessary because the simulator must later recover the input to answer consistently. The prior calls to $\mathsf{find\mbox{-}input}$ and projections onto good databases however break this assumption, making the database not fully valid.

However, we believe this gap is fixable. Using an argument similar to that in \Cref{subsec:usually-valid}, one could likely bound how much $\mathsf{find\mbox{-}input}$ deviates from the valid subspace. Observe that the good projector also negligibly disturbs the state, so it cannot dramatically leave the valid subspace. Applying this fix directly would almost certainly result in a looser bound: because one cannot simultaneously project onto good and valid databases, it seems one has to pay for both the bad and the invalid component on each query. Indeed, this is what happens in our \Cref{thm:sim-equiv-compute}, and is one of the reasons why our indifferentiability bound is weaker than our direct bounds for collision and preimage resistance presented in \cref{sec:bounds-proofs}.

\newpage
\appendix

\section{Deferred proofs}
\label{sec:def-proofs}

\subsection{Indifferentiability}

Restatement of \Cref{lem:comp-V-comm-K'}:
\begin{lemma}
    The commutator between the isometry $V$ and the local compression operators on $K'$ almost commute on the good subspace: \begin{align*}
        \norm{[V_{KK'H}, L_{XK'}]\Pi^t_{KK'H}\Pi^{\rg}_{KK'}} \leq \tilde O(\left(\sqrt{t^3 2^{-\min(r, c)}}\right).
    \end{align*}
\end{lemma}

\begin{proof}[Proof of \Cref{lem:comp-V-comm-K'}]
    The key observation is essentially that $V$ is a bijective function on good databases which leaves the $k, k'$ databases invariant, and a query to $k'$ remains mostly within the set of good databases. In more detail, consider projectors $\inD_{XK'}, \ninD_{XK'}$ which sum to the identity. We split into two cases. \begin{enumerate}[label=(Case \arabic*)]
        \item $x \not\in D_{k'}$, or $\norm{[V_{KK'H}, L_{XK'}]\Pi^{\rg}_{KK'}\ninD_{XK'}}$. Note that $\L_{XK'}$ preserves the computational basis everywhere except the $x$-th register of $D_k$, and on distinct databases the images of $V$ are orthogonal. By \Cref{lem:ortho-subspaces-norm}, it suffices to fix $x, D_k, D_{k'}, D_h$ as computational basis states with good databases of size at most $t$ and $x \not\in D_{k'}$, and consider the action on the state \begin{align*}
            \ket{\psi} =& \ket{x}_X \ket{D_k}_K \ket{D_{k'}}_{K'} \ket{D_h^R}_{H}.
        \end{align*}
        
        Let us define $D_{h}^I$ (in the ideal world) as the database of input/output pairs in ${D_h^R}$ (in the real world) without a tail under $D_k, D_{k'}$, and ${D_{f}^I}$ (in the ideal world) as the database for $f$ constructed from the set of input/output pairs in ${D_h}$ with a tail. These are such that \begin{align*}
            V \ket{D_k}_K\ket{D_{k'}}_{K'}\ket{D_h^R}_H = \ket{D_k}_K\ket{D_{k'}}_{K'}\ket{D_h^I}_H \ket{D_f^I}_F.
        \end{align*}
        Recall that databases without a superscript do not change between the ideal and real world.
        We may now compute \begin{align*}
                \ket{\psi^I} =& L_{XK'} V \ket{\psi} \\
                =& \ket{A}_A \ket{x}_X \ket{D_h^I}_H \ket{D_{k}}_{K}\ket{D_f^I}_F \left( \sum_{y \in \bit^c} 2^{-c/2} \ket{D_{k'}[x\rightarrow y]}_{K'}\right)
            \end{align*}
        as well as \begin{align*}
            \ket{\psi^{R}} =& L_{XK} \ket{\psi} \\
            =& \ket{A}_A \ket{x}_X \ket{D_h^R}_H \ket{D_{k}}_{K} \left( \sum_{y \in \bit^c} 2^{-c/2} \ket{D_{k'}[x\rightarrow y]}_{K'}\right).
        \end{align*}
        
        Let $B$ be the set of images $y$ of $x$ such that assigning $x$ to $y$ will cause a bad completion. Observe that, from the analysis of \Cref{lem:few-bad-comps}, we have $|B| \leq O(t^3 n + t^3 2^{r-c})$. For any value $y\not\in B$, we have the identity \begin{align*}
            V \ket{D_k}_X \ket{D_{k'}[x\rightarrow y]}_{K'} \ket{D_h^I}_H =& \ket{D_k}_X \ket{D_{k'}[x\rightarrow y]}_{K'} \ket{D_h^R}_H \ket{D_f^R},
        \end{align*}
        because in such cases no new values $z \in D_h$ will have a tail under the assignment $[x\rightarrow y]$. It follows that \begin{align}
            V \Pi^{B\perp} \ket{\psi^R} = \Pi^{B\perp} \ket{\psi^I}. \label{eqn:k'-xnin-nobad}
        \end{align}
        We can write \begin{align*}
            \norm{\ket{\psi^R} - \Pi^{B\perp}\ket{\psi^R}} =& \norm{\ket{A}_A \ket{x}_X \ket{D_h^R}_H \ket{D_{k}}_{K} \left( \sum_{y \in B} 2^{-c/2} \ket{D_{k'}[x\rightarrow y]}_{K'}\right)} \\
            =& \sqrt{|B| 2^{-c}} \numberthis \label{eqn:k'-xnin-real} \\
            \leq& \tilde O(\sqrt{t^3 2^{-\min(r, c)}}), \\
            \norm{\ket{\psi^I} - \Pi^{B\perp}\ket{\psi^I}} =& \norm{\ket{A}_A \ket{x}_X \ket{D_h^I}_H \ket{D_{k}}_{K} \ket{D_f^I}_F \left( \sum_{y \in B} 2^{-c/2} \ket{D_{k'}[x\rightarrow y]}_{K'}\right)} \\
            =& \sqrt{|B| 2^{-c}} \\
            \leq& \tilde O(\sqrt{t^3 2^{-\min(r, c)}}). \numberthis \label{eqn:k'-xnin-ideal}
        \end{align*}
        Putting everything together, we have \begin{align*}
            \norm{\ket{\psi^I} - V \ket{\psi^{R}}} \leq& \norm{\ket{\psi^I} - \Pi^{B\perp} \ket{\psi^{I}}} + \norm{\Pi^{B\perp} \ket{\psi^I} - V\Pi^{B\perp} \ket{\psi^R}}\\& + \norm{V\ket{\psi^R} - V \Pi^{B \perp}\ket{\psi^{R}}} & \text{(Triangle inequality)} \\
            =& \norm{\ket{\psi^I} - \Pi^{B\perp} \ket{\psi^{I}}} + \norm{V\ket{\psi^R} - V \Pi^{B \perp}\ket{\psi^{R}}} & \text{(\Cref{eqn:k'-xnin-nobad})} \\ 
            =& \norm{\ket{\psi^I} - \Pi^{B\perp} \ket{\psi^{I}}} + \norm{\ket{\psi^R} - \Pi^{B \perp}\ket{\psi^{R}}} & \text{($V$ an isometry)} \\
            \leq& \tilde O(\sqrt{t^3 2^{-\min(r, c)}}) & \text{(\Cref{eqn:k'-xnin-real,eqn:k'-xnin-ideal})}
        \end{align*}
        \item $x \in D_{k'}$, or $\norm{[V_{KK'H}, L_{XK'}]\Pi^{\rg}_{KK'}\inD_{XK'}}$. Once again, note that $\L_{XK'}$ preserves the computational basis everywhere except the $x$-th register of $D_k$, and on distinct databases the images of $V$ are orthogonal. By \Cref{lem:ortho-subspaces-norm}, it suffices to fix $x, D_k, D_{k'}, D_h$ as computational basis states with good databases of size at most $t$ and where $x \not\in D_{k'}$, and consider the action on a state of the form \begin{align*}
            \ket{\psi} =& \sum_{z \text{ s.t. } D_k, D_{k'}[x\rightarrow z] \text{ is good}} \alpha_z \ket{x}_X \ket{D_k}_K \ket{D_{k'}[x\rightarrow z]}_{K'} \ket{D_h}_H.
        \end{align*}
        Let us similarly define ${D_{h}^{I,y}}$ (in the ideal world) as the database of input/output pairs in ${D_h^R}$ (in the real world) without a tail under $D_k, D_{k'}[x\rightarrow y]$, and ${D_{f}^{I, y}}$ (in the ideal world) as the database for $f$ constructed from the set of input/output pairs in ${D_h}$ with a tail in $D_k, D_{k'}[x\rightarrow y]$. These are such that \begin{align*}
            V \ket{D_k}_K\ket{D_{k'}[x\rightarrow y]}_{K'}\ket{D_h^R}_H = \ket{D_k}_K\ket{D_{k'}[x\rightarrow y]}_{K'}\ket{D_h^{I, y}}_H \ket{D_f^{I, y}}_F.
        \end{align*}
        Recall that databases without a superscript do not change between the ideal and real world. We may now compute \begin{align*}
                \ket{\psi^I} =& L_{XK'} V \ket{\psi} \\
                =& \sum_{z \text{ s.t. } D_k, D_{k'}[x\rightarrow z] \text{ is good}} \alpha_z \ket{A}_A \ket{x}_X \ket{D_h^{I, z}}_H \ket{D_{k}}_{K}\ket{D_f^{I, z}}_F \otimes \\
                &\left(\ket{D_{k'}[x\rightarrow z]}_{K'} - 2^{-c} \sum_{u \in \bit^c} \ket{D_{k'}[x\rightarrow u]}_{K'} + 2^{-c/2} \ket{D_{k'}}_{K'}\right)
            \end{align*}
        as well as \begin{align*}
            \ket{\psi^{R}} =& L_{XK'} \ket{\psi} \\
                =& \sum_{z \text{ s.t. } D_k, D_{k'}[x\rightarrow z] \text{ is good}} \alpha_z \ket{A}_A \ket{x}_X \ket{D_h^R}_H \ket{D_{k}}_{K}\otimes \\
                &\left(\ket{D_{k'}[x\rightarrow z]}_{K'} - 2^{-c} \sum_{u \in \bit^c} \ket{D_{k'}[x\rightarrow u]}_{K'} + 2^{-c/2} \ket{D_{k'}}_{K'}\right)
        \end{align*}
        Let $B$ be the set of images $y$ of $x$ such that assigning $x$ to $y$ in $D_{k'}$ will cause a bad completion. Observe that, from the analysis of \Cref{lem:few-bad-comps}, we have $|B| \leq O(t^3 n + t^3 2^{r-c})$. For any value $y\not\in B$, we have the identity \begin{align*}
            V \ket{D_k}_X \ket{D_{k'}[x\rightarrow y]}_{K'} \ket{D_h^I}_H =& \ket{D_k}_X \ket{D_{k'}[x\rightarrow y]}_{K'} \ket{D_h^R}_H \ket{D_f^R},
        \end{align*}
        because in such cases no new state values which are in the database $D_h$ will have a tail under the assignment $[x\rightarrow y]$. Observe here that the initial state $\ket{\psi}$ may be supported on images that lead to a ``bad completion'', i.e. if $x$ is part of a tail. Let us analyze the difference \begin{align*}
            \ket{\psi^I} - V\ket{\psi^R} =& \sum_{z \text{ s.t. } D_k, D_{k'}[x\rightarrow z] \text{ is good}} \alpha_z \ket{A}_A \ket{x}_X \ket{D_h^{I, z}}_H \ket{D_{k}}_{K}\ket{D_f^{I, z}}_F \otimes \\
                &\left(\ket{D_{k'}[x\rightarrow z]}_{K'} - 2^{-c} \sum_{u \in \bit^c} \ket{D_{k'}[x\rightarrow u]}_{K'} + 2^{-c/2} \ket{D_{k'}}_{K'}\right)- \\
                &\sum_{z \text{ s.t. } D_k, D_{k'}[x\rightarrow z] \text{ is good}} \alpha_z \ket{A}_A \ket{x}_X \ket{D_{k}}_{K}\otimes \\
                &\Bigg(\ket{D_{k'}[x\rightarrow z]}_{K'} \ket{D_h^{I, z}}_H \ket{D_f^{I, z}}_F  - 2^{-c} \sum_{u \in \bit^c} \ket{D_{k'}[x\rightarrow u]}_K\ket{D_h^{I, u}}_H \ket{D_f^{I, u}}_F +\\& 2^{-c/2} \ket{D_{k'}}_{K'} \ket{D_h^{I, \bot}}_H \ket{D_f^{I, \bot}}_F\Bigg)
        \end{align*}
        Which, after collapsing terms, can be written as \begin{align*}
            \ket{\psi^I} - V\ket{\psi^R} =& \sum_{z \text{ s.t. } D_k, D_{k'}[x\rightarrow z] \text{ is good}} \alpha_z \ket{A}_A \ket{x}_X \ket{D_{k}}_{K} \otimes \\
                &\left(2^{-c} \sum_{u \in \bit^c} \ket{D_{k'}[x\rightarrow u]}_{K'}(\ket{D_h^{I, u}}_H \ket{D_f^{I, u}}_F - \ket{D_h^{I, z}}_H\ket{D_f^{I, z}}_F)\right) + \\
                &\sum_{z \text{ s.t. } D_k, D_{k'}[x\rightarrow z] \text{ is good}} \alpha_z \ket{A}_A \ket{x}_X \ket{D_{k}}_{K}\otimes \\
                &\left(\ket{D_{k'}}_{K'} 2^{-c/2}(\ket{D_h^{I, z}}_H \ket{D_f^{I, z}}_F  - \ket{D_h^{I, \bot}}_H \ket{D_f^{I, \bot}}_F)\right).
        \end{align*}
        We can then write \begin{align*}
            \norm{VL\ket{\psi} - LV \ket{\psi}} =& \norm{V\ket{\psi^R} - \ket{\psi^I}} \\
            \leq& 2\underbrace{\norm{\sum_{z \in \bit^c} \alpha_z \sum_{u \in \bit^c, (D^{I, u}_h D^{I, u}_f) \neq (D^{I, z}_h D^{I, z}_f)} 2^{-c} \ket{D_{k'}[x\rightarrow u]}}}_{T_1} +\\& \underbrace{2\norm{\sum_{z \in \bit^c, (D^{I, z}_h D^{I, z}_f) \neq (D^{I, \bot}_h D^{I, \bot}_f)} \alpha_z 2^{-c/2}}}_{T_2} & \text{(Triangle Inequality)}
        \end{align*}
        Let us focus on bounding each term individually. We begin with \begin{align*}
            T_1 =& \norm{\sum_{z \in \bit^c\quad} \sum_{u \in \bit^c, (D^{I, u}_h D^{I, u}_f) \neq (D^{I, z}_h D^{I, z}_f)} \alpha_z 2^{-c} \ket{D_{k'}[x\rightarrow u]}} \\
            \leq& 2 \underbrace{\norm{\sum_{z \in \bit^c, (D^{I, z}_h D^{I, z}_f) \neq (D^{I, \bot}_h D^{I, \bot}_f)\quad} \sum_{u \in \bit^c, (D^{I, u}_h D^{I, u}_f) \neq (D^{I, z}_h D^{I, z}_f)} \alpha_z 2^{-c} \ket{D_{k'}[x\rightarrow u]}}}_{T_{11}} + \\&2 \underbrace{\norm{\sum_{z \in \bit^c, (D^{I, z}_h D^{I, z}_f) = (D^{I, \bot}_h D^{I, \bot}_f)\quad} \sum_{u \in \bit^c, (D^{I, u}_h D^{I, u}_f) \neq (D^{I, z}_h D^{I, z}_f)} \alpha_z 2^{-c} \ket{D_{k'}[x\rightarrow u]}}}_{T_{12}} & \text{(Triangle inequality)}
        \end{align*}
        Observe that the second sum in $T_{11}$ is over at most $2^c$ terms, which gives an upper bound of $\abs{\alpha_z}2^{-c/2}$ for it's norm. The first sum in $T_{11}$ is over a set of size $O(t^3 n + t^3 2^{c-r})$ by \Cref{lem:tail-insensitive,lem:few-bad-comps}, so by the relation between $L_1$ and $L_2$ norm we have \begin{align*}
            \sum_{z \in \bit^c, (D^{I, z}_h D^{I, z}_f) \neq (D^{I, \bot}_h D^{I, \bot}_f)} \abs{\alpha_z} \leq O(\sqrt{t^3 n + t^3 2^{c-r}}). \numberthis \label{eqn:sum-alphs-k'}
        \end{align*}
        It follows that $T_{11} \leq \tilde O(\sqrt{t^3 2^{-\min(r, c)}})$.
        
        Observe that the second sum in $T_{12}$ is over a set of size $O(t^3 n + t^3 2^{c-r})$ by \Cref{lem:tail-insensitive,lem:few-bad-comps}, which gives an upper bound of $\abs{\alpha_z}\sqrt{O(t^3 + t^3 2^{c-r}} \cdot 2^{-c}$ for it's norm. The first sum in $T_{12}$ is over a set of size at most $2^{-c}$, so by the relation between $L_1$ and $L_2$ norm we have \begin{align*}
            \sum_{z \in \bit^c, (D^{I, z}_h D^{I, z}_f) = (D^{I, \bot}_h D^{I, \bot}_f)} \abs{\alpha_z} \leq 2^{c/2}.
        \end{align*}
        It follows that $T_{12} \leq \tilde O(\sqrt{t^3 2^{-\min(r, c)}})$.

        Finally, it follows from \Cref{eqn:sum-alphs-k'} that $T_2 \leq \tilde O(\sqrt{t^3 2^{-\min(r, c)}})$.
    \end{enumerate}
\end{proof}

\section{Permutation tail bounds}
\label{sec:perm-tail-bounds}

Let $X, Y$ be subsets of $[N]$, and consider choosing a random permutation $\pi \sim S_N$ acting on $[N]$. We are interested in the number of elements sent by $\pi$ from $X$ to $Y$, and potentially want to consider many different $X$ and $Y$ at the same time. The most relevant characterization is below, which depends on the lemmas which follow it. 

Intuitively, suppose we have a partition of preimages into equal sized bins, and a partition of images into equal sized buckets (potentially of a different size than the bins). \Cref{thm:most-collisions} states that the maximum number of elements sent from any given bin to any given bucket will be within a constant factor of the average, if the average is appreciable. If the average number is small (say $O(1)$ or even less than $1$), then the max will be a positive integer which scales like $n=\log N$.

\begin{theorem}
    Let $N \in \N$, $X_1\sqcup X_2 \sqcup ... \sqcup X_l$ be a partition of $[N]$ such that $|X_i|=x$ for all $i$, and similarly $Y_1\sqcup Y_2 \sqcup ... \sqcup Y_k$ be a partition of $[N]$ such that $|Y_i|=y$ for all $i$. Define $m=xy/N$ and suppose $N=2^n$. Then, over the uniform choice of $\pi \sim S_N$, we have the bound$$
        \Pr_{\pi \sim S_N}[\max_{i \in [l], j \in [k]} (|\pi(X_i) \cap Y_j|) \geq 7m+3n] \leq 2^{-n}.
    $$
    \label{thm:most-collisions}
\end{theorem}
\begin{proof} 
    Let us first consider any fixed $X =X_i$ and $Y=Y_j$. From \Cref{lem:X-pairs-uniform} and \Cref{lem:X-pairs-uniform-tail-bound}, we have the equation \begin{align*}
        \Pr_{\pi \sim S_N}[|\pi(X) \cap Y| \geq 7m + k] \leq& \exp(-3k/4) \\
        \leq& 2^{-k}.
    \end{align*}
    Let us choose $k$ to be $3n$, and union bound over all possible values of $i, j$. Note that there are at most $N=2^n$ values for each $i$ and $j$, as each element of the partition is of the same (non-empty) size. We then have  \begin{align*}
        \Pr_{\pi \sim S_N}[\max_{i \in [l], j \in [k]}(|\pi(X_i) \cap Y_j|) \geq 7m + 3n] \leq& \sum_{i \in [l], j \in [k]} \Pr_{\pi \sim S_N}[(|\pi(X_i) \cap Y_j|) \geq 7m + 3n] \\
        \leq& \sum_{i \in [l], j \in [k]} 2^{-3n} \\
        \leq& 2^{-n},
    \end{align*}
    where the first line follows from a union bound. This completes the proof.
\end{proof}

\subsection{Helper lemmas}

We first can compute the expectation by linearity of expectation.

\begin{lemma}
Let $N \in \N$ and let $X, Y \subseteq [N]$ be subsets. Then, on average over the uniform choice of $\pi \sim S_N$, the expected number of elements sent from $X$ to $Y$ by $\pi$  equals
$$
\underset{\pi \sim S_N}{\E}[|\pi(X) \cap Y|]=\frac{|X||Y|}{N} = m.
$$
    \label{lem:X-pairs-uniform}
\end{lemma}

\begin{proof} \cite{carolan24oneway}, Theorem 3.13.
\end{proof}

We now state our tail bound, which will be sufficiently tight for the case where the expected number of subset pairs is small.

\begin{lemma}
   Let $N \in \N$, $X, Y \subseteq [N]$ be subsets, and $N=2^n$. Denote $x=|X|$, $y=|Y|$, and $m=\underset{{\sigma \sim S_N}}{\mathbb{E}}[|\sigma(X) \cap Y|]$. Then, for any real number $u \geq 6m$, it holds that \begin{align*}
    \underset{\pi \sim S_N}{\Pr}\Big[|\pi(X) \cap Y|\geq m + u \Big] \leq \exp\left(-\frac{3}{4}u\right).
    \end{align*}
    \label{lem:X-pairs-uniform-tail-bound}
\end{lemma}
\begin{proof} \cite{carolan24oneway}, Theorem 3.15.
\end{proof}

\printbibliography

\end{document}